\documentclass[letterpaper,twocolumn,10pt]{article}
\usepackage{usenix-2020-09}

\usepackage{tikz}
\usepackage{amsmath}

\usepackage{filecontents}
\usepackage{graphicx}
\usepackage{amsthm}
\usepackage{amssymb}
\usepackage{thmtools}
\usepackage{color}
\usepackage{enumitem}
\usepackage{bm}

\usepackage[linesnumbered, lined, boxed]{algorithm2e}

\usepackage[normalem]{ulem}

\newtheorem{definition}{Definition}
\newtheorem{theorem}{Theorem}
\newtheorem{lemma}{Lemma}
\newtheorem{proposition}{Proposition}

\SetKwInput{KwData}{Input}
\SetKwInput{KwResult}{Output}

\newcommand{\nats}{\mathbb{N}}
\newcommand{\nnints}{\mathbb{Z}_{\ge0}}
\newcommand{\reals}{\mathbb{R}}
\newcommand{\nnreals}{\mathbb{R}_{\ge0}}

\renewcommand{\epsilon}{\varepsilon}

\newcommand{\calG}{\mathcal{G}}

\newcommand{\calR}{\mathcal{R}\hspace{0.1mm}}

\newcommand{\E}{\operatorname{\mathbb{E}}}
\newcommand{\V}{\operatorname{\mathbb{V}}}
\newcommand{\cov}{\operatorname{\text{Cov}}}
\newcommand{\bmA}{\mathbf{A}}

\newcommand{\bma}{\mathbf{a}}

\newcommand{\bmr}{\mathbf{r}}

\newcommand{\hf}{\hat{f}}

\newcommand{\hr}{\hat{r}}
\newcommand{\hw}{\hat{w}}
\newcommand{\td}{\tilde{d}}

\def\AlgOne{\textsf{ARRFull}$_{\triangle}$}
\def\AlgTwo{\textsf{ARROneNS}$_{\triangle}$}
\def\AlgThree{\textsf{ARRTwoNS}$_{\triangle}$}
\def\AlgSec{\textsf{RRFull}$_{\triangle}(d_{max})$}
\newcommand{\GPlus}{\textsf{Gplus}}
\newcommand{\IMDB}{\textsf{IMDB}}

\newcommand{\Lap}{\textrm{Lap}}
\newcommand{\CostDL}{\mathrm{Cost}_{DL}}
\newcommand{\CostUL}{\mathrm{Cost}_{UL}}

\usepackage{hyperref}
\newcommand{\footremember}[2]{%
    \thanks{#2}
    \newcounter{#1}
    \setcounter{#1}{\value{footnote}}%
}
\newcommand{\footrecall}[1]{%
    \footnotemark[\value{#1}]%
}

\hypersetup{
  colorlinks,
  linkcolor={black},
  citecolor={red!70!black},
  urlcolor={blue!70!black}
}

\newif\ifconferenceon\conferenceonfalse
\ifconferenceon
\newcommand{\conference}[1]{#1}
\newcommand{\arxiv}[1]{}
\else
\newcommand{\conference}[1]{}
\newcommand{\arxiv}[1]{#1}
\usepackage{balance}
\fi

\begin{document}

\title{\Large \bf Communication-Efficient Triangle Counting under Local Differential Privacy}

\author{
{\rm Jacob Imola}\footremember{contributions}{The first and second authors made equal contribution.}\\
UC San Diego
\and
{\rm Takao Murakami}\footrecall{contributions}\\
AIST
\and
{\rm Kamalika Chaudhuri}\\
UC San Diego
} 

\maketitle

\begin{abstract}

Triangle counting in networks under LDP (Local Differential Privacy) is a fundamental task for analyzing connection patterns or calculating a 
clustering coefficient while strongly protecting sensitive friendships from a central server. In particular, a recent study proposes an algorithm
for this task that uses two rounds of interaction between users and the server to significantly reduce estimation error. However,
this algorithm suffers from a prohibitively 
high 
communication cost due to a large noisy graph each user needs to download.

In this work, we propose triangle counting algorithms under LDP with a small estimation error and communication cost.
We first
propose two-rounds algorithms
consisting of edge sampling and carefully selecting edges each user downloads
so that the estimation error is small.
Then we propose
a double clipping technique,
which clips the number of edges and then the number of noisy triangles,
to significantly reduce the sensitivity of each user's query.
Through comprehensive evaluation, we 
show that our algorithms dramatically reduce the communication cost
of the existing algorithm, e.g., from
6 hours to 8 seconds or less at a 20 Mbps download rate, while keeping
a small estimation error.
\end{abstract}

\section{Introduction}
\label{sec:intro}
Counting subgraphs (e.g., triangles, stars, cycles) is
one of the most basic tasks 
for analyzing connection patterns
in
various graph data, e.g., social,
communication, and collaboration networks.
For example,
a triangle is given by a set of three nodes with three edges, whereas a $k$-star is given by a central node connected to $k$ other nodes.
These subgraphs
play a crucial role in calculating
a \textit{clustering coefficient} ($=\frac{3 \times \text{\#triangles}}{\text{\#2-stars}}$) (see Figure~\ref{fig:triangles_stars}). 
The clustering coefficient 
measures the average probability that
two friends of a user will also be a friend
in a social graph \cite{Newman_PRL09}. 
Therefore, it is useful for measuring the effectiveness of friend suggestions. 
In addition, the clustering coefficient represents the degree to which users tend to cluster together. 
Thus, if it is large in some services/communities, we can effectively apply social recommendations \cite{Kolluri_CCS21} to the users. 
Triangles 
and $k$-stars 
are also useful for 
constructing
graph models
\cite{Robins_SN07,Jorgensen_SIGMOD16}; 
see also \cite{Tsourakakis_JGAA11} for other applications of triangle counting. 
However, graph data often involve sensitive data such as sensitive edges (friendships),
and they 
can be leaked from 
exact numbers of triangles and $k$-stars \cite{Imola_USENIX21}.

To analyze subgraphs while protecting user privacy, DP (Differential Privacy) \cite{DP} has been widely adopted as a privacy metric \cite{Ding_TKDE21,Imola_USENIX21,Karwa_PVLDB11,Sun_CCS19,Ye_ICDE20,Ye_TKDE21,Zhang_SIGMOD15}.
DP protects user privacy against adversaries with arbitrary background knowledge and is known as a gold standard for data privacy.
According to the underlying model, DP can be categorized into \textit{central (or global) DP} and \textit{LDP (Local DP)}.
Central DP assumes a scenario where
a central server has personal data of all users.
Although accurate analysis of subgraphs is possible under
this model \cite{Ding_TKDE21,Karwa_PVLDB11,Zhang_SIGMOD15}, there is a risk that the entire graph is leaked from the server by illegal access or internal fraud \cite{data_breach2021,CambridgeAnalytica}.
In addition, central DP cannot be applied to  \textit{decentralized social networks} 
\cite{Diaspora,Mastodon,Minds,Paul_CN14} 
where the entire graph is distributed across many servers. 
We can even consider 
\textit{fully decentralized applications} where a server does not have
any original edge, 
e.g., a mobile app that sends a noisy degree (noisy number of friends) to the server, which then estimates a degree distribution. 
Central DP cannot be used in such applications. 

In contrast, LDP assumes a scenario where each user obfuscates her personal data (friends list in 
the case of graphs) 
by herself and sends the obfuscated data to a possibly malicious server; i.e., it does not assume trusted servers.
Thus, it does not suffer from a data breach and can also be applied to the decentralized applications. 
LDP has been widely studied in tabular data where each row corresponds to a user's personal data (e.g., age, browser setting, location) 
\cite{Acharya_AISTATS19,Bassily_NIPS17,Erlingsson_CCS14,Kairouz_ICML16,Murakami_USENIX19,Wang_USENIX17} 
and also in graph data \cite{Imola_USENIX21,Qin_CCS17,Ye_ICDE20,Ye_TKDE21}.
For example, $k$-star counts can be very accurately estimated under LDP because each user can count $k$-stars of which she is a center and sends a noisy version of her $k$-star count to the server \cite{Imola_USENIX21}.

However, more complex subgraphs such as triangles are much harder to count under LDP because each user
cannot see
edges between other users.
For example, in Figure~\ref{fig:triangles_stars}, user $v_1$ cannot see
edges between $v_2$, $v_3$, and $v_6$ 
and therefore 
cannot count triangles involving $v_1$.
Thus,
existing algorithms \cite{Imola_USENIX21,Ye_ICDE20,Ye_TKDE21}
obfuscate each user's edges (rather than her triangle count) by
RR (Randomized Response)
\cite{Warner_JASA65} and send noisy edges to a server.
Consequently, the server suffers from a
prohibitively
large estimation error
(e.g., relative error $> 10^2$ in large graphs, 
as shown in Appendix~\ref{sec:one-round}) 
because all three edges are noisy in any noisy triangle the server sees.

\begin{figure}[t]
  \centering
  \includegraphics[width=0.85\linewidth]{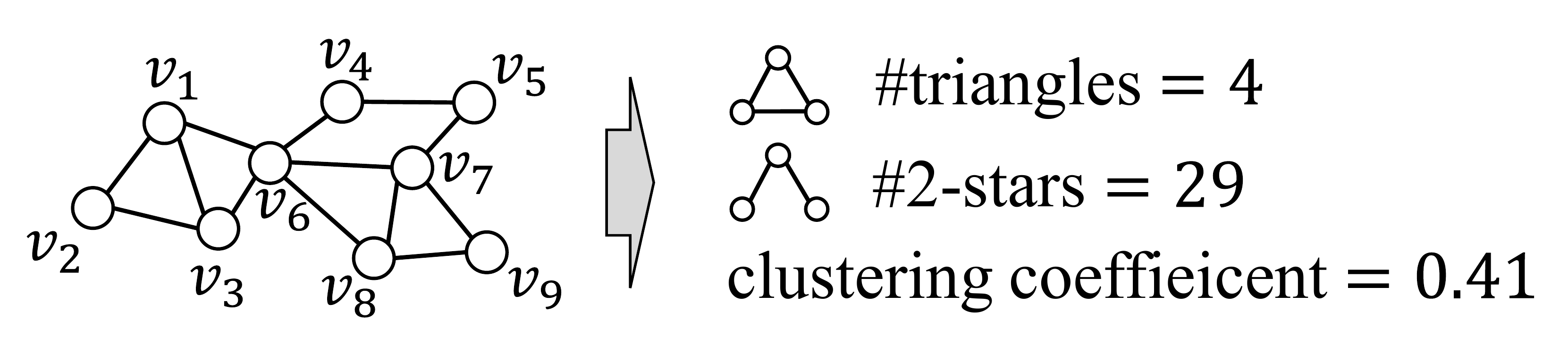}
  \vspace{-4mm}
  \caption{Triangles, $2$-stars, and clustering coefficient.}
  \label{fig:triangles_stars}
\end{figure}

A recent study \cite{Imola_USENIX21} shows that the estimation error in
locally private triangle counting
is significantly reduced by introducing an additional round of interaction between users and the server.
Specifically, if the server publishes the noisy graph (all noisy edges) sent by users at the first round, then each user can count her noisy triangles such that \textit{only one edge} is noisy (as she knows two edges connected to her).
Thus,
the algorithm in \cite{Imola_USENIX21} sends each user's noisy triangle count (with additional noise) to the server at the second round.
Then the server can accurately estimate the triangle count. 
This algorithm also requires a much smaller number of interactions 
(i.e., only two) than collaborative approaches \cite{Kairouz_FTML21,Shokri_CCS15} that generally require many interactions. 

Unfortunately, 
the algorithm in \cite{Imola_USENIX21} 
is still
impractical 
for a large-scale graph.
Specifically, the noisy graph sent by users is dense, hence extremely large for a large-scale graph, e.g.,
$500$
Gbits for a graph of
a million users.
The problem is that \textit{every user} needs to download such huge data; e.g., when the download speed is $20$ Mbps (which is a recommended speed in YouTube \cite{YouTube_speed}), every user needs about 7 hours to download the noisy graph.
Since the communication ability might be limited for some users, the algorithm in \cite{Imola_USENIX21} cannot be
used for
applications with large and diverse users.

In summary, existing triangle algorithms under LDP suffer from either a prohibitively large estimation error or a prohibitively 
high 
communication cost.
They also suffer from the same issues when calculating the clustering coefficient.

\smallskip
\noindent{\textbf{Our Contributions.}}~~We 
propose locally private triangle counting algorithms with a small estimation error and small communication cost.
Our contributions are as follows:

\begin{itemize}
    \item We propose two-rounds triangle algorithms
consisting of \textit{edge sampling} after RR and \textit{selecting edges each user downloads}.
In particular, we show that a simple extension of \cite{Imola_USENIX21} with edge sampling suffers from a large estimation error for a large or dense graph where the number of 4-cycles (such as $v_1$-$v_2$-$v_3$-$v_6$-$v_1$
in Figure~\ref{fig:triangles_stars}) is large.
To address this issue, we propose some strategies for selecting edges to download to reduce the error caused by the 4-cycles, which we call the \textit{4-cycle trick}.
\item We
show that
the algorithms with the $4$-cycle trick
still suffer from a large estimation error due to
large Laplacian noise for each user.
To significantly reduce the Laplacian noise, 
we
propose a \textit{double clipping} technique,
which clips 
a degree (the number of edges) of each user 
with LDP and then clips the number of noisy triangles. 
\item We evaluate our algorithms using two real datasets.
We show that our entire algorithms with the 4-cycle trick and double clipping
dramatically reduce the communication cost
of
\cite{Imola_USENIX21}.
For example,
for a graph with about $900000$ users,
we reduce the download cost from $400$ Gbits ($6$ hours when $20$ Mbps) to 
$160$ Mbits ($8$ seconds) or less 
while keeping the relative error much smaller than 1.
\end{itemize}
Thus, locally private triangle counting is now much more practical. 
In Appendix~\ref{sec:cluster}, we also show that we can estimate the clustering coefficient with a small estimation error and download cost. 
For example, our algorithms are useful for measuring the effectiveness of friend suggestions or social recommendations in decentralized 
social networks, 
e.g., Diaspora \cite{Diaspora}, Mastodon \cite{Mastodon}. 
Our source code
is available at 
\cite{TriangleLDP}. 

All the proofs of our privacy and utility analysis 
are given in \conference{the full version \cite{Imola_arXiv22}}\arxiv{Appendices~\ref{sec:proof_seq_comp_edge_LDP}, \ref{sec:proof_algorithms}, and \ref{sec:proof_double_clip}}.

\smallskip
\noindent{\textbf{Technical Novelty.}}~~Below we explain more about 
the technical novelty of this paper. 
Although we focus on two-rounds local algorithms in the same way as \cite{Imola_USENIX21}, we introduce several new algorithmic ideas previously unknown in the literature. 

First, our 4-cycle trick is totally new. 
Although some studies focus on 4-cycle counting \cite{Bera_STACS17,Kallaugher_PODS19,Manjunath_ESA11,McGregor_PODS20}, this work is the first to use 4-cycles to improve communication efficiency. 
Second, selective download of parts of a centrally computed quantity is also new. 
This is not limited to graphs -- even in machine learning, 
there are no such strategic download techniques previously, to our knowledge. 
Third, our utility analysis of our 
triangle 
algorithms (Theorem~\ref{thm:l2loss_algorithms}) is 
different from \cite{Imola_USENIX21} in that ours introduces subgraphs such as 4-cycles and $k$-stars. 
This leads us to our 4-cycle trick. 
Fourth, we propose two triangle algorithms that introduce the 4-cycle trick and show that the more tricky one provides the best performance because of 
its low sensitivity in DP.

Finally, 
our double clipping is new. 
Andrew \textit{et al.} \cite{Andrew_NeurIPS21} propose an adaptive clipping technique, which applies clipping twice. 
However, they focus on federated averaging, 
and their problem setting is different from our graph setting. 
In particular, they require a private quantile of the norm distribution. 
In contrast, we need only a much simpler estimate: a private degree. 
Here, we use the fact that the degree has a small sensitivity (sensitivity $=1$) in DP for edges. 
We also provide a new, reasonably tight bound on the probability that the noisy triangle count exceeds a clipping threshold (Theorem~\ref{thm:triangle_excess}). 
Thanks to the two differences, we obtain a significant communication improvement: two or three orders of magnitude.

\section{Related Work}
\label{sec:related}
\noindent{\textbf{Triangle Counting.}}~~Triangle
counting has been extensively studied in a non-private setting \cite{Bera_KDD20,Bera_PODS20,Chu_KDD11,Eden_FOCS15,Seshadhri_SDM13,Suri_WWW11,Tsourakakis_KDD09,Wu_TKDE16} 
(it is almost a sub-field in itself)
because it requires high time complexity for large graphs.

Edge sampling \cite{Bera_PODS20,Eden_FOCS15,Tsourakakis_KDD09,Wu_TKDE16} is one of the most basic techniques to improve 
scalability.
Although edge sampling is simple, it is quite effective -- it is reported in \cite{Wu_TKDE16} that edge sampling outperforms other sampling techniques such as node sampling and triangle sampling.
Based on this, we adopt edge sampling after RR\footnote{We also note that a study in \cite{Nguyen_TDP16} proposes a graph publishing algorithm in the central model that independently changes 1-cells (edges) to 0-cells (no edges) with some probability and then 
changes a fixed number of 0-cells to 1-cells \textit{without replacement}. 
However, each 0-cell is \textit{not} independently sampled in this case, and consequently, 
their proof that relies on the independence of the noise to each 0-cell is incorrect. 
In contrast, our algorithms provide DP because we apply sampling after RR, i.e., post-processing.} 
with new techniques such as the 4-cycle trick and double clipping.
Our entire algorithms significantly improve the communication cost, as well as the space and time complexity, under LDP (see Sections~\ref{sub:clip_theoretical_analysis}
and \ref{sec:experiments}).

\smallskip
\noindent{\textbf{DP on Graphs.}}~~For private graph analysis, DP has been widely adopted as a privacy metric.
Most of them adopt central (or global) DP
\cite{Day_SIGMOD16,Ding_TKDE21,Hay_ICDM09,Karwa_PVLDB11,Kasiviswanathan_TCC13,Raskhodnikova_arXiv15,Zhang_SIGMOD15}, 
which suffers from the data breach issue.

LDP on graphs has recently studied in some studies, e.g., synthetic data generation \cite{Qin_CCS17}, subgraph counting \cite{Imola_USENIX21,Sun_CCS19,Ye_ICDE20,Ye_TKDE21}.
A study in
\cite{Sun_CCS19} proposes subgraph counting algorithms in a setting where each user
allows her friends to see all her connections.
However, this setting is unsuitable for
many applications; e.g., in Facebook, a user can easily change her setting so that
her friends cannot see her connections.

Thus, we consider a model where each user can see only her friends.
In this model, some one-round algorithms \cite{Ye_ICDE20,Ye_TKDE21}
and two-rounds algorithms\cite{Imola_USENIX21} have been proposed.
However, they suffer from a prohibitively large estimation error or high communication cost, as explained in Section~\ref{sec:intro}.

Recently proposed network LDP protocols \cite{Cyffers_arXiv21} consider, instead of a central server, collecting private data with user-to-user communication protocols along a graph. 
They focus on sums, histograms, and SGD (Stochastic Gradient Descent) and do not provide subgraph counting algorithms. 
Moreover, they focus on hiding each user's private dataset rather than hiding an edge in a graph. 
Thus, their approach cannot be applied to our task of subgraph counting under LDP for edges. 
The same applies to 
another work \cite{Sabater_arXiv21} 
that improves 
the utility of an averaging query
by correlating the noise of users according to a graph.

\smallskip
\noindent{\textbf{LDP.}}~~RR~\cite{Kairouz_ICML16,Warner_JASA65} 
and 
RAPPOR~\cite{Erlingsson_CCS14} 
have been widely used 
for tabular data 
in 
LDP. 
Our work uses RR in part of our algorithm but
builds off of it significantly. One noteworthy result in this area is HR (Hadamard Response) \cite{Acharya_AISTATS19}, which is state-of-the-art for tabular data
and requires low communication. However, this result is not applied to graph
data and does not address the communication issues considered in this paper.
Specifically, applying HR to each bit in a neighbor list will result in 
$O(n^2)$ ($n$: \#users) download cost 
in the same way as 
the previous work \cite{Imola_USENIX21} that uses RR. 
Applying HR to an entire neighbor list
(which has $2^n$ possible values) will similarly result in 
$O(n \log 2^n) = O(n^2)$ download cost. 

Previous work on distribution estimation~\cite{Kairouz_ICML16,Murakami_USENIX19,Wang_USENIX17} or 
heavy hitters \cite{Bassily_NIPS17} 
addresses a different problem than ours, as they assume that every user has
i.i.d. 
(independent and identically distributed) 
samples. 
In our setting, a user's neighbor list is non-i.i.d. (as one edge is shared by two users), 
which does not
fit into their statistical framework.

\section{Preliminaries}
\label{sec:preliminaries}

\subsection{Notations}
\label{sub:notations}
We begin with basic notations. 
Let $\nats$, $\reals$, $\nnints$, and $\nnreals$ be the sets of natural numbers, real numbers, non-negative integers, and non-negative real numbers, respectively. 
For $z\in\nats$, let $[z]$ a set of natural numbers from $1$ to $z$; i.e., $[z] = \{1, 2, \ldots, z\}$. 

Let $G=(V,E)$ be an undirected graph, where $V$ is a set of nodes and $E \subseteq V \times V$ is a set of edges. 
Let $n\in\nats$ be the number of nodes in $V$. 
Let $v_i \in V$ be the $i$-th node; i.e., $V=\{v_1,\ldots,v_n\}$. 
We consider a social graph where each node in $V$ represents a user and an edge $(v_i,v_j) \in E$ represents that $v_i$ is a friend with $v_j$. 
Let $d_{max} \in \nats$ be the maximum degree of $G$. 
Let $\calG$ be a set of graphs with $n$ nodes. 
Let $f_\triangle: \calG \rightarrow \nnints$ be a triangle 
count query 
that takes $G \in \calG$ as input and outputs 
a triangle count $f_\triangle(G)$ (i.e., number of triangles) in $G$.

Let $\bmA=(a_{i,j}) \in \{0,1\}^{n \times n}$ be a symmetric adjacency matrix corresponding to $G$; i.e., $a_{i,j} = 1$ if and only if $(v_i,v_j) \in E$. 
We consider a local privacy model~\cite{Qin_CCS17,Imola_USENIX21}, where each user obfuscates her \textit{neighbor list} $\bma_i = (a_{i,1}, \ldots, a_{i,n})\in\{0,1\}^n$ (i.e., the $i$-th row of $\bmA$) using 
a \textit{local randomizer} 
$\calR_i$ with domain $\{0,1\}^n$ and sends obfuscated data $\calR_i(\bma_i)$ to a server. 
We also assume a two-rounds algorithm in which user $v_i$ downloads a message $M_i$ from the server at the second round. 

We also show the basic notations in Table~\ref{tab:notations} of Appendix~\ref{sec:notations_subgraphs}.

\subsection{Local Differential Privacy on Graphs}
\label{sub:LDP}

\noindent{\textbf{LDP on Graphs.}}~~When 
we apply LDP (Local DP) to graphs, 
we follow the direction of \textit{edge DP}~\cite{Nissim_STOC07,Raskhodnikova_Encyclopedia16} that has been developed for the central DP model. 
In 
edge DP, 
the existence of an edge
between any two users is protected; i.e., two computations, one using a graph with the
edge and one using the graph without the edge, 
are indistinguishable. 
There is also another privacy notion called \textit{node DP}~\cite{Hay_ICDM09,Zhang_USENIX20}, which hides the existence of one user along with 
all her edges. 
However, in the local model, many applications send a user ID to a server; e.g., each user sends the number of her friends along with her user ID. 
For such applications, we cannot use node DP but can use edge DP to hide her edges, i.e., friends. 
Thus, we focus on edge DP in the local model in the same way as~\cite{Imola_USENIX21,Qin_CCS17,Sun_CCS19,Ye_ICDE20,Ye_TKDE21}. 

Specifically, 
assume that user $v_i$ uses her local randomizer $\calR_i$. 
We assume that the server and other users can be 
honest-but-curious adversaries and that they can obtain all edges except for user $v_i$'s edges 
as prior knowledge. 
Then we 
use the following definition for $\calR_i$:

\begin{definition} [$\epsilon$-edge LDP~\cite{Qin_CCS17}] \label{def:edge_LDP} 
Let $\epsilon \in \nnreals$. 
  For 
  $i \in [n]$, 
  let $\calR_i$ be a 
  local randomizer 
  of user $v_i$ that 
  takes $\bma_i$ as input. We say $\calR_i$ provides
  \emph{$\epsilon$-edge LDP} 
  if for any two neighbor lists 
  $\bma_i, \bma'_i \in \{0,1\}^n$ 
  that differ in one bit and any 
  $s \in \mathrm{Range}(\calR_i)$, 
\begin{align}
\Pr[\calR_i(\bma_i) = s] \leq e^\epsilon \Pr[\calR_i(\bma'_i) = s].
\label{eq:edge_LDP}
\end{align}
\end{definition}
For example, a local randomizer $\calR_i$ that applies Warner's RR (Randomized Response) \cite{Warner_JASA65}, which flips 0/1 with probability $\frac{1}{e^\epsilon + 1}$, to each bit of $\bma_i$ 
provides $\epsilon$-edge LDP. 

The parameter $\epsilon$ is called the privacy budget. 
When 
$\epsilon$ is small (e.g., $\epsilon \leq 1$~\cite{DP_Li}), each bit is strongly protected by edge LDP. 
Edge LDP can also be used to hide \textit{multiple bits} -- 
by group privacy~\cite{DP}, two neighbor lists $\bma_i, \bma'_i \in \{0,1\}^n$ that differ in $k \in \nats$ bits are indistinguishable up to the factor $k\epsilon$. 

Edge LDP is useful for protecting a neighbor list $\bma_i$ of each user $v_i$. 
For example, 
a user in Facebook can change her setting so that anyone (except for the central server) cannot see her friend list $\bma_i$. 
Edge LDP hides $\bma_i$ even from the server. 

As with regular LDP, the guarantee of edge LDP does not break 
even if 
the server or other users act maliciously. 
However, 
adding or removing an edge
affects the neighbor list of two users. 
This means that each user needs to trust 
her friend 
to not reveal 
an edge between them. 
This also applies to Facebook -- even if $v_i$ keeps $\bma_i$ secret, her edge with $v_j$ can be disclosed if $v_j$ reveals $\bma_j$. 
To protect each edge during the whole process, 
we use 
another 
privacy notion 
called relationship DP~\cite{Imola_USENIX21}:

\begin{definition} [$\epsilon$-relationship DP~\cite{Imola_USENIX21}] 
\label{def:entire_edge_LDP} 
  Let $\epsilon \in \nnreals$. For 
  $i \in [n]$, 
  let $\calR_i$ be a 
  local randomizer of user $v_i$ that 
  takes $\bma_i$ as input. We say 
  $(\calR_1, \ldots, \calR_n)$ provides 
\emph{$\epsilon$-relationship DP} 
if for any two neighboring graphs $G, G' \in \calG$ that differ in one edge and 
  any $(s_1, \ldots, s_n) \in \mathrm{Range}(\calR_1) \times \ldots \times \mathrm{Range}(\calR_n)$, 
\begin{align}
  &\Pr[(\calR_1(\bma_1), \ldots, \calR_n(\bma_n)) = (s_1, \ldots, s_n)] \nonumber\\
  &\leq e^\epsilon \Pr[(\calR_1(\bma'_1), \ldots, \calR_n(\bma'_n)) = (s_1,
  \ldots, s_n)],
\label{eq:entire_edge_LDP}
\end{align}
  where $\bma_i$ (resp. $\bma_i'$) $\in \{0,1\}^n$ is the $i$-th row of the
  adjacency matrix of graph $G$ (resp. $G'$).
\end{definition}
If 
users $v_i$ and $v_j$ follow the protocol, 
\eqref{eq:entire_edge_LDP} holds for graphs $G,G'$ that differ 
in $(v_i, v_j)$. 
Thus, 
relationship DP applies to
all edges of a user 
whose 
neighbors are 
trustworthy.

While users need to trust other 
friends 
to maintain a relationship
DP guarantee, only one edge per user is at risk for each malicious 
friend 
that
does not follow the protocol. This is
because only one edge can exist between two users.  
Thus, although the trust assumption in relationship DP is stronger than that of LDP, it is much weaker than that of central DP in which all edges can be revealed by the server.

It is possible to use
a tuple of local randomizers with edge LDP to obtain a relationship DP guarantee:
\begin{proposition} [Edge LDP and relationship DP~\cite{Imola_USENIX21}] 
\label{prop:edge_LDP_entire_edge_LDP} 
  If 
  each 
  of local randomizers $\calR_1, \ldots, \calR_n$ 
  provides 
  $\epsilon$-edge LDP, then $(\calR_1, \ldots, \calR_n)$ provides 
  $2\epsilon$-relationship DP. 
  Additionally, if each $\calR_i$ uses only bits $a_{i,1}, \ldots, a_{i,i-1}$ for users with smaller IDs (i.e., only the lower triangular part of $\bmA$), then $(\calR_1, \ldots, \calR_n)$ provides 
  $\epsilon$-relationship DP. 
\end{proposition}
The doubling factor in $\epsilon$ comes from the fact that 
\eqref{eq:entire_edge_LDP} applies to an entire edge, whereas
\eqref{eq:edge_LDP} applies to just one neighbor list, 
and 
adding an entire
edge may cause changes to two neighbor lists. 
However, 
if 
each $\calR_i$ ignores
bits $a_{i,i}, \ldots, a_{i,n}$ for users with larger IDs, 
then this doubling
factor can be avoided. 
Our algorithms also use only the lower triangular part of $\bmA$ to avoid this doubling issue.

\smallskip
\noindent{\textbf{Interaction among Users and Multiple Rounds.}}~~While interaction in LDP has been studied before~\cite{Joseph_SODA20}, neither of Definitions~\ref{def:edge_LDP} and \ref{def:entire_edge_LDP} 
allows the interaction among users in a one-round protocol where 
user $v_i$ sends $\calR_i(\bma_i)$ to the server. 

However, the interaction 
among users 
is possible in a multi-round protocol. 
Specifically, 
at the first round, 
user $v_i$ applies a randomizer $\calR_i^1$ and 
sends $\calR_i^1(\bma_i)$ to the server. 
At the second round, the server 
calculates a message $M_i$ for $v_i$ by 
performing 
some post-processing on $\calR_i^1(\bma_i)$, possibly with the private outputs by other users. 
Let $\lambda_i$ be the post-processing algorithm on $\calR_i^1(\bma_i)$; 
i.e., $M_i = \lambda_i(\calR_i^1(\bma_i))$. 
The server sends $M_i$ to $v_i$. 
Then, $v_i$ uses a randomizer $\calR_i^2(M_i)$ that depends on $M_i$ and sends $\calR_i^2(M_i)(\bma_i)$ back to the server.
This entire computation 
provides 
DP by 
a (general) sequential composition \cite{DP_Li}: 

\begin{proposition} [Sequential composition of edge LDP] 
\label{prop:seq_comp_edge_LDP} 
  For 
  $i \in [n]$, let 
  $\calR_i^1$ be a local randomizer of user $v_i$ that takes $\bma_i$ as input. 
  Let $\lambda_i$ be a post-processing algorithm on $\calR_i^1(\bma_i)$, and $M_i = \lambda_i(\calR_i^1(\bma_i))$ be its output. 
  Let $\calR_i^2(M_i)$ be a local randomizer of $v_i$ that depends on $M_i$. 
  If $\calR_i^1$ provides $\epsilon_1$-edge LDP and for any 
  $M_i \in \mathrm{Range}(\lambda_i)$,
  $\calR_i^2(M_i)$ provides $\epsilon_2$-edge LDP, 
then the sequential composition 
$(\calR_i^1(\bma_i), \calR_i^2(M_i)(\bma_i))$
provides $(\epsilon_1 + \epsilon_2)$-edge LDP.
\end{proposition}
We provide a proof of Proposition~\ref{prop:seq_comp_edge_LDP} in \conference{the full version \cite{Imola_arXiv22}}\arxiv{Appendix~\ref{sec:proof_seq_comp_edge_LDP}}.

\smallskip
\noindent{\textbf{Global Sensitivity.}}~~We use the notion of global sensitivity~\cite{DP} to provide edge LDP: 
\begin{definition}
In edge LDP (Definition~\ref{def:edge_LDP}), the global sensitivity of a function $f: \{0,1\}^n \rightarrow \reals$ is given by:
\begin{align*}
    GS_f = \underset{\bma_i, \bma'_i \in \{0,1\}^n, \bma_i \sim \bma'_i}{\max} |f(\bma_i) - f(\bma'_i)|,
\end{align*}
where $\bma_i \sim \bma'_i$ represents that $\bma_i$ and $\bma'_i$ differ in one bit.
\end{definition}
For example, adding the Laplacian noise with mean $0$ scale $\frac{GS_f}{\epsilon}$ (denoted by $\Lap(\frac{GS_f}{\epsilon})$) to 
$f(\bma_i)$ 
provides $\epsilon$-edge LDP. 

\subsection{Utility and Communication-Efficiency}
\label{sub:utility_communication_efficiency}
\noindent{\textbf{Utility.}}
We consider a private estimate of $f_\triangle(G)$. 
Our private estimator $\hf_\triangle : \calG \rightarrow \reals$ is a post-processing of 
local randomizers $(\calR_1, \ldots, \calR_n)$ 
that 
satisfy $\epsilon$-edge LDP.
Following previous work, we use the $l_2$ loss (i.e., squared error)
\cite{Kairouz_ICML16,Wang_USENIX17,Murakami_USENIX19} and the relative error 
\cite{Bindschaedler_SP16,Chen_CCS12,Xiao_SIGMOD11} 
as utility metrics.

Specifically,
let $l_2^2$ be the expected $l_2$ loss function on a graph $G$, which maps the 
estimate $\hf_\triangle(G)$ and the true value $f_\triangle(G)$ to the expected $l_2$ loss; i.e., 
$l_2^2(f_\triangle(G), \hf_\triangle(G)) = \E[(\hf_\triangle(G) - f_\triangle(G))^2]$.
The 
expectation is taken over the randomness in the estimator $\hf$, which is
necessarily a randomized algorithm since it satisfies edge LDP.
In our theoretical analysis, we analyze the expected $l_2$ loss, as with~\cite{Kairouz_ICML16,Wang_USENIX17,Murakami_USENIX19}.

Note that the $l_2$ loss is large when $f_\triangle(G)$ is large. 
Therefore, in our experiments, we use the relative error given by $\frac{|\hf_\triangle(G) - f_\triangle(G)|}{\max\{f_\triangle(G), \eta\}}$, 
where $\eta \in \nnreals$ is a small value. 
Following convention 
\cite{Bindschaedler_SP16,Chen_CCS12,Xiao_SIGMOD11}, 
we set $\eta$ to $0.001n$. 
The estimate is very accurate when the relative error is much smaller than $1$.

\smallskip
\noindent{\textbf{Communication-Efficiency.}}
A prominent concern when performing local computations is that the computing power of
individual users is often limited. Of particular concern to our private
estimators, and a bottleneck of previous work in locally private triangle
counting~\cite{Imola_USENIX21}, is the communication overhead between users and
the server. This communication takes the form of users 
\emph{downloading} any necessary data required to compute their local randomizers and 
\emph{uploading} the output of their local randomizers. We distinguish the two
quantities because often downloading is cheaper than uploading.

Consider a $\tau$-round protocol, where $\tau \in \nats$. 
At round $j \in [\tau]$, user $v_i$ applies a local randomizer $\calR_i^j(M_i^j)$ to her neighbor list $\bma_i$, where
$M_i^j$ is a message sent from the server to user $v_i$ during round $j$.
We define the \emph{download cost} 
as 
the number of bits required to describe
$M_i^j$ and the \emph{upload cost} 
as 
the number of bits required to
describe $\calR_i^j(M_i^j)(\bma_i)$. 
Over all rounds and all users, we evaluate the \textit{maximum per-user download/upload cost}, which is given by:
\begin{align}
  \CostDL &= \textstyle{\max_{i=1}^n \sum_{j=1}^\tau \E[|M_i^j|] ~~~~ \text{(bits)}} \label{eq:cost_DL}\\
  \CostUL &= \textstyle{\max_{i=1}^n \sum_{j=1}^\tau \E[|\calR_i^j(M_i^j)(\bma_i)|]} ~~~~ \text{(bits)}. \label{eq:cost_UL}
\end{align}

The above expectations go over the probability distributions of computing the local
randomizers 
and any post-processing done by the server. 
We evaluate the maximum of the expected download/upload cost over users.

\section{Communication-Efficient Triangle Counting Algorithms}
\label{sec:algorithms}
The current state-of-the-art triangle counting algorithm~\cite{Imola_USENIX21} under edge LDP suffers from an extremely large per-user download cost; 
e.g., every user has to download a message of $400$ Gbits or more when $n=900000$. 
Therefore, it is impractical for a large graph. 
To address this issue, we propose three communication-efficient triangle algorithms under edge LDP.

We explain the overview and details of our proposed algorithms in Sections~\ref{sub:algorithms_overview} and \ref{sub:three_algorithms}, respectively.
Then we analyze the theoretical properties of our algorithms in Section~\ref{sub:algorithms_theoretical_analysis}.

\subsection{Overview}
\label{sub:algorithms_overview}

\noindent{\textbf{Motivation.}}~~The drawback of the triangle algorithm in \cite{Imola_USENIX21} is a prohibitively 
high 
download cost at the second round.
This comes from the fact that
in their algorithm, 
each user $v_i$ applies Warner's RR
(Randomized Response)~\cite{Warner_JASA65} to
bits for smaller user IDs in her neighbor list $\bma_i$ (i.e., lower triangular part of $\bmA$)
and then downloads the whole noisy graph.
Since Warner's RR outputs 1 (edge) with high probability (e.g., about $0.5$ when $\epsilon$ is close to $0$), the
number of edges in the noisy graph is extremely large---about half of the $\binom{n}{2}$ possible edges will be edges.

In this paper, we address this issue by introducing two strategies: \textit{sampling edges} and \textit{selecting edges each user downloads}.
First, each user $v_i$ samples each 1 (edge) after applying Warner's RR.
Edge sampling has been widely studied in a
non-private triangle counting problem \cite{Bera_PODS20,Eden_FOCS15,Tsourakakis_KDD09,Wu_TKDE16}.
In particular, Wu \textit{et al.}~\cite{Wu_TKDE16} compare various non-private triangle algorithms (e.g., edge sampling, node sampling, triangle sampling) and show that edge sampling provides almost the lowest estimation error.
They also formally prove that edge sampling outperforms node sampling.
Thus, sampling edges after Warner's RR is a natural choice for our private setting.

Second, we propose three strategies for selecting edges each user downloads.
The first strategy is to simply select all noisy edges; i.e., each user downloads the whole noisy graph in the same way as \cite{Imola_USENIX21}.
The second and third strategies select some edges (rather than all edges) in a more clever manner so that the estimation error is significantly reduced.
We provide a more detailed explanation in Section~\ref{sub:three_algorithms}.

\begin{figure}[t]
  \centering
  \includegraphics[width=0.99\linewidth]{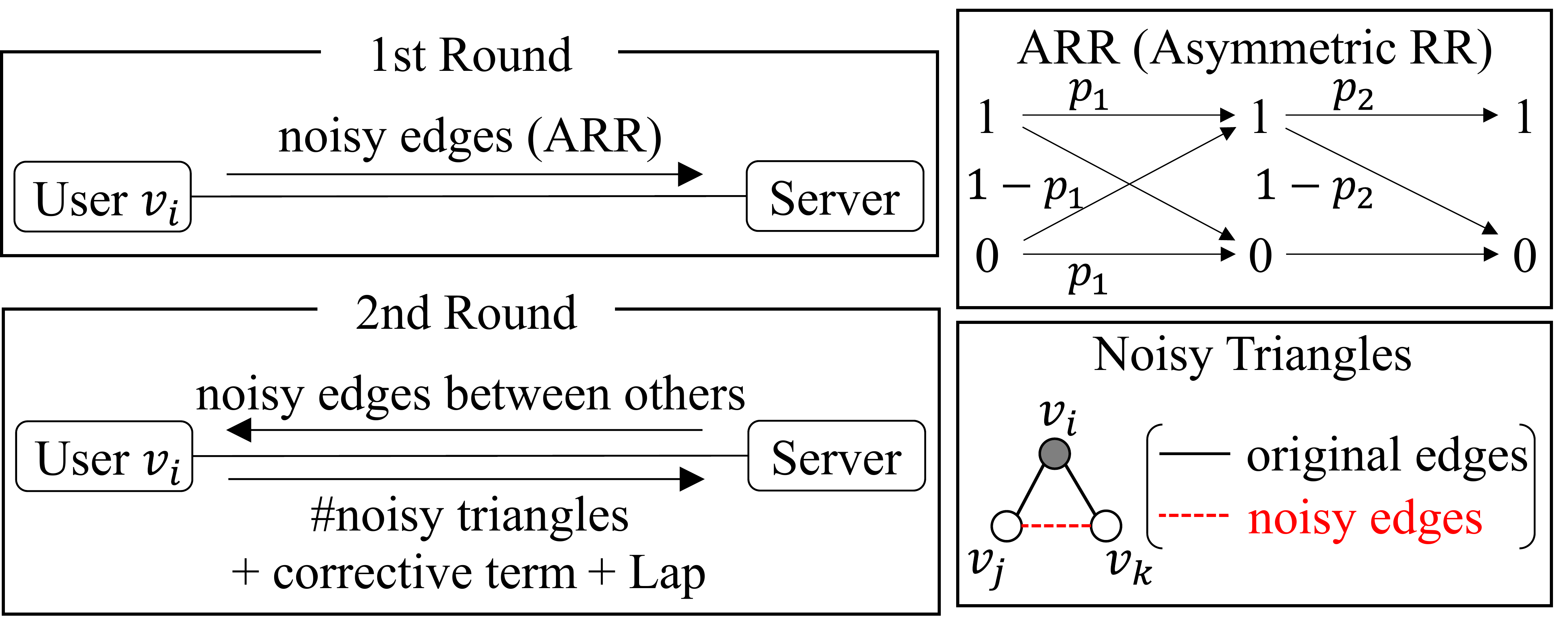}
  \vspace{-4mm}
  \caption{Overview of our communication-efficient triangle counting algorithms
  ($p_1 =\frac{e^{\epsilon}}{e^{\epsilon}+1}$,
  $p_2 \in [0,1]$).}
  \label{fig:alg_overview}
\end{figure}

\smallskip
\noindent{\textbf{Algorithm Overview.}}~~Figure~\ref{fig:alg_overview} shows the overview of our proposed algorithms.

At the first round, each user $v_i$ obfuscates
bits for smaller user IDs in her neighbor list $\bma_i$
by an LDP mechanism which we call the \textit{ARR (Asymmetric Randomized Response)} 
and sends the obfuscated bits to a server. 
The ARR is a combination of Warner's RR and edge sampling; i.e.,
we apply Warner's RR that outputs 1 or 0 as it is with probability
$p_1$ ($=\frac{e^{\epsilon}}{e^{\epsilon}+1}$)
and then sample each 1 with probability $p_2\in[0,1]$.
Unlike Warner's RR, the ARR is asymmetric in that the flip probability in the whole process is different
depending on the input value.
As with Warner's RR, the ARR provides edge LDP.
We can also significantly reduce the number of 1s (hence the communication cost) by setting
$p_2$ small.

At the second round, the server calculates a message $M_i$ for user $v_i$ consisting of some or all noisy edges between others. 
We propose three strategies for calculating $M_i$. 
User $v_i$ downloads $M_i$ from the server.
Then, since user $v_i$ knows her edges, $v_i$ can count \textit{noisy triangles} ($v_i$, $v_j$, $v_k$) such that $j<k<i$ and only one edge ($v_j$, $v_k$) is noisy, as shown in Figure~\ref{fig:alg_overview}. 
The condition $j<k<i$ is imposed to use only the lower triangular part of $\bmA$, i.e., 
to avoid the doubling issue in Section~\ref{sub:LDP}. 
User $v_i$ adds 
a corrective term
and the Laplacian noise
to the noisy triangle count 
and sends it to a server.
The corrective term is added to enable the server to obtain an unbiased estimate of $f_\triangle(G)$. 
The Laplacian noise provides
edge LDP.
Finally,
the server calculates an unbiased estimate of $f_\triangle(G)$
from the noisy data sent by users.
By composition (Proposition~\ref{prop:seq_comp_edge_LDP}),
our algorithms provide
edge LDP in total.

\smallskip
\noindent{\textbf{Remark.}}~~Note that it is also possible for the server to calculate an unbiased estimate of $f_\triangle(G)$ at the first round.
However, this results in a
prohibitively
large estimation error
because
all edges sent by users are noisy; i.e., three edges are noisy in any triangle.
In contrast, only one edge is noisy in each noisy triangle at the second round because each user $v_i$ knows two original edges
connected to $v_i$.
Consequently, we can obtain an unbiased estimate with a much smaller variance.
See Appendix~\ref{sec:one-round} for a detailed comparison.

\subsection{Algorithms}
\label{sub:three_algorithms}

\smallskip
\noindent{\textbf{ARR.}}~~First, we formally define the ARR.
The ARR has two parameters: $\epsilon \in \nnreals$ and $\mu \in [0,\frac{e^{\epsilon}}{e^{\epsilon} + 1}]$.
The parameter $\epsilon$ is the privacy budget, and $\mu$ controls the communication cost.

Let
$ARR_{\epsilon,\mu}$ be the ARR with parameters $\epsilon$ and $\mu$. It takes $0/1$ as input and outputs $0/1$ with the following probability:
\begin{align}
    \Pr[ARR_{\epsilon,\mu}(1) = b] &= \begin{cases}\mu & (b=1) \\ 1-\mu & (b=0)\end{cases} \label{eq:ARR_1}\\
    \Pr[ARR_{\epsilon,\mu}(0) = b] &= \begin{cases}\mu\rho & (b=1) \\ 1-\mu \rho & (b=0), \end{cases} \label{eq:ARR_0}
\end{align}
where $\rho = e^{-\epsilon}$.
By Figure~\ref{fig:alg_overview},
we can view this randomizer as a combination of Warner's RR~\cite{Warner_JASA65}
and edge sampling, where $\mu=p_1 p_2$.
In fact, the ARR with $\mu = p_1 =\frac{e^{\epsilon}}{e^{\epsilon}+1}$ (i.e., $p_2=1$) is equivalent to Warner's RR.

Each user $v_i$ applies the ARR to bits for smaller user IDs in her neighbor list $\bma_i$; i.e., $\calR_i(\bma_i) = (ARR_{\epsilon,\mu}(a_{i,1}), \ldots, \allowbreak ARR_{\epsilon,\mu}(a_{i,i-1}))$.
Then $v_i$ sends $\calR_i(\bma_i)$ to the server.
Since applying Warner's RR to $\bma_i$ provides $\epsilon$-edge LDP (as described in Section~\ref{sub:LDP}) and the sampling is a post-processing process, applying the ARR to $\bma_i$ also provides $\epsilon$-edge LDP by the immunity to post-processing~\cite{DP}.

Let $E' \subseteq V \times V$ be a set of noisy edges sent by users.

\smallskip
\noindent{\textbf{Which Noisy Edges to Download?}}~~Now, the main question
tackled in this paper is: \textit{Which noisy edges should each user $v_i$
download at the second round?}
Note that
user $v_i$
is not allowed to
download only a set of noisy edges that form noisy triangles
(i.e., $\{(v_j,v_k) \in E' | (v_i,v_j) \in E, (v_i,v_k) \in E$\}),
because it tells the server
who are friends with $v_i$.
In other words, user $v_i$ cannot leak her original edges to the server when she
downloads noisy edges; the server must choose which part of $E'$ to include in
the message $M_i$ it sends her.

Thus, a natural solution would be to download \textit{all noisy edges between others}
(with smaller user IDs); i.e.,
$M_i =\{(v_j, v_k) \in E' | j<k<i\}$.
We denote our algorithm with this full download strategy by \AlgOne{}.
The (inefficient) two-rounds algorithm in~\cite{Imola_USENIX21} is a special case of \AlgOne{}
without sampling ($\mu = p_1$).
In other words, \AlgOne{} is a generalization of the two-rounds algorithm in~\cite{Imola_USENIX21} using the ARR.

In this paper,
we show that we can do much better
than \AlgOne{}.
Specifically, we prove in Section~\ref{sub:algorithms_theoretical_analysis} that \AlgOne{} results in a high estimation error when the number of 4-cycles (cycles of length 4) in $G$ is large.
Intuitively, this can be explained as follows.
Suppose that $v_i$, $v_j$, $v_{i'}$, and $v_k$
($j<k<i$, $j<k<i'$)
form a 4-cycle. 
There is no triangle in this graph.
However, if there is a noisy edge between $v_j$ and $v_k$, then two (incorrect) noisy triangles appear: ($v_i$, $v_j$, $v_k$) counted by $v_i$ and ($v_{i'}$, $v_j$, $v_k$) counted by $v_{i'}$.
More generally, let $E_{ijk}$ (resp.~$E_{i'jk}$) $\in \{0,1\}$ be a random variable that takes $1$ if ($v_i$, $v_j$, $v_k$) (resp.~($v_{i'}$, $v_j$, $v_k$)) forms a noisy triangle and $0$ otherwise.
Then, the covariance $\cov(E_{ijk},E_{i'jk})$ between $E_{ijk}$ and $E_{i'jk}$ is large because the presence/absence of a single noisy edge ($v_j$, $v_k$) affects the two noisy triangles.

To address this issue, we introduce a
trick 
that makes the two noisy triangles \textit{less correlated with each other}.
We call this the \textit{4-cycle trick}.
Specifically,
we propose two algorithms in which
the server uses noisy edges connected to $v_i$ when it calculates a message $M_i$ for $v_i$.
In the first algorithm,
the server selects
noisy edges $(v_j, v_k)$ such that
one noisy edge is connected from $v_k$ to $v_i$;
i.e.,
$M_i = \{(v_j, v_k) \in E' | (v_i, v_k) \in E', j<k<i\}$.
We call this algorithm \AlgTwo{}, as one noisy edge is connected to $v_i$.
In the second algorithm,
the server selects
noisy edges $(v_j, v_k)$ such that
two noisy edges are connected from these nodes to $v_i$; i.e.,
$M_i = \{(v_j, v_k) \in E' | (v_i, v_j) \in E', (v_i, v_k) \in E', j<k<i\}$.
We call this algorithm \AlgThree{}, as two noisy edges are connected to $v_i$.
Note that user $v_i$ does not leak her original edges to the server at the time of download in these algorithms, because
the server uses only noisy edges $E'$ sent by users to calculate $M_i$.

Figure~\ref{fig:noisy_edge_DL} shows our three algorithms.
The download cost $\CostDL{}$ in (\ref{eq:cost_DL}) is
$O(\mu n^2 \log n)$, $O(\mu^2 n^2 \log n)$, and $O(\mu^3 n^2 \log n)$,
respectively, when we regard $\epsilon$ as a constant.
In our experiments, we set
the parameter $\mu$ in the ARR so that
$\mu$ in \AlgOne{} is equal to $\mu^2$ in \AlgTwo{} and also equal to $\mu^3$ in \AlgThree{};
e.g., $\mu=10^{-6}$, $10^{-3}$, and $10^{-2}$ in \AlgOne{}, \AlgTwo{}, and \AlgThree{}, respectively.
Then
the download cost
is the same between the three algorithms.

\begin{figure}[t]
  \centering
  \includegraphics[width=0.99\linewidth]{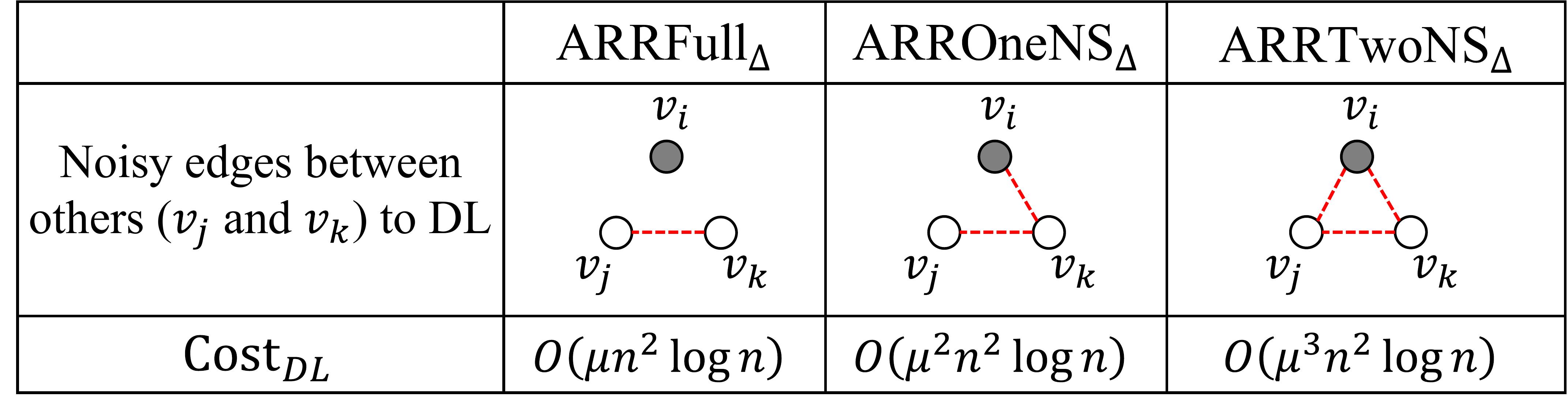}
  \vspace{-4mm}
  \caption{Noisy edges to download
  in our three algorithms.}
  \label{fig:noisy_edge_DL}
\vspace{4mm}
  \centering
  \includegraphics[width=0.95\linewidth]{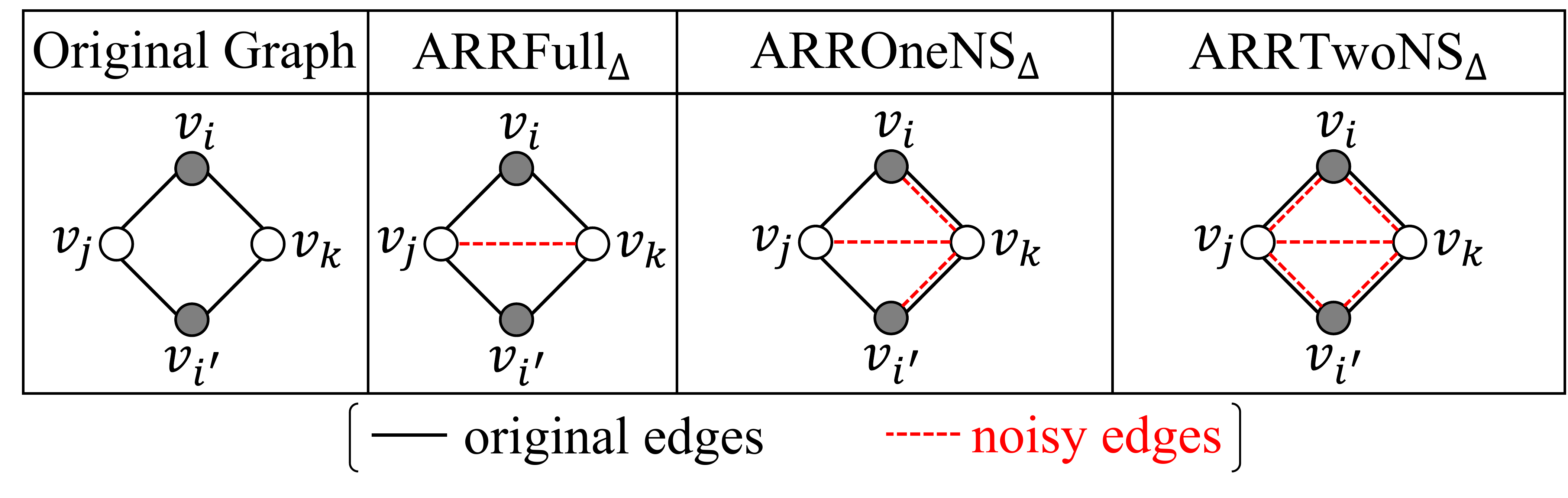}
  \vspace{-4mm}
  \caption{4-cycle trick.
  \AlgOne{} counts two (incorrect) noisy triangles when one noisy edge appears.
  \AlgTwo{} and \AlgThree{} avoid this by increasing independent noise.}
  \label{fig:four-cycle}
\end{figure}

Figure~\ref{fig:four-cycle} shows our $4$-cycle trick.
\AlgOne{} counts two (incorrect) noisy triangles when a noisy edge ($v_j$, $v_k$) appears.
In contrast,
\AlgTwo{} (resp.~\AlgThree{}) counts both the two noisy triangles
only when three (resp.~five) independent noisy edges appear,
as shown in Figure~\ref{fig:four-cycle}.
Thus, this bad event happens with a much smaller probability.
For example,
\AlgOne{} ($\mu=10^{-6}$),
\AlgTwo{} ($\mu=10^{-3}$), and
\AlgThree{} ($\mu=10^{-2}$)
count both the two noisy triangles with probability $10^{-6}$, $10^{-9}$, and $10^{-10}$, respectively.
The covariance $\cov(E_{ijk},E_{i'jk})$ of \AlgTwo{} and \AlgThree{} is also much smaller than that of \AlgOne{}.

In our experiments, we show that \AlgTwo{} and \AlgThree{} significantly outperforms \AlgOne{}
for a large-scale graph or dense graph, in both of which
the number of 4-cycles in $G$ 
is 
large.

\smallskip
\noindent{\textbf{\AlgTwo{} vs. \AlgThree{}.}}~~One
might expect that \AlgThree{} outperforms \AlgTwo{} because \AlgThree{} addresses the 4-cycle issue more aggressively; i.e.,
the number of independent noisy edges in a 4-cycle is larger in \AlgThree{},
as shown in Figure \ref{fig:four-cycle}.
However,
\AlgTwo{} can reduce the global sensitivity of the Laplacian noise
at the second round
more effectively than \AlgThree{}, as explained in Section~\ref{sec:double_clip}.
Consequently, \AlgTwo{}, which is the most tricky algorithm, achieves the smallest estimation error in our experiments.
See Sections~\ref{sec:double_clip} and \ref{sec:experiments} for details of the global sensitivity and experiments, respectively.

\smallskip
\noindent{\textbf{Three Algorithms.}}~~Below we explain the details of our three algorithms.
For ease of explanation, we assume that the maximum degree $d_{max}$ is public in Section~\ref{sub:three_algorithms}\footnote{For example, $d_{max}$ is public in Facebook: $d_{max} = 5000$~\cite{Facebook_Limit}.
If the server does not have prior knowledge about $d_{max}$, she can privately estimate $d_{max}$ and use graph projection to guarantee that each user's degree never exceeds the private estimate of $d_{max}$~\cite{Imola_USENIX21}.
In any case, the assumption in Section~\ref{sub:three_algorithms} does not undermine our algorithms, because our entire algorithms with double clipping in Section~\ref{sec:double_clip} does \textit{not} assume that $d_{max}$ is public.}.
Note, however, that
our double clipping (which is proposed to significantly reduce the global sensitivity) in Section~\ref{sec:double_clip} does \textit{not} assume that $d_{max}$ is public.
Consequently, our entire algorithms 
do \textit{not} require the assumption that $d_{max}$ is public.

\setlength{\algomargin}{5mm}
\begin{algorithm}[t]
  \SetAlgoLined
  \KwData{Graph $G \in \calG$ represented as neighbor lists $\bma_1, \ldots, \bma_n
    \in \{0,1\}^n$, privacy budgets
  $\epsilon_1,\epsilon_2 \in \nnreals$, $d_{max} \in \nnints$,
  $\mu \in [0,\frac{e^{\epsilon_1}}{e^{\epsilon_1} + 1}]$.
  }
  \KwResult{Private estimate $\hf_\triangle(G)$ of $f_\triangle(G)$.}
  [s] $\rho \leftarrow e^{-\epsilon_1}$\;
  [$v_i$, s] $\mu^* \leftarrow \mu$, $\mu^2$, and $\mu^3$ in F, O, and T, respectively\;
  \tcc{First round.}
  \For{$i=1$ \KwTo $n$}{
    [$v_i$] $\bmr_i \leftarrow (ARR_{\epsilon_1,\mu}(a_{i,1}), \ldots,
    ARR_{\epsilon_1,\mu}(a_{i,i-1}))$\;
    [$v_i$] Upload $\bmr_i = (r_{i,1}, \ldots, r_{i,i-1})$ to the server\;
  }
  [s] $E' = \{(v_j, v_k) :r_{k,j} = 1, j < k\}$\;
  \tcc{Second round.}
  \For{$i=1$ \KwTo $n$}{
    [s] Compute
    $M_i$ by (\ref{eq:M_i_I}), (\ref{eq:M_i_II}), and (\ref{eq:M_i_III}) in F, O, and T, respectively\;
    [$v_i$] Download $M_i$ from the server\;
    [$v_i$] $t_i \leftarrow |\{(v_i,v_j,v_k) :
    a_{i,j} = a_{i,k} = 1, (v_j,v_k) \in M_i, j<k<i \}|$\;
    [$v_i$] $s_i \leftarrow |\{(v_i,v_j,v_k) :
    a_{i,j} = a_{i,k} = 1, j<k<i\}|$\;
    [$v_i$] $w_i \leftarrow t_i - \mu^* \rho s_i$\;
    [$v_i$] $\hw_i \leftarrow w_i + \Lap(\frac{d_{max}}{\epsilon_2})$\;
    [$v_i$] Upload $\hw_i$ to the server\;
  }
  [s] $\hf_\triangle(G) \leftarrow \frac{1}{\mu^*(1-\rho)}\sum_{i=1}^n \hw_i$\;
  \KwRet{$\hf_\triangle(G)$}
  \caption{Our three algorithms.
  ``F'', ``O'', ``T'' are shorthands for
  \AlgOne{}, \AlgTwo{}, and \AlgThree{}, respectively.
  [$v_i$] and [s] represent that the process is run by $v_i$ and the server, respectively.
  }\label{alg:unify}
\end{algorithm}

Recall that the server calculates a message $M_i$ for $v_i$ as:
\begin{align}
\hspace{-1mm} M_i \hspace{-0.5mm} &= \hspace{-0.5mm} \{(v_j, v_k) \in E' | j<k<i\} \label{eq:M_i_I}\\
\hspace{-1mm} M_i \hspace{-0.5mm} &= \hspace{-0.5mm} \{(v_j, v_k) \in E' | (v_i, v_k) \in E', j<k<i\} \label{eq:M_i_II}\\
\hspace{-1mm} M_i \hspace{-0.5mm} &= \hspace{-0.5mm}  \{(v_j, v_k) \in E' | (v_i, v_j) \in E', (v_i, v_k) \in E', j<k<i\} \label{eq:M_i_III}
\end{align}
in \AlgOne{}, \AlgTwo{}, \AlgThree{}, respectively.

Algorithm~\ref{alg:unify} shows our three algorithms.
These algorithms are processed differently in lines 2 and 9; ``F'', ``O'', ``T'' are shorthands for \AlgOne{}, \AlgTwo{}, and \AlgThree{}, respectively.
The privacy budgets for the first and second
rounds are $\epsilon_1, \epsilon_2 \in \nnreals$, respectively.

The first round appears in lines 3-7 of Algorithm~\ref{alg:unify}.
In this round, each user applies
$ARR_{\epsilon_1, \mu}$
defined by (\ref{eq:ARR_1}) and (\ref{eq:ARR_0}) to bits $a_{i,1}, \ldots, a_{i,i-1}$ for smaller user IDs in her neighbor list $\bma_i$, i.e., lower triangular part of $\bmA$.
Let $\bmr_i = (r_{i,1}, \ldots, r_{i,i-1}) \in \{0,1\}^{i-1}$ be the obfuscated bits of $v_i$.
User $v_i$ uploads $\bmr_i$ to the server.
Then the server combines the noisy edges together, forming $E' = \{(v_j, v_k) : r_{k,j} = 1, j < k\}$.

The second round appears in lines 8-17 of Algorithm~\ref{alg:unify}.
In this round, the server computes a message $M_i$
by (\ref{eq:M_i_I}), (\ref{eq:M_i_II}), or (\ref{eq:M_i_III}),
and user $v_i$ downloads it.
Then user $v_i$ calculates the number $t_i \in \nnints$ of noisy triangles ($v_i$, $v_j$, $v_k$) such that
only one edge ($v_j$, $v_k$) is noisy, as shown in Figure~\ref{fig:alg_overview}.
User $v_i$ also calculate a corrective term $s_i \in \nnints$.
The corrective term $s_i$ is the number of
possible triangles involving $v_i$ 
and is computed to obtain an unbiased estimate of $f_\triangle(G)$.
User $v_i$ calculates $w_i = t_i - \mu^* \rho s_i$, where $\rho = e^{-\epsilon_1}$ and
$\mu^* = \mu$, $\mu^2$, and $\mu^3$
in ``F'', ``O'', and ``T'', respectively.
Then $v_i$ adds the Laplacian noise $\Lap(\frac{d_{max}}{\epsilon_2})$ to $w_i$ to provide $\epsilon_2$-edge LDP and sends the noisy value $\hw_i$ ($= w_i + \Lap(\frac{d_{max}}{\epsilon_2})$) to the server. 
Note that adding one edge increases both $t_i$ and $s_i$ 
by at most $d_{max}$. 
Thus, the global sensitivity of $w_i$ is at most $d_{max}$. 
Finally, the server calculates an estimate of $f_\triangle(G)$ as: $\hf_\triangle(G) = \frac{1}{\mu^*(1-\rho)}\sum_{i=1}^n \hw_i$.
As we prove later, $\hf_\triangle(G)$ is an unbiased estimate of $f_\triangle(G)$.

\subsection{Theoretical Analysis}
\label{sub:algorithms_theoretical_analysis}
We now introduce the theoretical guarantees on
the privacy, communication, and
utility of 
our algorithms. 

\smallskip
\noindent{\textbf{Privacy.}}~~We first show the privacy guarantees:
\begin{theorem}\label{thm:privacy_algorithms}
  For $i \in [n]$,
  let
  $\calR_i^1, \calR_i^2(M_i)$
  be the randomizers used by user $v_i$ in
  rounds $1$ and $2$ of Algorithm~\ref{alg:unify}. 
  Let
  $\calR_i(\bma_i) = (\calR_i^1(\bma_i), \calR_i^2(M_i)(\bma_i))$
  be the composition of the two randomizers. Then,
  $\calR_i$
  satisfies
  $(\epsilon_1+\epsilon_2)$-edge LDP and
  $(\calR_1,
  \ldots, \calR_n)$ satisfies  $(\epsilon_1+\epsilon_2)$-relationship DP.
\end{theorem}

Note that the doubling issue in Section~\ref{sub:LDP} does not occur, 
because we use only the lower triangular part of $\bmA$. 
By the immunity to post-processing, the estimate  $\hf_\triangle(G)$
also satisfies $(\epsilon_1+\epsilon_2)$-edge LDP and $(\epsilon_1+\epsilon_2)$-relationship DP.

\smallskip
\noindent{\textbf{Communication.}}~~Recall that we evaluate the algorithms based on their download cost~\eqref{eq:cost_DL} and upload cost~\eqref{eq:cost_UL}.

\textit{Download Cost:} The download cost is the number of bits required
to download $M_i$.
$M_i$ can be represented
as
a list of edges between others, and each edge
can be
identified with two indices (user IDs), i.e., $2 \log n$ bits.
There are $\frac{(n-1)(n-2)}{2} \approx \frac{n^2}{2}$ edges between others.
$ARR_{\epsilon_1,\mu}$ outputs 1 with probability at most $\mu$. 
In addition, each noisy triangle must have $1$, $2$, and $3$ noisy edges in \AlgOne{}, \AlgTwo{}, and \AlgThree{}, respectively, as shown in Figure~\ref{fig:noisy_edge_DL}.

Thus, 
the download cost in Algorithm~\ref{alg:unify} 
can be written as:
\begin{align}
   Cost_{DL} &\leq \mu^* n^2 \log n,  \label{eq:CostDL_F}
\end{align}
where $\mu^* = \mu$, $\mu^2$, and $\mu^3$ in \AlgOne{}, \AlgTwo{}, and \AlgThree{}, respectively. 
In (\ref{eq:CostDL_F}), 
we upper-bounded $\CostDL{}$ by using the fact that $ARR_{\epsilon_1,\mu}$ outputs 1 with probability at most $\mu$. 
However, when $d_{max} \ll n$, 
$ARR_{\epsilon_1,\mu}$ outputs 1 with probability $\mu e^{- \epsilon_1}$ in most cases. 
In that case, we can roughly approximate $\CostDL{}$ by replacing $\mu$ with $\mu e^{- \epsilon_1}$ in (\ref{eq:CostDL_F}). 

\textit{Upload Cost:}
The upload cost comes
from the number of bits required to upload $\calR_i^1(\bma_i)$ and
$\calR_i^2(M_i)(\bma_i)$.
Uploading $\calR_i^1(\bma_i)$ involves uploading
$\bmr_i$ (line 5), which is a list of up to $n$ noisy neighbors. By sending just
the indices (user IDs) of the $1$s in $\bmr_i$, each user sends $\|\bmr_i\|_1 \log n$ bits,
where $\|\bmr_i\|_1$ is the number of 1s in $\bmr_i$.
When we use $ARR_{\epsilon_1,\mu}$,
we have $\E[\|\bmr_i\|_1] \leq \mu n$.
Uploading
$\calR_i^2(M_i)$
involves uploading a single real number $\hw_i$ (line 15), which is negligibly small (e.g., 64 bits when we use a double-precision floating-point).

Thus,
the upload cost in Algorithm~\ref{alg:unify} can be written as:
\begin{align}
  Cost_{UL} \leq \mu n \log n.
\label{eq:CostUL_proposal}
\end{align}
Clearly, 
$\CostUL{}$ 
is much smaller than 
$\CostDL{}$ 
for large $n$.

\smallskip
\noindent{\textbf{Utility.}}~~Analyzing the expected $l_2$ loss $l_2^2(f_\triangle(G), \hf_\triangle(G))$ of the algorithms involves first proving that
the estimator $\hf_\triangle$ is unbiased
and then
analyzing
the
variance $\V[\hf_\triangle(G)]$ to obtain an upper-bound on
$l_2^2(f_\triangle(G), \hf_\triangle(G))$.
This is given in the following:

\begin{theorem}\label{thm:l2loss_algorithms}
  Let $G \in \calG$, $\epsilon_1, \epsilon_2 \in \nnreals$, and
  $\mu \in [0,\frac{e^{\epsilon_1}}{e^{\epsilon_1} + 1}]$.
  Let $\hf_\triangle^F(G), \hf_\triangle^O(G)$,
  and $\hf_\triangle^T(G)$ be
  the
  estimates output
  respectively by \AlgOne{}, \AlgTwo{}, and \AlgThree{} in Algorithm~\ref{alg:unify}.
  Then, $\E[\hf_\triangle^F(G)] = \E[\hf_\triangle^O(G)] = \E[\hf_\triangle^T(G)] = f_\triangle(G)$ (i.e., estimates are unbiased) and
   \begin{align*}
      l_2^2(f_\triangle(G), \hf_\triangle^F(G)) &\leq \textstyle{\frac{2C_4(G)+S_2(G)}{\mu(1-e^{\epsilon_1})^2} + \frac{2nd_{max}^2}{\mu^2(1-e^{\epsilon_1})^2\epsilon_2^2}} \\
      l_2^2(f_\triangle(G), \hf_\triangle^O(G)) &\leq \textstyle{\frac{\mu(2 C_4(G) + 6S_3(G)) + S_2(G)}{\mu^2(1-e^{\epsilon_1})^2} \hspace{-0.5mm}+\hspace{-0.5mm} \frac{2nd_{max}^2}{\mu^4(1-e^{\epsilon_1})^2\epsilon_2^2}}\\
      l_2^2(f_\triangle(G), \hf_\triangle^T(G)) &\leq \textstyle{\frac{\mu^2(2 C_4(G) + 6S_3(G)) + S_2(G)}{\mu^3(1-e^{\epsilon_1})^2} \hspace{-0.5mm}+\hspace{-0.5mm} \frac{2nd_{max}^2}{\mu^6(1-e^{\epsilon_1})^2\epsilon_2^2},}
   \end{align*}
   where $C_4(G)$ is the number of $4$-cycles in $G$
   and $S_k(G)$ is the number of $k$-stars in $G$. 
\end{theorem}
For each of the three upper-bounds in Theorem~\ref{thm:l2loss_algorithms}, the first and second terms are the estimation errors caused by empirical estimation and the Laplacian noise, respectively.
We also note that
$C_4(G) = S_3(G) = O(n d_{max}^3)$
and $S_2(G) = O(n d_{max}^2)$.
Thus, for small
$\mu$,
the $l_2$ loss of empirical estimation can be expressed as $O(n d_{max}^3)$, $O(n d_{max}^2)$, and $O(n d_{max}^2)$ in \AlgOne{}, \AlgTwo{}, \AlgThree{}, respectively
(as the factors of $C_4(G)$ and $S_3(G)$
diminish for small $\mu$).

This highlights our $4$-cycle trick.
The large $l_2$ loss of \AlgOne{} is caused by the number $C_4(G) = O(n d_{max}^3)$ of $4$-cycles.
\AlgTwo{} and \AlgThree{} addresses this issue by increasing independent noise, as shown in Figure~\ref{fig:four-cycle}.

\section{Double Clipping}
\label{sec:double_clip}

In Section~\ref{sec:algorithms}, 
we showed that the estimation error caused by empirical estimation (i.e., the first term in Theorem~\ref{thm:l2loss_algorithms}) is significantly reduced by the $4$-cycle trick. 
However, 
the estimation error is 
still very large in our algorithms presented in Section~\ref{sec:algorithms}, as shown in our experiments. 
This is because 
the estimation error by the Laplacian noise (i.e., the second term in Theorem~\ref{thm:l2loss_algorithms}) 
is very large, especially for small $\epsilon_2$ or 
$\mu$. 
This error term is tight and unavoidable as long as we use $d_{max}$ as a global sensitivity, which suggests that we need a better global sensitivity analysis.
To significantly reduce the global sensitivity, we propose a novel \textit{double clipping} technique. 

We describe the overview and details of our double clipping in Sections~\ref{sub:clip_overview} and \ref{sub:algorithms}, respectively. 
Then we perform theoretical analysis in Section~\ref{sub:clip_theoretical_analysis}.

\subsection{Overview}
\label{sub:clip_overview}

\noindent{\textbf{Motivation.}}~~Figure~\ref{fig:reduce_noisy_triangles} shows noisy triangles involving edge $(v_i,v_j)$ counted by user $v_i$ in our three algorithms. 
Our algorithms in Section~\ref{sec:algorithms} 
use the fact that the number of such noisy triangles (hence the global sensitivity) is upper-bounded by the maximum degree $d_{max}$ because 
adding one edge increases the triangle count by at most $d_{max}$. 
Unfortunately, this upper-bound is too large, as shown in our experiments. 

In this paper, we 
significantly reduce this upper-bound by using the parameter $\mu$ in the ARR and user $v_i$'s degree $d_i \in \nnints$ for users with smaller IDs. 
For example, 
the number of noisy triangles involving $(v_i,v_j)$ in \AlgOne is 
expected to be around 
$\mu d_i$ 
because one noisy edge is included in each noisy triangle (as shown in Figure~\ref{fig:reduce_noisy_triangles}) and all noisy edges are independent. 
$\mu d_i$ is very small, especially when we set $\mu \ll 1$ to reduce the communication cost.  

However, we cannot directly use $\mu d_i$ as an upper-bound of the global sensitivity in \AlgOne for two reasons. 
First, $\mu d_i$ leaks the exact value of user $v_i$'s degree $d_i$ 
and 
violates edge LDP. 
Second, the number of noisy triangles involving $(v_i,v_j)$ exceeds $\mu d_i$ with high probability (about $0.5$). 
Thus, the noisy triangle count cannot be upper-bounded by $\mu d_i$. 

\begin{figure}[t]
  \centering
  \includegraphics[width=0.78\linewidth]{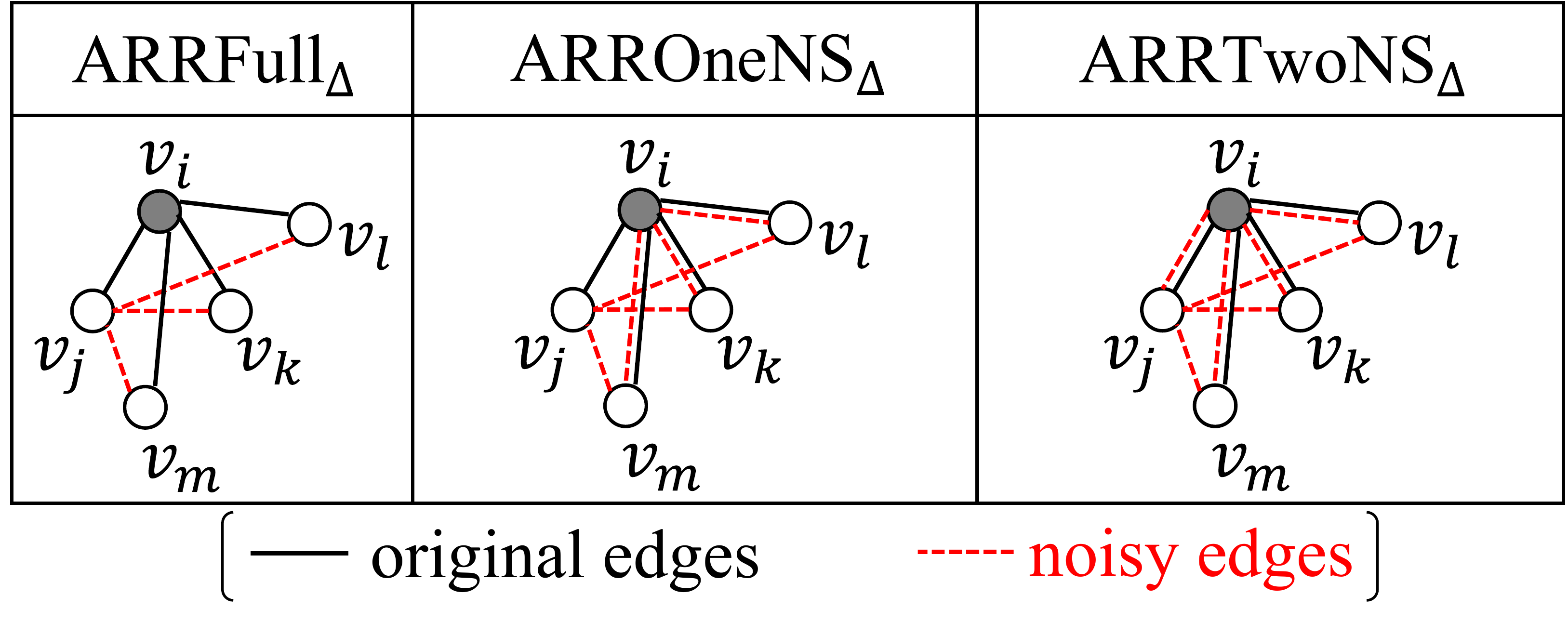}
  \vspace{-4mm}
  \caption{Noisy triangles involving edge $(v_i,v_j)$ counted by user $v_i$ ($j<k,l,m<i$).} 
  \label{fig:reduce_noisy_triangles}
\vspace{2mm}
  \centering
  \includegraphics[width=0.99\linewidth]{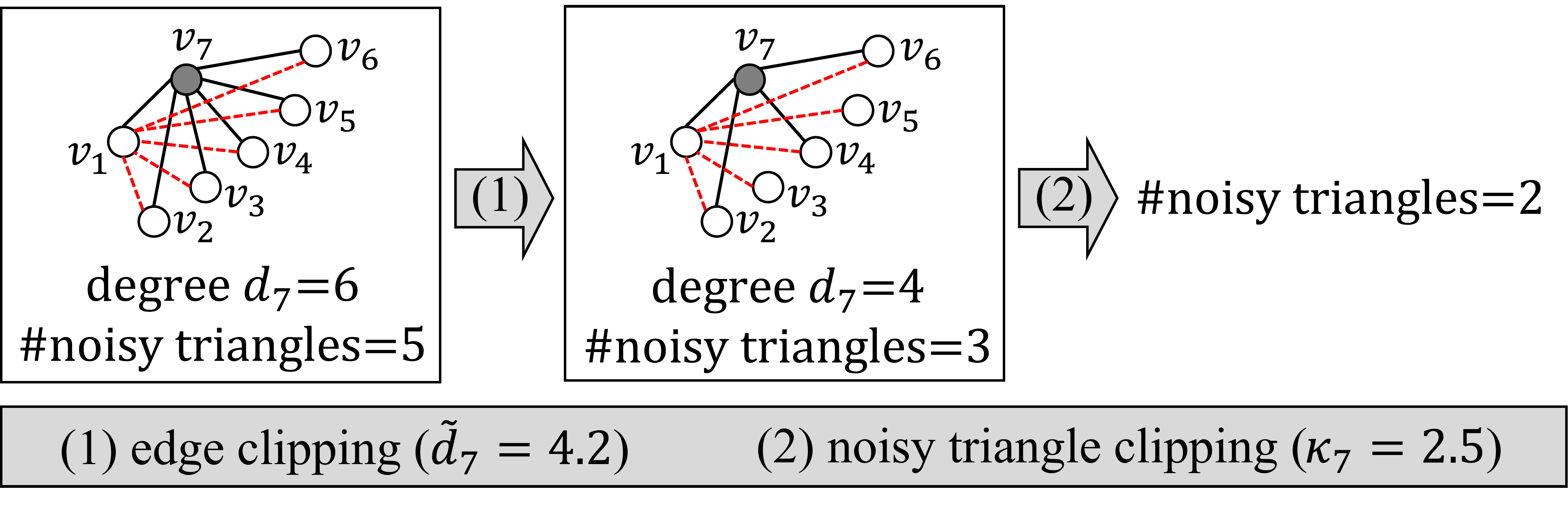}
  \vspace{-4mm}
  \caption{Overview of double clipping applied to edge ($v_1,v_7$).} 
  \label{fig:double-clip_overview}
\end{figure}

To address these two issues, 
we propose a double clipping technique, which is explained below. 

\smallskip
\noindent{\textbf{Algorithm Overview.}}~~Figure~\ref{fig:double-clip_overview} shows 
the overview of our double clipping, which consists of an 
\textit{edge clipping} 
and \textit{noisy triangle clipping}. 
The edge clipping 
addresses the first issue (i.e., leakage of $d_i$) 
as follows. 
It privately computes 
a noisy version of 
$d_i$ (denoted by $\td_i$) with edge LDP. 
Then it 
removes some neighbors from a neighbor list $\bma_i$ so that the degree of $v_i$ never exceeds 
the noisy degree $\td_i$. 
This removal process is also known as graph projection \cite{Day_SIGMOD16,Ding_TKDE21,Kasiviswanathan_TCC13,Raskhodnikova_arXiv15}. 
Edge clipping 
is 
used in \cite{Imola_USENIX21} to obtain a 
noisy version of 
$d_{max}$. 

The main novelty in our double clipping lies at the \textit{noisy triangle clipping} to address the second issue (i.e., excess of the noisy triangle count). 
This issue appears 
when 
we attempt to reduce the global sensitivity by using 
a very small sampling probability for each edge. 
Therefore, the noisy triangle clipping has not been studied in the existing works on private triangle counting 
\cite{Ding_TKDE21,Imola_USENIX21,Karwa_PVLDB11,Kasiviswanathan_TCC13,Sun_CCS19,Ye_ICDE20,Ye_TKDE21,Zhang_SIGMOD15}, 
because they do not apply a sampling technique. 

Our noisy triangle clipping reduces the noisy triangle count so that it never exceeds a 
user-dependent clipping threshold 
$\kappa_i \in \nnreals$. 
Then a crucial issue is how to set 
an appropriate 
threshold 
$\kappa_i$. 
We theoretically analyze the probability that the noisy triangle count exceeds $\kappa_i$ 
(referred to as the \textit{triangle excess probability}) 
as a function of 
the ARR parameter $\mu$ and the 
noisy degree $\td_i$. 
Then we set $\kappa_i$ so that 
the triangle excess probability 
is very small ($=10^{-6}$ in our experiments). 

We use the clipping threshold $\kappa_i$ as a global sensitivity. 
Note that $\kappa_i$ provides edge LDP because $\td_i$ provides edge LDP, 
i.e., immunity to post-processing \cite{DP}. 
$\kappa_i$ is also very small when $\mu \ll 1$, as it is determined based on $\mu$.

\subsection{Algorithms}
\label{sub:algorithms}
Algorithm~\ref{alg:clip} shows our double clipping algorithm. 
All the processes are run by user $v_i$ at the second round. 
Thus, there is no interaction with the server in Algorithm~\ref{alg:clip}.

\setlength{\algomargin}{5mm}
\begin{algorithm}[t]
  \SetAlgoLined
  \KwData{Neighbor list $\bma_i \in \{0,1\}^n$, privacy budget
  $\epsilon_0 \in \nnreals$ 
  $\mu \in [0,\frac{e^{\epsilon_1}}{e^{\epsilon_1} + 1}]$, 
  $\alpha \in \nnreals$, 
  $\beta \in \nnreals$.
  }
  \KwResult{$\hw_i$.}
  $\mu^* \leftarrow \mu$, $\mu^2$, and $\mu^3$ in F, O, and T, respectively\;
  \tcc{Edge clipping.}
  $\td_i = \max\{d_i + \Lap(\frac{1}{\epsilon_0}) + \alpha$, 0\}\;
  \tcc{Remove $d_i - \lfloor \td_i \rfloor$ neighbors if $d_i > \td_i$.}
  $\bma_i \leftarrow \texttt{GraphProjection}(\bma_i, \td_i)$\;
  \tcc{Noisy triangle clipping.}
  \For{$j$ \rm{such that} $a_{i,j} = 1$ \rm{and} $j<i$}{
    $t_{i,j} \leftarrow |\{(v_i,v_j,v_k) : a_{i,k} = 1, (v_j,v_k) \in M_i, j<k<i \}|$\;
  }
   \tcc{Calculate $\kappa_i \in [\mu^* \td_i, \td_i]$ s.t. the triangle excess probability is $\beta$ or less.}
  $\kappa_i \leftarrow \texttt{ClippingThreshold}(\mu, \td_i, \beta)$\;
  $t_i \leftarrow \sum_{a_{i,j} = 1, j<i} \min \{t_{i,j}, \kappa_i\}$\;
  $s_i \leftarrow |\{(v_i,v_j,v_k) : a_{i,j} = a_{i,k} = 1, j<k<i\}|$\;
  $w_i \leftarrow t_i - \mu^* \rho s_i$\;
  $\hw_i \leftarrow w_i + \Lap(\frac{\kappa_i}{\epsilon_2})$\;
  \KwRet{$\hw_i$}
  \caption{Our double clipping algorithm. 
  ``F'', ``O'', ``T'' are shorthands for 
  \AlgOne{}, \AlgTwo{}, and \AlgThree{}, respectively.
  All the processes are run by user $v_i$.
  }\label{alg:clip}
\end{algorithm}

\smallskip
\noindent{\textbf{Edge Clipping.}}~~The edge clipping appears in lines 2-3 of Algorithm~\ref{alg:clip}. 
It uses a privacy budget $\epsilon_0 \in \nnreals$. 

In line 2, user $v_i$ adds the Laplacian noise $\Lap(\frac{1}{\epsilon_0})$ to her degree $d_i$. 
Since adding/removing one edge changes $d_i$ by at most $1$, this process provides $\epsilon_0$-edge LDP. 
$v_i$ also adds some non-negative constant 
$\alpha \in \nnreals$ 
to $d_i$. 
We add this value so that edge removal (in line 3) occurs with a very small probability; 
e.g., in our experiments, we set $\alpha = 150$, where 
edge removal occurs with probability $1.5 \times 10^{-7}$ when $\epsilon_0 = 0.1$. 
A similar technique is introduced in \cite{Sun_CCS19} to provide ($\epsilon, \delta)$-DP \cite{DP} with small $\delta$. 
The difference between ours and \cite{Sun_CCS19} is that we perform edge clipping 
to always provide $\epsilon$-DP; i.e., $\delta = 0$.
Let $\td_i \in \nnreals$ be the noisy degree of $v_i$.

In line 3, user $v_i$ calls the function \texttt{GraphProjection}, which performs graph projection as follows; 
if $d_i > \td_i$, randomly remove $d_i - \lfloor \td_i \rfloor$ neighbors from $\bma_i$; otherwise, do nothing. 
Consequently, the degree 
of $v_i$ never exceeds $\td_i$. 

\begin{table*}[t]
  \centering
  \begin{tabular}{|l|c|c|c|}
    \hline
    & \AlgOne & \AlgTwo & \AlgThree \\ \hline
    Privacy 
    & \multicolumn{3}{|c|}{$(\epsilon_0 + \epsilon_1 + \epsilon_2)$-edge LDP and $(\epsilon_0 + \epsilon_1 + \epsilon_2)$-relationship DP} \\ \hline
    Expected $l_2$ loss 
    & $O\left(\frac{n d_{max}^3}{\mu(1-e^{-\epsilon_1})^2} + \frac{2 \sum_{i=1}^n \kappa_{i}^2}{\mu^2(1-e^{-\epsilon_1})^2 \epsilon_2^2}\right)$
    & $O\left(\frac{n d_{max}^2}{\mu^2(1-e^{-\epsilon_1})^2} + \frac{2 \sum_{i=1}^n \kappa_{i}^2}{\mu^4 (1-e^{-\epsilon_1})^2 \epsilon_2^2} \right)$
    & $O\left(\frac{n d_{max}^2}{\mu^3(1-e^{-\epsilon_1})^2} + \frac{2 \sum_{i=1}^n \kappa_{i}^2}{\mu^6 (1-e^{-\epsilon_1})^2 \epsilon_2^2} \right)$ \\ \hline
    $\text{Cost}_{DL}$ & 
    $\mu n^2 \log n$ 
    & 
    $\mu^2 n^2 \log n$ 
    & 
    $\mu^3 n^2 \log n$ 
    \\ \hline
    $\text{Cost}_{UL}$ & 
    $\mu n \log n$ & $\mu n \log n$ & $\mu n \log n$ 
    \\ \hline
  \end{tabular}
  \vspace{-2mm}
  \caption{Performance guarantees 
  of our three algorithms with double clipping when the edge removal and triangle removal do not occur. 
  The expected $l_2$ loss assumes that $\mu$ is small. 
  The download (resp.~upload) cost is an upper-bound in (\ref{eq:CostDL_F}) 
  (resp.~(\ref{eq:CostUL_proposal})).
  }
  \label{tab:privacy_utility_cost}
\end{table*}

\smallskip
\noindent{\textbf{Noisy Triangle Clipping.}}~~The noisy triangle clipping appears in lines 4-11 of Algorithm~\ref{alg:clip}. 

In lines 4-6, 
user $v_i$ calculates the number $t_{i,j} \in \nnints$ of noisy triangles ($v_i, v_j, v_k$) ($j<k<i$) involving $(v_i,v_j)$ 
(as shown in Figure~\ref{fig:reduce_noisy_triangles}). 
Note that the total number $t_i$ of noisy triangles of $v_i$ can be expressed as: 
$t_i = \sum_{a_{i,j}=1, j<i} t_{i,j}$. 
In line 7, $v_i$ calls the function \texttt{ClippingThreshold}, which calculates a clipping threshold 
$\kappa_i \in [\mu^* \td_i, \td_i]$ 
($\mu^* = \mu$, $\mu^2$, and $\mu^3$ in 
``F'', ``O'', and ``T'', respectively) 
based on the ARR parameter $\mu$ and the noisy degree $\td_i$ so that 
the triangle excess probability does not exceed some constant $\beta \in \nnreals$. 
We explain how to calculate 
the triangle excess probability 
in Section~\ref{sub:clip_theoretical_analysis}. 
In line 8, $v_i$ calculates the total number $t_i$ of noisy triangles by summing up $t_{i,j}$, with the exception that $v_i$ adds $\kappa_i$ 
if $t_{i,j} > \kappa_i$. 
In other words, triangle removal occurs 
if 
$t_{i,j} > \kappa_i$. 
Then, 
the number 
of noisy triangles involving $(v_i,v_j)$ never exceeds $\kappa_i$. 

Lines 9-11 in Algorithm~\ref{alg:clip} are the same as lines 12-14 in Algorithm~\ref{alg:unify}, except that 
the global sensitivity in the former (resp.~latter) is $\kappa_i$ (resp.~$d_{max}$). 
Line 11 in Algorithm~\ref{alg:clip} provides $\epsilon_2$-edge LDP because the number of triangles involving $(v_i,v_j)$ is now upper-bounded by $\kappa_i$. 

\smallskip
\noindent{\textbf{Our Entire Algorithms with Double Clipping.}}~~We can run our algorithms \AlgOne{}, \AlgTwo{}, \AlgThree{} with double clipping just by replacing lines 11-14 in Algorithm~\ref{alg:unify} with lines 2-11 in Algorithm~\ref{alg:clip}. 
That is, after calculating $\hw_i$ by Algorithm~\ref{alg:clip}, $v_i$ uploads $\hw_i$ to the server. 
Then the server calculates an estimate of $f_\triangle(G)$ as $\hf_\triangle(G) = \frac{1}{\mu^*(1-\rho)}\sum_{i=1}^n \hw_i$. 

We also note that the input $d_{max}$ in Algorithm~\ref{alg:unify} is no longer necessary thanks to the edge clipping; i.e., our entire algorithms with double clipping do not assume that $d_{max}$ is public. 

\subsection{Theoretical Analysis}
\label{sub:clip_theoretical_analysis}
We now perform a theoretical analysis on the privacy and utility of our double clipping. 

\smallskip
\noindent{\textbf{Privacy.}}~~We begin with the privacy guarantees:
\begin{theorem}\label{thm:privacy_DC}
  For $i \in [n]$, 
  let $\calR_i^1, \calR_i^2(M_i)$ be the randomizers used by user $v_i$ in
  rounds $1$ and $2$ of our algorithms with double clipping (Algorithms~\ref{alg:unify} and \ref{alg:clip}). 
  Let $\calR_i(\bma_i) = (\calR_i^1(\bma_i), \calR_i^2(M_i)(\bma_i))$ 
  be the composition of the two randomizers. 
  Then,
  $\calR_i$ satisfies $(\epsilon_0 + \epsilon_1 + \epsilon_2)$-edge LDP, 
  and $(\calR_1,
  \ldots, \calR_n)$ satisfies $(\epsilon_0 + \epsilon_1 + \epsilon_2)$-relationship DP.
\end{theorem}

\smallskip
\noindent{\textbf{Utility.}}~~Next, we show the triangle excess probability: 

\begin{theorem}\label{thm:triangle_excess}
In Algorithm~\ref{alg:clip}, the triangle excess probability (i.e., probability that the number of noisy triangles $t_{i,j}$ involving edge $(v_i, v_j)$ exceeds a clipping threshold $\kappa_i$) is:
  \begin{align}
    \hspace{-1mm} \Pr(t_{i,j} > \kappa_i) &\leq \textstyle{\exp \left[-\td_i D \left(\frac{\kappa_i}{\td_i} \parallel \mu \right) \right]} \label{eq:AlgI_clip_bound} \\
    \hspace{-1mm} \Pr(t_{i,j} > \kappa_i) &\leq \textstyle{\exp \left[-\td_i D \left(\frac{\kappa_i}{\td_i} \parallel \mu^2 \right) \right]} \label{eq:AlgII_clip_bound}\\
    \hspace{-1mm} \Pr(t_{i,j} > \kappa_i) &\leq 
    \textstyle{\mu \exp \left[-\td_i D \left(\frac{\max\{\kappa_i,\mu^2 \td_i\}}{\td_i} \parallel \mu^2 \right) \right]}
    \label{eq:AlgIII_clip_bound}
  \end{align}
  in \AlgOne{}, \AlgTwo{}, and \AlgThree{}, respectively, 
  where 
  $D(p_1 \parallel p_2)$ is the Kullback-Leibler divergence between two Bernoulli distributions; i.e., 
\begin{align*}
    \textstyle{D(p_1 \parallel p_2) = p_1 \log \frac{p_1}{p_2} + (1-p_1) \log \frac{1-p_1}{1-p_2}.}
\end{align*}
\end{theorem}
In all of (\ref{eq:AlgI_clip_bound}), (\ref{eq:AlgII_clip_bound}), and (\ref{eq:AlgIII_clip_bound}), we use the Chernoff bound, which is known to be reasonably tight \cite{Arratia_BMB89}. 

\smallskip
\noindent{\textbf{Setting $\kappa_i$.}}~~The function \texttt{ClippingThreshold} in Algorithm~\ref{alg:clip} sets a clipping threshold $\kappa_i$ of user $v_i$ based on Theorem~\ref{thm:triangle_excess}. 
Specifically, we set $\kappa_i = \lambda_i \mu^* \td_i$, where $\lambda_i \in \nats$, and calculate $\lambda_i$ as follows. 
We initially set $\lambda_i = 1$ and keep increasing $\lambda_i$ by $1$ 
until the upper-bound (i.e., right-hand side of (\ref{eq:AlgI_clip_bound}), (\ref{eq:AlgII_clip_bound}), or (\ref{eq:AlgIII_clip_bound})) is smaller than or equal to the triangle excess probability $\beta$. 
In our experiments, we set $\beta = 10^{-6}$. 

\smallskip
\noindent{\textbf{Large $\kappa_i$ of \AlgThree{}.}}~~By 
(\ref{eq:AlgI_clip_bound}) and (\ref{eq:AlgII_clip_bound}), the upper-bound on the triangle excess probability is the same between \AlgOne{} and \AlgTwo{}. 
In contrast, 
\AlgThree{} has a larger upper-bound. 
For example, 
when $\kappa_i = 15 \mu^* \td_i$, $\mu^* = 10^{-3}$, and $\td_i=1000$, 
the right-hand sides of (\ref{eq:AlgI_clip_bound}), (\ref{eq:AlgII_clip_bound}), and (\ref{eq:AlgIII_clip_bound}) are $2.5 \times 10^{-12}$, $2.5 \times 10^{-12}$, and $3.3 \times 10^{-2}$, respectively. 
Consequently, \AlgThree{} has a larger global sensitivity $\kappa_i$ for the same value of $\beta$.

We can explain a large global sensitivity $\kappa_i$ of \AlgThree{} as follows. 
The number $t_{i,j}$ of noisy triangles involving $(v_i,v_j)$ in 
\AlgOne{} is expected to be around $\mu d_i$ because one noisy edge is in each noisy triangle (as in Figure~\ref{fig:reduce_noisy_triangles}) and all noisy edges are independent. 
For the same reason, $t_{i,j}$ in \AlgTwo{} is expected to be around $\mu^2 d_i$.  
However, 
$t_{i,j}$ in \AlgThree{} is \textit{not} expected to be around $\mu^3 d_i$, because all the noisy triangles have noisy edge $(v_i,v_j)$ in common (as in Figure~\ref{fig:reduce_noisy_triangles}). 
Then, 
the expectation of $t_{i,j}$ 
largely depends on the presence/absence of the noisy edge $(v_i,v_j)$; i.e., if noisy edge $(v_i,v_j)$ exists, 
it is $\mu^2 d_i$; otherwise, $0$. 
Thus, 
$\kappa_i$ 
cannot be effectively reduced by double clipping. 

\smallskip
\noindent{\textbf{Summary.}}~~The 
performance guarantees 
of our three algorithms with double clipping can be summarized in Table~\ref{tab:privacy_utility_cost}.

The first and second terms of the expected 
$l_2$ loss are the $l_2$ loss of empirical estimation and that of the Laplacian noise, respectively. 
For small 
$\mu$, 
the $l_2$ loss of empirical estimation can be expressed as $O(n d_{max}^3)$, $O(n d_{max}^2)$, and $O(n d_{max}^2)$ in \AlgOne{}, \AlgTwo{}, \AlgThree{}, respectively, as explained in Section~\ref{sub:algorithms_theoretical_analysis}. 
The $l_2$ loss of the Laplacian noise is 
$O(\sum_{i=1}^n \kappa_i^2)$, 
which is much smaller than $O(n d_{max}^2)$. 
Thus, our \AlgTwo{} that effectively reduces $\kappa_i$ provides the smallest error, 
as shown in our experiments.

We also note that 
both the space and the time complexity to compute and send $M_i$ in our algorithms 
are $O(\mu^* n^2)$ (as $|E'| =  O(\mu^* n^2)$), which is much smaller than \cite{Imola_USENIX21} ($=O(n^2)$).

\section{Experiments}
\label{sec:experiments}

To evaluate each component of our algorithms in Sections~\ref{sec:algorithms} and \ref{sec:double_clip} as well as our entire algorithms (i.e., \AlgOne, \AlgTwo, \AlgThree with double clipping), we pose the following three research questions:
\begin{description}[leftmargin=9.75mm]
    \item[RQ1.] 
    How do our three triangle counting algorithms 
    (i.e., \AlgOne, \AlgTwo, \AlgThree) in Section~\ref{sec:algorithms} compare with each other in terms of accuracy?
    \item[RQ2.] 
    How much does our double clipping technique in Section~\ref{sec:double_clip} decrease the estimation error?
    \item[RQ3.] 
    How much do our entire algorithms reduce the communication cost, compared to the existing algorithm~\cite{Imola_USENIX21}, while keeping high utility (e.g., relative error $\ll 1$)?
\end{description}
In Appendix~\ref{sec:one-round}, we also compare our entire algorithms with one-round algorithms.

\subsection{Experimental Set-up}
\label{sub:setup}
In our experiments, we used two real graph datasets:

\smallskip
\noindent{\textbf{Gplus.}}~~The Google+ dataset~\cite{McAuley_NIPS12} (denoted by \GPlus{}) 
was collected from users who had shared circles. 
From the dataset, we 
constructed 
a social graph $G=(V,E)$ with $107614$ nodes (users) and $12238285$ edges, where 
edge $(v_i,v_j) \in E$ represents that $v_i$ follows or is followed by $v_j$. 
The average (resp.~maximum) degree in $G$ is $113.7$ (resp.~$20127$). 

\smallskip
\noindent{\textbf{IMDB.}}~~The IMDB (Internet Movie Database)~\cite{IMDB_GD05} (denoted by \IMDB{}) includes a bipartite graph between $896308$ actors and $428440$ movies. 
From this, we constructed a graph $G=(V,E)$ with $896308$ nodes (actors) and $57064358$ edges, where edge $(v_i,v_j) \in E$ represents that $v_i$ and $v_j$ have played in the same movie. 
The average (resp.~maximum) degree in $G$ is $63.7$ (resp.~$15451$). 
Thus, \IMDB{} is more sparse than \GPlus{}. 

\smallskip
In \conference{the full version \cite{Imola_arXiv22}}\arxiv{Appendix~\ref{sec:BAmodel}}, we also evaluate our algorithms using a synthetic graph based on the Barab\'{a}si-Albert model~\cite{NetworkScience}, which has a power-law degree distribution. 

We evaluated our algorithms while changing $\mu^*$, where $\mu^* = \mu$, $\mu^2$, and $\mu^3$ in \AlgOne{}, \AlgTwo{}, and \AlgThree{}, respectively. 
$\CostDL{}$ is the same between the three algorithms. 
We typically 
set the total privacy budget $\epsilon$ to $\epsilon=1$ 
(at most $2$) 
because it is acceptable in many practical scenarios \cite{DP_Li}. 

In our double clipping, we set $\alpha = 150$ and $\beta = 10^{-6}$ so that both edge removal and triangle removal occur with a very small probability ($\leq 10^{-6}$ when $\epsilon_0 = 0.1$). 
Then for each algorithm, we evaluated the relative error between the true triangle count $f_\triangle(G)$ and its estimate $\hf_\triangle(G)$. 
Since the estimate $\hf_\triangle(G)$ varies depending on the randomness of LDP mechanisms, we ran each algorithm $\tau \in \nats$ times ($\tau=20$ and $10$ for \GPlus{} and \IMDB{}, respectively) and averaged the relative error over the $\tau$ cases.

\subsection{Experimental Results}
\label{sub:results}

\smallskip
\noindent{\textbf{Performance Comparison.}}~~First, 
we 
evaluated our algorithms with the Laplacian noise. 
Specifically, we evaluated all possible combinations of our three algorithms with and without our double clipping (six combinations in total) and compared them with 
the existing two-rounds algorithm in~\cite{Imola_USENIX21}.  
For algorithms with double clipping, we divided the total privacy budget $\epsilon$ as: 
$\epsilon_0 = \frac{\epsilon}{10}$ and 
$\epsilon_1 = \epsilon_2 = \frac{9\epsilon}{20}$. 
Here, we set a very small budget ($\epsilon_0 = \frac{\epsilon}{10}$) for edge clipping because the degree has a small sensitivity (sensitivity$=1$). 
For algorithms without double clipping, we divided $\epsilon$ as $\epsilon_1 = \epsilon_2 = \frac{\epsilon}{2}$ and 
used the maximum degree $d_{max}$ as the global sensitivity. 

\begin{figure}[t]
  \centering
  \includegraphics[width=0.96\linewidth]{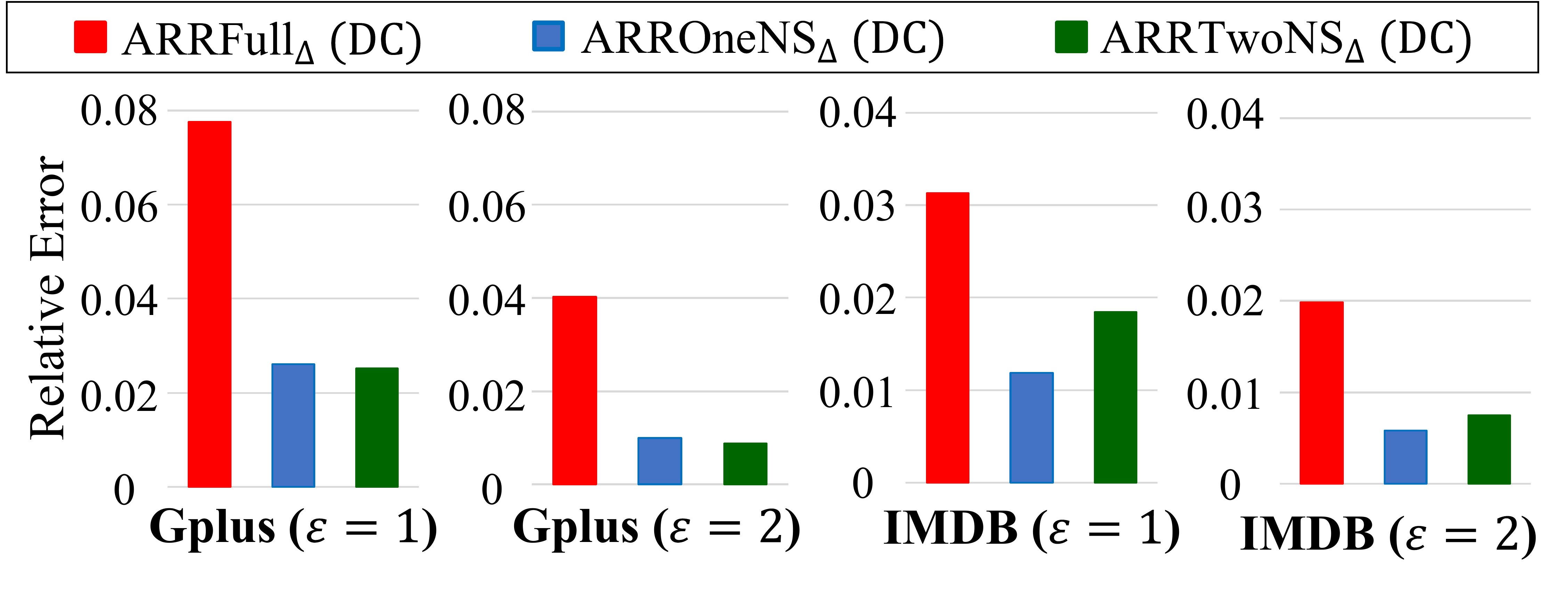}
  \vspace{-6mm}
  \caption{Relative error of our three algorithms with double clipping (``DC'') when $\epsilon=1$ or $2$ and 
  $\mu^*=10^{-3}$ 
  ($n=107614$ in \GPlus{}, $n=896308$ in \IMDB{}).} 
  \label{fig:res2_w_Lap_abst}
 \vspace{1mm}
  \centering
  \includegraphics[width=0.99\linewidth]{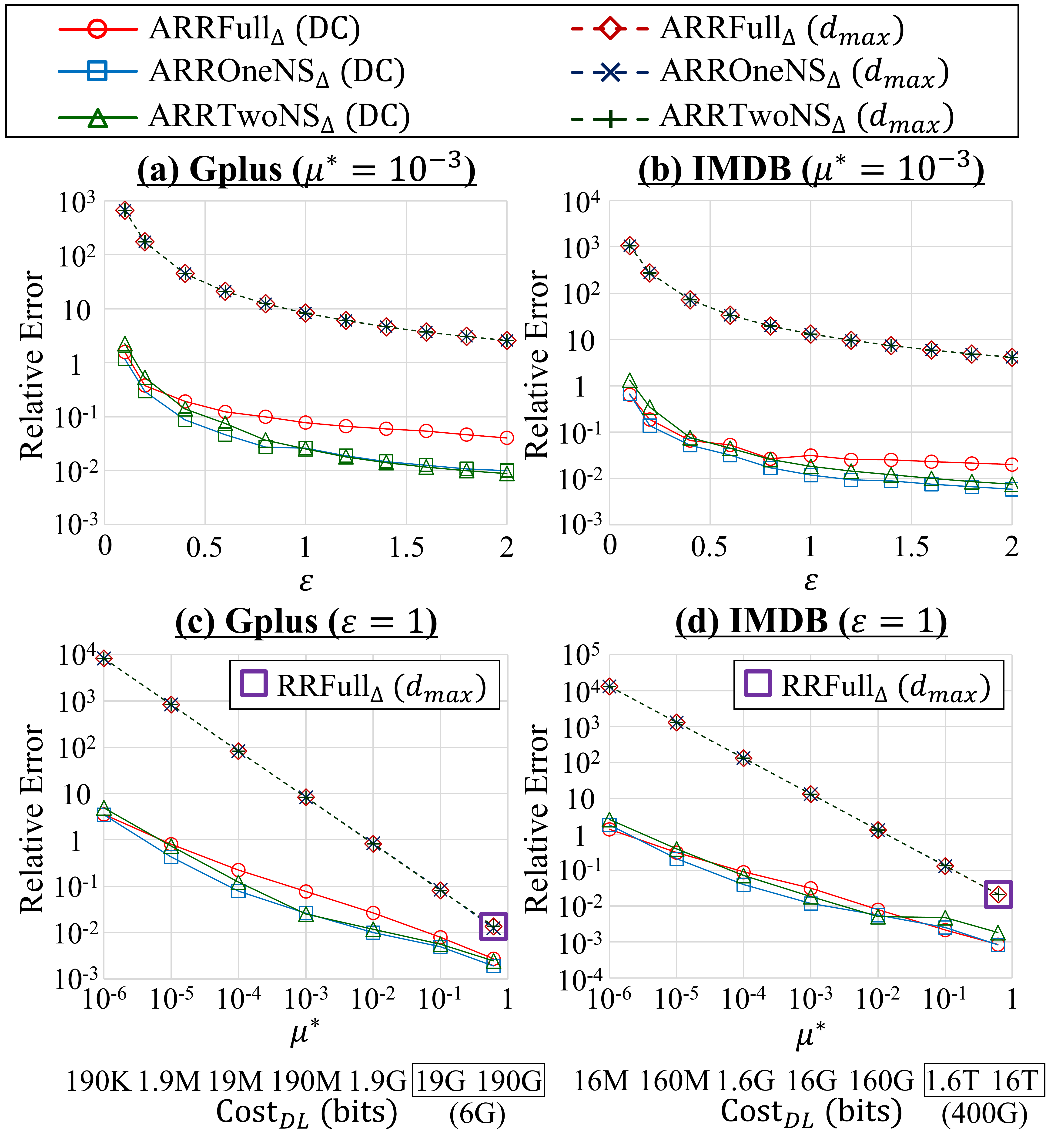}
  \vspace{-5mm}
  \caption{Relative error of our three algorithms with (``DC'') or without (``$d_{max}$'') double clipping ($n=107614$ in \GPlus{}, $n=896308$ in \IMDB{}). \AlgSec{} is the 
  algorithm in~\cite{Imola_USENIX21}. 
  $\CostDL$ is 
  an upper-bound in 
  (\ref{eq:CostDL_F}). 
  When $\mu^* \geq 0.1$, 
  $\CostDL$ can be $6$ Gbits and $400$ Gbits in \GPlus{} and \IMDB{}, respectively, by downloading only 0/1 for each pair of users $(v_j,v_k)$.} 
  \label{fig:res2_w_Lap}
\end{figure}

Figures~\ref{fig:res2_w_Lap_abst} and \ref{fig:res2_w_Lap} show the results. 
Figure~\ref{fig:res2_w_Lap_abst} highlights the relative error of our three algorithms with double clipping when $\epsilon=1$ or $2$ and $\mu^*=10^{-3}$. 
``DC'' (resp.~``$d_{max}$'') represents algorithms with (resp.~without) double clipping. 
\AlgSec{} (marked with purple square) in Figure~\ref{fig:res2_w_Lap} (c) and (d) represents the two-rounds algorithm in~\cite{Imola_USENIX21}. 
Note that this is a special case of our \AlgOne{} without sampling ($\mu =\frac{e^{\epsilon_1}}{e^{\epsilon_1}+1} = 0.62$). 
Figure~\ref{fig:res2_w_Lap} (c) and (d) also show the download cost $\CostDL$ calculated by 
(\ref{eq:CostDL_F}). 
Note that 
when $\mu^* \geq 0.1$ (marked with squares), 
$\CostDL$ can be $6$Gbits and $400$Gbits in \GPlus{} and \IMDB{}, respectively, by downloading only 0/1 for each pair of users $(v_j,v_k)$; $\CostDL = \frac{(n-1)(n-2)}{2}$ in this case. 

Figures~\ref{fig:res2_w_Lap_abst} and \ref{fig:res2_w_Lap} show that our \AlgTwo{} (DC) 
provides the best (or almost the best) performance in all cases. 
This is because \AlgTwo{} (DC) introduces the $4$-cycle trick 
and effectively reduces the global sensitivity of the Laplacian noise by double clipping. 
Later, we will 
investigate 
the effectiveness of the $4$-cycle trick in detail 
by not adding the Laplacian noise. 
We will also investigate 
the impact of the Laplacian noise 
while changing $n$. 

Figure~\ref{fig:res2_w_Lap} also shows that 
the relative error is almost the same between our three algorithms without double clipping (``$d_{max}$'') and that it is too large. 
This is because $\Lap(\frac{d_{max}}{\epsilon_2}$) is too large and dominant. 
The relative error is significantly reduced by introducing our double clipping in all cases. 
For example, when $\mu^* = 10^{-3}$, our double clipping reduces the relative error of \AlgTwo{} by two or three orders of magnitude. 
The improvement is larger for smaller $\mu^*$. 

In \conference{the full version \cite{Imola_arXiv22}}\arxiv{Appendix~\ref{sec:EC_DC}}, we also evaluate the effect of edge clipping and noisy triangle clipping independently and show that each component significantly reduces the relative error. 

\smallskip
\noindent{\textbf{Communication Cost.}}~~From Figure~\ref{fig:res2_w_Lap} (c) and (d), we can explain how much our algorithms can reduce the download cost while keeping high utility, e.g., relative error $\ll 1$. 

For example, when we use the algorithm in \cite{Imola_USENIX21}, the download cost is 
$\CostDL = 400$ Gbits in \IMDB{}. 
Thus, when the download speed is $20$ Mbps 
(recommended speed in YouTube \cite{YouTube_speed}), every user $v_i$ needs 6 hours to download the message $M_i$, which is far from practical. 
In contrast, our \AlgTwo{} (DC) can reduce it to 
$160$ 
Mbits (8 seconds when $20$ Mbps download rate) or less 
while keeping relative error $= 0.21$, 
which is practical and a dramatic improvement over \cite{Imola_USENIX21}. 

We also note that since $d_{max} \ll n$ in \IMDB{}, 
$\CostDL$ of our \AlgTwo{} (DC) 
can also be roughly approximated by $60$ Mbits (3 seconds) by replacing $\mu$ with $\mu e^{-\epsilon_1}$ in 
(\ref{eq:CostDL_F}).

\smallskip
\noindent{\textbf{4-Cycle Trick.}}~~We 
also investigated 
the effectiveness of our $4$-cycle trick in \AlgTwo{} and \AlgThree{} 
in detail. 
To this end, we evaluated our three algorithms when we did \textit{not} add the Laplacian noise at the second round. 
Note that they do not provide edge LDP, as $\epsilon_2 = \infty$. 
The purpose here is to purely investigate the effectiveness of the $4$-cycle trick 
related to our first research question RQ1. 

\begin{figure}[t]
  \centering
  \includegraphics[width=0.99\linewidth]{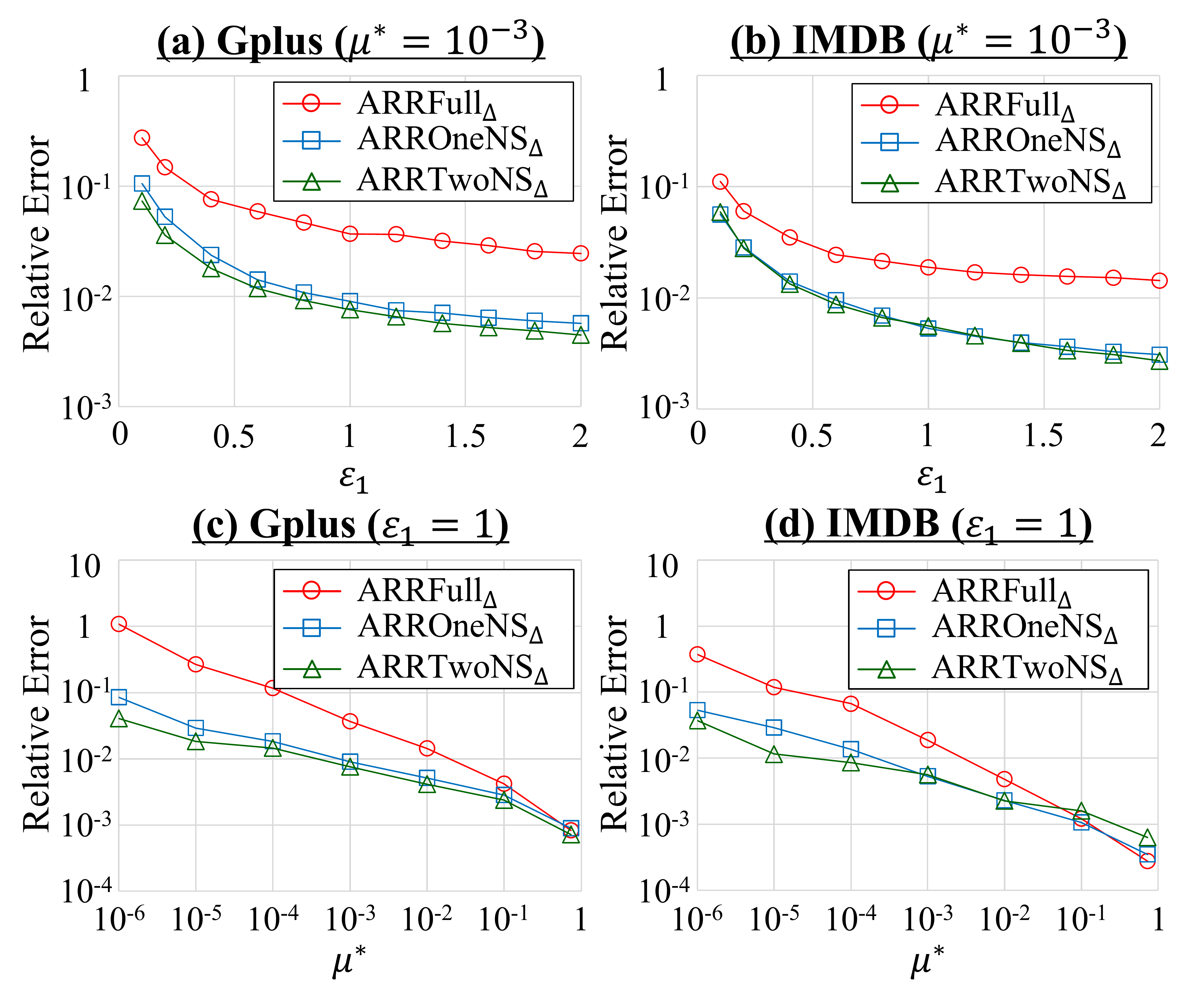}
  \vspace{-5mm}
  \caption{Relative error of our three algorithms without the Laplacian noise 
  ($n=107614$ in \GPlus{}, $n=896308$ in \IMDB{}).} 
  \label{fig:res1_wo_Lap}
\end{figure}

Figure~\ref{fig:res1_wo_Lap} shows the results, where 
$\epsilon_1$ and 
$\mu^*$ 
are changed to various values. 
Figure~\ref{fig:res1_wo_Lap} shows that \AlgTwo{} and \AlgThree{} significantly outperform \AlgOne{} when 
$\mu^*$ 
is small. 
This is because in both \AlgTwo{} and \AlgThree{}, 
the factors of 
$C_4$ (\#4-cycles) and $S_3$ (\#3-stars) 
in the expected $l_2$ loss 
diminish 
for small $\mu$, 
as explained in Section~\ref{sub:algorithms_theoretical_analysis}. 
In other words, \AlgTwo{} and \AlgThree{} effectively address the $4$-cycle issue. 
Figure~\ref{fig:res1_wo_Lap} also shows that \AlgThree{} slightly outperforms \AlgTwo{} when $\mu^*$ is small. 
This is because 
the factors of $C_4$ and $S_3$ 
diminish 
more rapidly; i.e., \AlgThree{} addresses the $4$-cycle issue more aggressively. 

However, when we add the Laplacian noise, \AlgThree{} (DC) is outperformed by \AlgTwo{} (DC), as shown in Figure~\ref{fig:res2_w_Lap}. 
This is because \AlgThree{} cannot effectively reduce the global sensitivity by double clipping. 
In Figure~\ref{fig:res2_w_Lap}, the difference between \AlgTwo{} (DC) and \AlgOne{} (DC) is also small for very small $\epsilon$ or $\mu^*$ (e.g., $\epsilon=0.1$, $\mu^*=10^{-6}$) because the Laplacian noise is dominant in this case. 

\smallskip
\noindent{\textbf{Changing $\bm{n}$.}}~~We 
finally 
evaluated our three algorithms with double clipping while changing the number $n$ of users. 
In both \GPlus{} and \IMDB{}, we randomly selected $n$ users from all users and extracted a graph with $n$ users. 
Then we evaluated the relative error while changing $n$ to various values.

\begin{figure}[t]
  \centering
  \includegraphics[width=0.99\linewidth]{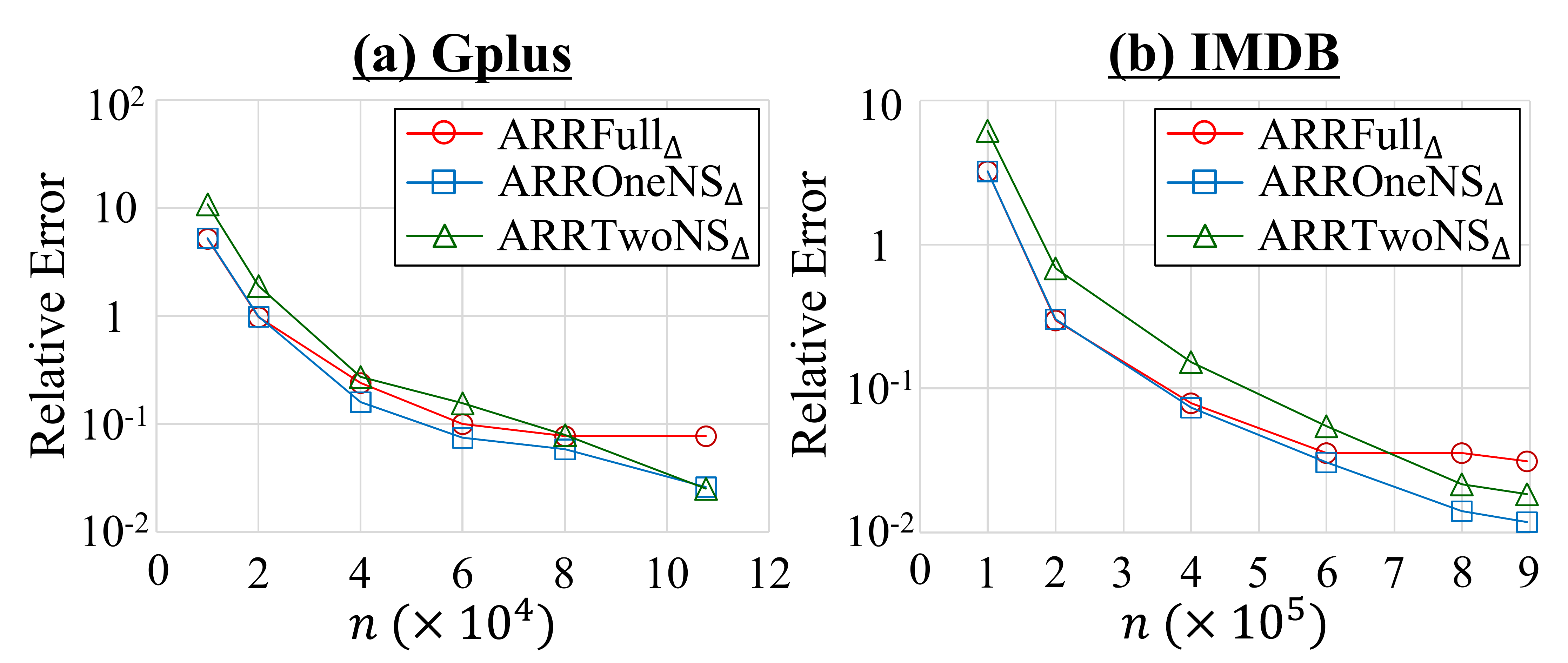}
  \vspace{-4mm}
  \caption{Relative error of our three algorithms with double clipping for various values of $n$ 
  ($\epsilon=1$, $\mu^*=10^{-3}$).} 
  \label{fig:res3_n}
\end{figure}

Figure~\ref{fig:res3_n} shows the results, where $\epsilon=1$ ($\epsilon_0=0.1$, $\epsilon_1 = \epsilon_2 = 0.45$) and $\mu^* = 10^{-3}$. 
In all 
three algorithms, the relative error decreases with increase in $n$.
This is because the expected $l_2$ loss can be expressed as 
$O(n d_{max}^3)$ or $O(n d_{max}^2)$ 
in these algorithms as shown in Section~\ref{sub:clip_theoretical_analysis} and the square of the true triangle count can be expressed as $\Omega(n^2)$.
In other words, when $d_{max} \ll n$, the relative error becomes smaller for larger $n$. 
Figure~\ref{fig:res3_n} also shows that for small $n$, \AlgThree{} provides the worst performance and 
\AlgTwo{} performs almost the same as \AlgOne{}. 
For large $n$, 
\AlgOne{} performs the worst and 
\AlgTwo{} performs the best.

To investigate the reason for this, we decomposed the estimation error into two components -- 
the first error caused by empirical estimation and the second error caused by the Laplacian noise. 
Specifically, for each algorithm, 
we evaluated the first error by calculating the relative error when we did not add the Laplacian noise ($\epsilon_1 = 0.45$). 
Then we evaluated the second error by subtracting the first error from the relative error when we used double clipping ($\epsilon_0=0.1$, $\epsilon_1 = \epsilon_2 = 0.45$). 

Figure~\ref{fig:res3_emp_Lap} shows the results for some values of $n$, where ``emp'' represents the first error by empirical estimation and ``Lap'' represents the second error by the Laplacian noise. 
We observe that the second error 
rapidly decreases with increase in $n$. 
In addition, 
the first error of \AlgOne{} is much larger than those of \AlgTwo{} and \AlgThree{} when $n$ is large. 

We also examined 
the number $C_4$ of $4$-cycles as a function of $n$. Figure~\ref{fig:res4_4cycles} shows the results. 
We observe that $C_4$ (which is $O(n d_{max}^3)$) is quartic in $n$; e.g., $C_4$ is increased by $2^4 \approx 10$ and $6^4 \approx 10^3$ when $n$ is multiplied by $2$ and $6$, respectively. 
This is because we randomly selected $n$ users from all users and $d_{max}$ is almost proportional to $n$ (though $d_{max} \ll n$). 

\begin{figure}[t]
  \centering
  \includegraphics[width=0.99\linewidth]{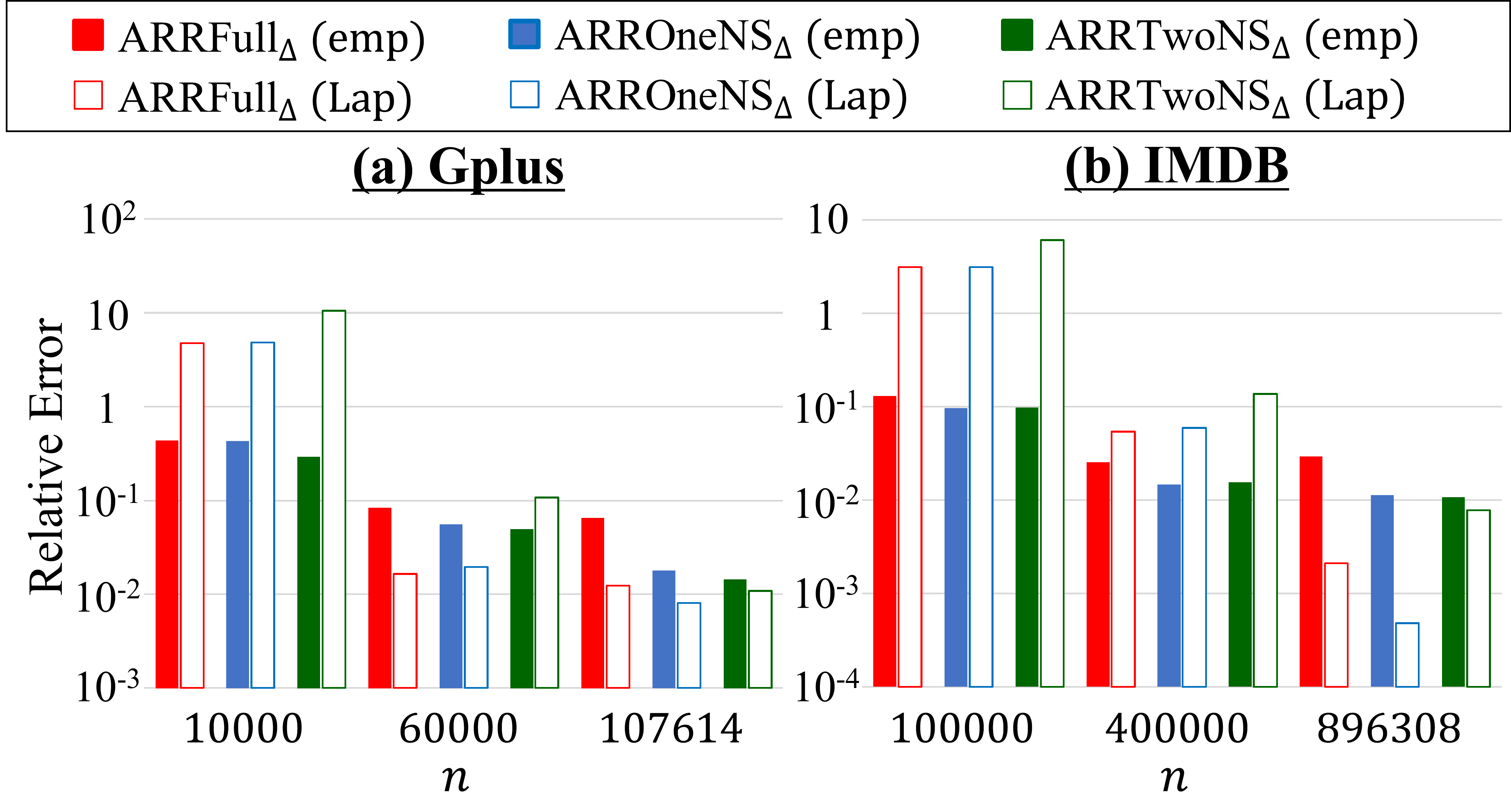}
  \vspace{-5mm}
  \caption{Relative error of empirical estimation and the Laplacian noise in our three algorithms with double clipping ($\epsilon=1$, $\mu^*=10^{-3}$).} 
  \label{fig:res3_emp_Lap}
 \vspace{1mm}
  \centering
  \includegraphics[width=0.99\linewidth]{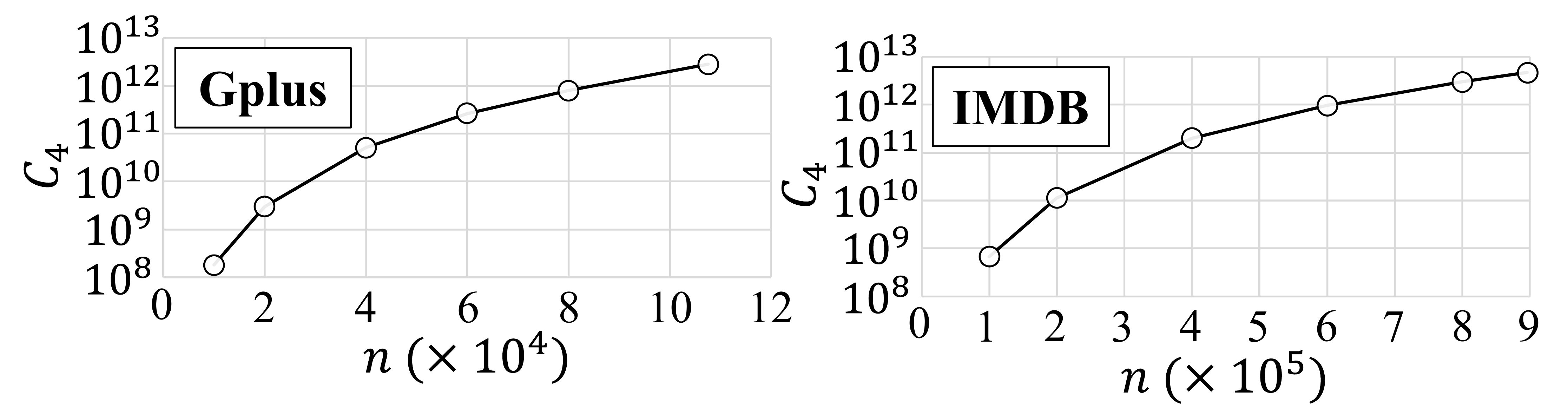}
  \vspace{-5mm}
  \caption{\#$4$-cycles $C_4$.} 
  \label{fig:res4_4cycles}
\end{figure}

Based on Figures~\ref{fig:res3_emp_Lap} and \ref{fig:res4_4cycles}, we can explain 
Figure~\ref{fig:res3_n} as follows. 
As shown in Section~\ref{sub:clip_theoretical_analysis}, the $l_2$ loss of empirical estimation can be expressed as 
$O(n d_{max}^3)$, $O(n d_{max}^2)$, and $O(n d_{max}^2)$ 
in \AlgOne{}, \AlgTwo{}, and \AlgThree{}, respectively. 
The large $l_2$ loss of \AlgOne{} is caused by a large value of $C_4$. 
The expected $l_2$ loss of the Laplacian noise is 
$O(\sum_{i=1}^n \kappa_i^2)$, 
which is much smaller than $O(n d_{max}^2)$. 
Thus, 
as $n$ increases, 
the Laplacian noise becomes relatively very small, 
as shown in Figure~\ref{fig:res3_emp_Lap}. 
Consequently, \AlgTwo{} provides the best performance for large $n$ because it addresses the $4$-cycle issue 
and effectively reduces the global sensitivity. 
This explains the results in Figure~\ref{fig:res3_n}. 
It is also interesting that when $n \approx 10^5$, \AlgOne{} performs the worst in \GPlus{} and almost the same as \AlgTwo{} in \IMDB{} (see Figure~\ref{fig:res3_n}). 
This is because \GPlus{} is more dense than \IMDB{} 
and $C_4$ is much larger in \GPlus{} when $n \approx 10^5$, as in Figure~\ref{fig:res4_4cycles}.

In other words, Figures~\ref{fig:res3_n}, \ref{fig:res3_emp_Lap}, and \ref{fig:res4_4cycles} are consistent with our theoretical results in Section~\ref{sub:clip_theoretical_analysis}. 
From these results, we conclude that \AlgTwo{} is effective especially for 
a large graph (e.g., $n \approx 10^6$) or dense graph (e.g., \GPlus{}) where the number $C_4$ of $4$-cycles is large.

\smallskip
\noindent{\textbf{Summary.}}~~In summary, our answers to our three research questions RQ1-3 
are as follows. 
[RQ1]: Our \AlgTwo{} achieves almost the smallest estimation error in all cases and outperforms the other two, especially for a large graph or dense graph where $C_4$ is large. 
[RQ2]: Our double clipping reduces the estimation error by two or three orders of magnitude. 
[RQ3]: Our entire algorithm (\AlgTwo{} with double clipping) dramatically reduces the communication cost, 
e.g., 
from $6$ hours to $8$ seconds or less (relative error $=0.21$) in \IMDB{} at a $20$ Mbps download rate \cite{YouTube_speed}. 

Thus, triangle counting under edge LDP is now 
much more 
practical. 
In Appendix~\ref{sec:cluster}, we show that the clustering coefficient can also be accurately estimated using our algorithms.

\section{Conclusions}
\label{sec:conclusions}
We proposed
triangle counting algorithms under edge LDP with a small estimation error and small communication cost.
We showed that
our entire algorithms with the 4-cycle trick and double clipping
can dramatically reduce the download cost
of
\cite{Imola_USENIX21}
(e.g., from 6 hours to 8 seconds or less). 

We assumed that each user $v_i$ honestly inputs her neighbor list $\bma_i$, as in most previous work on LDP.
However, recent studies \cite{Cao_USENIX21,Cheu_SP21} show that the estimate in LDP can be skewed by data poisoning attacks.
As future work, we would like to analyze the impact of data poisoning on our algorithms and develop defenses (e.g., detection) against it.

\section*{Acknowledgments}
Kamalika Chaudhuri and Jacob Imola would like to thank ONR under N00014-20-1-2334 and UC Lab Fees under LFR 18-548554  for research support.
Takao Murakami was supported in part by JSPS KAKENHI JP19H04113.
We thank Quentin Hillebrand and Vorapong Suppakitpaisarn for pointing out the sensitivity issue in our double clipping technique, which is described in Appendix~\ref{sec:addendum}.

\bibliographystyle{plain}
\bibliography{main_sshort}

\begin{thebibliography}{10}

\bibitem{TriangleLDP}
Tools: {TriangleLDP}.
\newblock \url{https://github.com/TriangleLDP/TriangleLDP}.

\bibitem{IMDB_GD05}
{12th Annual Graph Drawing Contest}.
\newblock \url{http://mozart.diei.unipg.it/gdcontest/contest2005/index.html},
  2005.

\bibitem{Facebook_Limit}
{What to Do When Your Facebook Profile is Maxed Out on Friends}.
\newblock
  \url{https://authoritypublishing.com/social-media/what-to-do-when-your-facebook-profile-is-maxed-out-on-friends/},
  2012.

\bibitem{Diaspora}
The diaspora* project.
\newblock \url{https://diasporafoundation.org/}, 2021.

\bibitem{Mastodon}
Mastodon: Giving social networking back to you.
\newblock \url{https://joinmastodon.org/}, 2021.

\bibitem{Minds}
Minds: The leading alternative social network.
\newblock \url{https://wefunder.com/minds}, 2021.

\bibitem{YouTube_speed}
{YouTube: System requirements}.
\newblock \url{https://support.google.com/youtube/answer/78358?hl=en}, 2021.

\bibitem{Acharya_AISTATS19}
J.~Acharya, Z.~Sun, and H.~Zhang.
\newblock Hadamard response: Estimating distributions privately, efficiently,
  and with little communication.
\newblock In {\em Proc. AISTATS'19}, pages 1120--1129, 2019.

\bibitem{Andrew_NeurIPS21}
G.~Andrew, O.~Thakkar, H.~B. McMahan, and S.~Ramaswamy.
\newblock Differentially private learning with adaptive clipping.
\newblock In {\em Proc. NeurIPS'21}, pages 1--12, 2021.

\bibitem{Arratia_BMB89}
R.~Arratia and L.~Gordon.
\newblock Tutorial on large deviations for the binomial distribution.
\newblock {\em Bulletin of Mathematical Biology}, 51(1):125--131, 1989.

\bibitem{NetworkScience}
A.~L. Barab\'{a}si.
\newblock {\em Network Science}.
\newblock Cambridge University Press, 2016.

\bibitem{Bassily_NIPS17}
R.~Bassily, K.~Nissim, U.~Stemmer, and A.~Thakurta.
\newblock Practical locally private heavy hitters.
\newblock In {\em Proc. NIPS'17}, pages 2285–--2293, 2017.

\bibitem{Bera_STACS17}
S.~K. Bera and A.~Chakrabarti.
\newblock Towards tighter space bounds for counting triangles and other
  substructures in graph streams.
\newblock In {\em Proc. STACS'17}, pages 11:1--11:14, 2017.

\bibitem{Bera_PODS20}
S.~K. Bera and C.~Seshadhri.
\newblock How the degeneracy helps for triangle counting in graph streams.
\newblock In {\em Proc. PODS'20}, pages 457--467, 2020.

\bibitem{Bera_KDD20}
S.~K. Bera and C.~Seshadhri.
\newblock How to count triangles, without seeing the whole graph.
\newblock In {\em Proc. KDD'20}, pages 306--316, 2020.

\bibitem{Bindschaedler_SP16}
V.~Bindschaedler and R.~Shokri.
\newblock Synthesizing plausible privacy-preserving location traces.
\newblock In {\em Proc. S\&P'16}, pages 546--563, 2016.

\bibitem{Cao_USENIX21}
X.~Cao, J.~Jia, and N.~Z. Gong.
\newblock Data poisoning attacks to local differential privacy protocols.
\newblock In {\em Proc. Usenix Security'21}, pages 947--964, 2021.

\bibitem{CambridgeAnalytica}
R.~Chan.
\newblock The cambridge analytica whistleblower explains how the firm used
  facebook data to sway elections.
\newblock
  \url{https://www.businessinsider.com/cambridge-analytica-whistleblower-christopher-wylie-facebook-data-2019-10},
  2019.

\bibitem{Chen_CCS12}
R.~Chen, G.~Acs, and C.~Castelluccia.
\newblock Differentially private sequential data publication via
  variable-length n-grams.
\newblock In {\em Proc. CCS'12}, pages 638--649, 2012.

\bibitem{Cheu_SP21}
A.~Cheu, A.~Smith, and J.~Ullman.
\newblock Manipulation attacks in local differential privacy.
\newblock In {\em Proc. S\&P'21}, pages 883--900, 2021.

\bibitem{Chu_KDD11}
S.~Chu and J.~Cheng.
\newblock Triangle listing in massive networks and its applications.
\newblock In {\em Proc. KDD'11}, pages 672--680, 2020.

\bibitem{Cyffers_arXiv21}
E.~Cyffers and A.~Bellet.
\newblock Privacy amplification by decentralization.
\newblock {\em CoRR}, 2012.05326, 2021.

\bibitem{Day_SIGMOD16}
W.~Y. Day, N.~Li, and M.~Lyu.
\newblock Publishing graph degree distribution with node differential privacy.
\newblock In {\em Proc. SIGMOD'16}, pages 123--138, 2016.

\bibitem{Ding_TKDE21}
X.~Ding, S.~Sheng, H.~Zhou, X.~Zhang, Z.~Bao, P.~Zhou, and H.~Jin.
\newblock Differentially private triangle counting in large graphs.
\newblock {\em IEEE Transactions on Knowledge and Data Engineering (Early
  Access)}, pages 1--14, 2021.

\bibitem{DP}
C.~Dwork and A.~Roth.
\newblock {\em The Algorithmic Foundations of Differential Privacy}.
\newblock Now Publishers, 2014.

\bibitem{Eden_FOCS15}
T.~Eden, A.~Levi, D.~Ron, and C.~Seshadhri.
\newblock Approximately counting triangles in sublinear time.
\newblock In {\em Proc. FOCS'15}, pages 614--633, 2015.

\bibitem{Erlingsson_CCS14}
U.~Erlingsson, V.~Pihur, and A.~Korolova.
\newblock {RAPPOR}: Randomized aggregatable privacy-preserving ordinal
  response.
\newblock In {\em Proc. CCS'14}, pages 1054--1067, 2014.

\bibitem{Hagberg_SciPy08}
A.~A. Hagberg, D.~A. Schult, and P.~J. Swart.
\newblock Exploring network structure, dynamics, and function using networkx.
\newblock In {\em Proc. SciPy'08}, pages 11--15, 2008.

\bibitem{Hay_ICDM09}
M.~Hay, C.~Li, G.~Miklau, and D.~Jensen.
\newblock Accurate estimation of the degree distribution of private networks.
\newblock In {\em Proc. ICDM'09}, pages 169--178, 2009.

\bibitem{Imola_USENIX21}
J.~Imola, T.~Murakami, and K.~Chaudhuri.
\newblock Locally differentially private analysis of graph statistics.
\newblock In {\em Proc. USENIX Security'21}, pages 983--1000, 2021.

\bibitem{Jorgensen_SIGMOD16}
Z.~Jorgensen, T.~Yu, and G.~Cormode.
\newblock Publishing attributed social graphs with formal privacy guarantees.
\newblock In {\em Proc.SIGMOD'16}, pages 107--122, 2016.

\bibitem{Joseph_SODA20}
M.~Joseph, J.~Mao, and A.~Roth.
\newblock Exponential separations in local differential privacy.
\newblock In {\em Proc. SODA'20}, pages 515--527, 2020.

\bibitem{Kairouz_ICML16}
P.~Kairouz, K.~Bonawitz, and D.~Ramage.
\newblock Discrete distribution estimation under local privacy.
\newblock In {\em Proc. ICML'16}, pages 2436--2444, 2016.

\bibitem{Kairouz_FTML21}
P.~Kairouz, H.~B. McMahan, and B.~Avent \textit{et al.}
\newblock Advances and open problems in federated learning.
\newblock {\em Foundations and Trends in Machine Learning}, 14(1-2):1--210,
  2021.

\bibitem{Kallaugher_PODS19}
J.~Kallaugher, A.~McGregor, E.~Price, and S.~Vorotnikova.
\newblock The complexity of counting cycles in the adjacency list streaming
  model.
\newblock In {\em Proc. PODS'19}, pages 119--133, 2019.

\bibitem{Karwa_PVLDB11}
V.~Karwa, S.~Raskhodnikova, A.~Smith, and G.~Yaroslavtsev.
\newblock Private analysis of graph structure.
\newblock {\em Proceedings of the {VLDB} Endowment}, 4(11):1146--1157, 2011.

\bibitem{Kasiviswanathan_TCC13}
S.~P. Kasiviswanathan, K.~Nissim, S.~Raskhodnikova, and A.~Smith.
\newblock Analyzing graphs with node differential privacy.
\newblock In {\em Proc. TCC'13}, pages 457--476, 2013.

\bibitem{Kolluri_CCS21}
A.~Kolluri, T.~Baluta, and P.~Saxena.
\newblock Private hierarchical clustering in federated networks.
\newblock In {\em Proc. CCS'21}, pages 2342--2360, 2021.

\bibitem{DP_Li}
N.~Li, M.~Lyu, and D.~Su.
\newblock {\em Differential Privacy: From Theory to Practice}.
\newblock Morgan \& Claypool Publishers, 2016.

\bibitem{Manjunath_ESA11}
M.~Manjunath, K.~Mehlhorn, K.~Panagiotou, and H.~Sun.
\newblock Approximate counting of cycles in streams.
\newblock In {\em Proc. ESA'11}, pages 677--688, 2011.

\bibitem{McAuley_NIPS12}
J.~McAuley and J.~Leskovec.
\newblock Learning to discover social circles in ego networks.
\newblock In {\em Proc. NIPS'12}, pages 539--547, 2012.

\bibitem{McGregor_PODS20}
A.~McGregor and S.~Vorotnikova.
\newblock Triangle and four cycle counting in the data stream model.
\newblock In {\em Proc. PODS'20}, pages 445--456, 2020.

\bibitem{data_breach2021}
C.~Morris.
\newblock The number of data breaches in 2021 has already surpassed last year's
  total.
\newblock
  \url{https://fortune.com/2021/10/06/data-breach-2021-2020-total-hacks/},
  2021.

\bibitem{Murakami_USENIX19}
T.~Murakami and Y.~Kawamoto.
\newblock Utility-optimized local differential privacy mechanisms for
  distribution estimation.
\newblock In {\em Proc. USENIX Security'19}, pages 1877--1894, 2019.

\bibitem{Newman_PRL09}
M.~E.~J. Newman.
\newblock Random graphs with clustering.
\newblock {\em Physical Review Letters}, 103(5):058701, 2009.

\bibitem{Nguyen_TDP16}
H.~H. Nguyen, A.~Imine, and M.~Rusinowitch.
\newblock Network structure release under differential privacy.
\newblock {\em Transactions on Data Privacy}, 9(3):215--214, 2016.

\bibitem{Nissim_STOC07}
K.~Nissim, S.~Raskhodnikova, and A.~Smith.
\newblock Smooth sensitivity and sampling in private data analysis.
\newblock In {\em Proc. STOC'07}, pages 75--84, 2007.

\bibitem{Paul_CN14}
T.~Paul, A.~Famulari, and T.~Strufe.
\newblock A survey on decentralized online social networks.
\newblock {\em Computer Networks}, 75:437--452, 2014.

\bibitem{Qin_CCS17}
Z.~Qin, T.~Yu, Y.~Yang, I.~Khalil, X.~Xiao, and K.~Ren.
\newblock Generating synthetic decentralized social graphs with local
  differential privacy.
\newblock In {\em Proc. CCS'17}, pages 425--438, 2017.

\bibitem{Raskhodnikova_arXiv15}
S.~Raskhodnikova and A.~Smith.
\newblock Efficient lipschitz extensions for high-dimensional graph statistics
  and node private degree distributions.
\newblock {\em CoRR}, 1504.07912, 2015.

\bibitem{Raskhodnikova_Encyclopedia16}
S.~Raskhodnikova and A.~Smith.
\newblock {\em Differentially Private Analysis of Graphs}, pages 543--547.
\newblock Springer, 2016.

\bibitem{Robins_SN07}
G.~Robins, P.~Pattison, Y.~Kalish, and D.~Lusher.
\newblock An introduction to exponential random graph ($p^*$) models for social
  networks.
\newblock {\em Social Networks}, 29(2):173--191, 2007.

\bibitem{Sabater_arXiv21}
C.~Sabater, A.~Bellet, and J.~Ramon.
\newblock An accurate, scalable and verifiable protocol for federated
  differentially private averaging.
\newblock {\em CoRR}, 2006.07218, 2021.

\bibitem{Seshadhri_SDM13}
C.~Seshadhri, A.~Pinar, and T.~G. Kolda.
\newblock Triadic measures on graphs: The power of wedge sampling.
\newblock In {\em Proc. SDM'13}, pages 10--18, 2013.

\bibitem{Shokri_CCS15}
R.~Shokri and V.~Shmatikov.
\newblock Privacy-preserving deep learning.
\newblock In {\em Proc. CCS'15}, pages 1310--1321, 2015.

\bibitem{Sun_CCS19}
H.~Sun, X.~Xiao, I.~Khalil, Y.~Yang, Z.~Qui, H.~Wang, and T.~Yu.
\newblock Analyzing subgraph statistics from extended local views with
  decentralized differential privacy.
\newblock In {\em Proc. CCS'19}, pages 703--717, 2019.

\bibitem{Suri_WWW11}
S.~Suri and S.~Vassilvitskii.
\newblock Counting triangles and the curse of the last reducer.
\newblock In {\em Proc. WWW'11}, pages 607--614, 2011.

\bibitem{Tsourakakis_KDD09}
C.~E. Tsourakakis, U.~Kang, G.~L. Miller, and C.~Faloutsos.
\newblock {DOULION}: Counting triangles in massive graphs with a coin.
\newblock In {\em Proc. KDD'09}, pages 837--846, 2009.

\bibitem{Tsourakakis_JGAA11}
C.~E. Tsourakakis, M.~N. Kolountzakis, and G.~L. Miller.
\newblock Triangle sparsifiers.
\newblock {\em Journal of Graph Algorithms and Applications}, 15(6):703--726,
  2011.

\bibitem{Wang_USENIX17}
T.~Wang, J.~Blocki, N.~Li, and S.~Jha.
\newblock Locally differentially private protocols for frequency estimation.
\newblock In {\em Proc. USENIX Security'17}, pages 729--745, 2017.

\bibitem{Warner_JASA65}
S.~L. Warner.
\newblock Randomized response: A survey technique for eliminating evasive
  answer bias.
\newblock {\em Journal of the American Statistical Association},
  60(309):63--69, 1965.

\bibitem{Wu_TKDE16}
B.~Wu, K.~Yi, and Z.~Li.
\newblock Counting triangles in large graphs by random sampling.
\newblock {\em IEEE Transactions on Knowledge and Data Engineering},
  28(8):2013--2026, 2016.

\bibitem{Xiao_SIGMOD11}
X.~Xiao, G.~Bender, M.~Hay, and J.~Gehrke.
\newblock i{R}educt: Differential privacy with reduced relative errors.
\newblock In {\em Proc. SIGMOD'11}, pages 229--240, 2011.

\bibitem{Ye_ICDE20}
Q.~Ye, H.~Hu, M.~H. Au, X.~Meng, and X.~Xiao.
\newblock Towards locally differentially private generic graph metric
  estimation.
\newblock In {\em Proc. ICDE'20}, pages 1922--1925, 2020.

\bibitem{Ye_TKDE21}
Q.~Ye, H.~Hu, M.~H. Au, X.~Meng, and X.~Xiao.
\newblock {LF-GDPR}: A framework for estimating graph metrics with local
  differential privacy.
\newblock {\em IEEE Transactions on Knowledge and Data Engineering (Early
  Access)}, pages 1--16, 2021.

\bibitem{Zhang_USENIX20}
H.~Zhang, S.~Latif, R.~Bassily, and A.~Rountev.
\newblock Differentially-private control-flow node coverage for software usage
  analysis.
\newblock In {\em Proc. USENIX Security'20}, pages 1021--1038, 2020.

\bibitem{Zhang_SIGMOD15}
J.~Zhang, G.~Cormode, C.~M. Procopiuc, D.~Srivastava, and X.~Xiao.
\newblock Private release of graph statistics using ladder functions.
\newblock In {\em Proc. SIGMOD'15}, pages 731--745, 2015.

\end{thebibliography}

\appendix

\section{Basic Notations}
\label{sec:notations_subgraphs}

Table~\ref{tab:notations} shows the basic notations used in this paper.

\begin{table}[t]
\caption{Basic notations.}
\centering
\hbox to\hsize{\hfil
\begin{tabular}{l|l}
\hline
Symbol		&	Description\\
\hline
$G=(V,E)$   &	    Graph with $n$ users $V$ and edges $E$.\\
$v_i$       &       $i$-th user in $V$ (i.e., $V=\{v_1,\ldots,v_n\}$).\\
$d_{max}$   &       Maximum degree of $G$.\\
$\calG$     &       Set of possible graphs with $n$ users.\\
$f_\triangle(G)$   &  Triangle count in $G$.\\
$\bmA=(a_{i,j})$	    &		Adjacency matrix.\\
$\bma_i$	&		Neighbor list of $v_i$ (i.e., $i$-th row of $\bmA$).\\
$\calR_i$     &       Local randomizer of $v_i$.\\
$M_i$     &       Message sent from the server to user $v_i$.\\
$\mu$     &       Parameter in the ARR.\\
$\mu^*$     &   $=\mu, \mu^2, \mu^3$ in \AlgOne{}, \AlgTwo{}, \\
    &   and \AlgThree{}, respectively.\\
$\td_i$     &   Noisy degree of user $v_i$.\\
$\kappa_i$ &   Clipping threshold of user $v_i$.\\
$\epsilon_0$     &       Privacy budget for edge clipping.\\
$\epsilon_1$     &       Privacy budget for the ARR.\\
$\epsilon_2$     &       Privacy budget for the Laplacian noise.\\
$\epsilon$     &       Total privacy budget.\\
\hline
\end{tabular}
\hfil}
\label{tab:notations}
\end{table}

\section{Comparison with One-Round Algorithms}
\label{sec:one-round}

Below we show that one-round triangle counting algorithms suffer from a prohibitively large estimation error.

First, we note that all of the existing one-round triangle algorithms in \cite{Imola_USENIX21,Ye_ICDE20,Ye_TKDE21} are inefficient and \textit{cannot be directly applied to a large-scale graph} such as \GPlus{} and \IMDB{} in Section~\ref{sec:experiments}.
Specifically, in their algorithms, each user $v_i$ applies Warner's RR to each bit of her neighbor list $\bma_i$ and sends the noisy neighbor list to the server.
Then the server counts the number of noisy triangles, each of which has three noisy edges,
and estimates $f_\triangle(G)$ based on the noisy triangle count. 
The noisy graph $G'$ in the server is dense, and there are $O(n^3)$ noisy triangles in $G'$. 
Thus, the time complexity of the existing one-round algorithms \cite{Imola_USENIX21,Ye_ICDE20,Ye_TKDE21} is $O(n^3)$.
It is also reported in \cite{Imola_USENIX21} that when $n=10^6$, the one-round algorithms  would require about $35$ years even using a supercomputer, due to the enormous number of noisy triangles.

Therefore, we evaluated the existing one-round algorithms by taking the following two steps.
First, we evaluate all the existing algorithms in \cite{Imola_USENIX21,Ye_ICDE20,Ye_TKDE21} using small graph datasets ($n=10000$) and show that the algorithm in \cite{Imola_USENIX21} achieves the lowest estimation error.
Second, we improve the time complexity of the algorithm in \cite{Imola_USENIX21} using the ARR (i.e., edge sampling after Warner's RR) and compare it with our two-rounds algorithms using large graph datasets in Section~\ref{sec:experiments}.

\smallskip
\noindent{\textbf{Small Datasets.}}~~We first evaluated the existing algorithms in \cite{Imola_USENIX21,Ye_ICDE20,Ye_TKDE21} using small datasets.
For both \GPlus{} and \IMDB{} in Section~\ref{sec:experiments}, we first randomly selected $n=10000$ users from all users and extracted a graph with $n$ users.
Then we evaluated the relative error of the following three algorithms:
(i) \textsf{RR (biased)} \cite{Imola_USENIX21,Ye_ICDE20},
(ii) \textsf{RR (bias-reduced)} \cite{Ye_TKDE21}, and
(iii) \textsf{RR (unbiased)} \cite{Imola_USENIX21}.
All of them provide $\epsilon$-edge LDP.

\textsf{RR (biased)} simply uses the number of noisy triangles in the noisy graph $G'$ obtained by Warner's RR
as an estimate of $f_\triangle(G)$.
Clearly, it suffers from a very large bias, as $G'$ is dense.
\textsf{RR (bias-reduced)} reduces
this bias
by using a noisy degree sent by each user.
However, it introduces some approximation to estimate $f_\triangle(G)$, and consequently, it is not clear whether the estimate is unbiased.
We used the mean of the noisy degrees as a representative degree to obtain the optimal privacy budget allocation (see \cite{Ye_TKDE21} for details).
\textsf{RR (unbiased)} calculates an unbiased estimate of $f_\triangle(G)$ via empirical estimation. It is proved that the estimate is unbiased \cite{Imola_USENIX21}.

In all of the three algorithms, each user $v_i$ obfuscates bits for smaller user IDs in her neighbor list $\bma_i$. 
We averaged the relative error over $10$ runs.

\begin{figure}[t]
  \centering
  \includegraphics[width=0.99\linewidth]{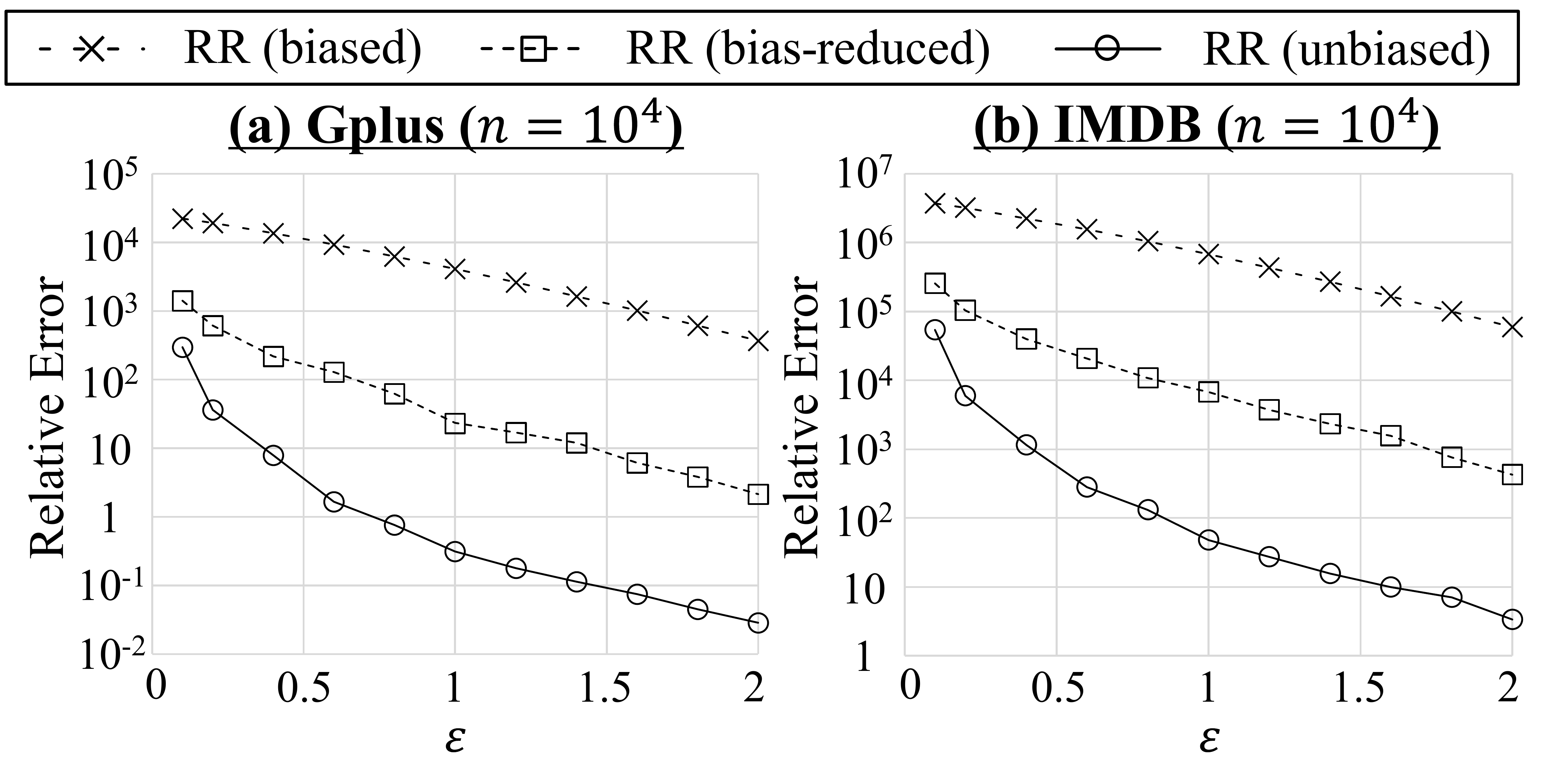}
  \vspace{-4mm}
  \caption{Relative error of one-round algorithms for small datasets ($n=10000$).}
  \label{fig:resB_small}
\end{figure}

Figure~\ref{fig:resB_small} shows the results.
\textsf{RR (bias-reduced)} significantly outperforms \textsf{RR (biased)} and is significantly outperformed by \textsf{RR (unbiased)}.
We believe 
this is caused by the fact that \textsf{RR (bias-reduced)} introduces some approximation and does not calculate an unbiased estimate of $f_\triangle(G)$.

\smallskip
\noindent{\textbf{Large Datasets.}}~~Based on Figure~\ref{fig:resB_small}, we improve the time complexity of \textsf{RR (unbiased)} using the ARR and compare it with our two-rounds algorithms in large datasets.

Specifically, \textsf{RR (unbiased)} counts \textit{triangles}, \textit{$2$-edges} (three nodes with two edges), \textit{$1$-edges} (three nodes with one edge), and \textit{no-edges} (three nodes with no edges) in $G'$ obtained by Warner's RR.
Let $m_3, m_2, m_1, m_0 \in \nnints$ be the numbers of triangles, $2$-edges, $1$-edges, and no-edges, respectively, after applying Warner's RR.
\textsf{RR (unbiased)} calculates an unbiased estimate of $f_\triangle(G)$ from these four values.
Thus, we improve \textsf{RR (unbiased)} by using the ARR, which samples each edge with probability $p_2$ after Warner's RR, and then calculating unbiased estimates of $m_3$, $m_2$, $m_1$, and $m_0$.

Let $\hat{m}_3, \hat{m}_2, \hat{m}_1, \hat{m}_0 \in \reals$ be the unbiased estimates of $m_3$, $m_2$, $m_1$, and $m_0$, respectively. 
Let $m_3^*, m_2^*, m_1^*, m_0^* \in \nnints$ be the number of triangles, 2-edges, 1-edges, no-edges, respectively, after applying the ARR.
Since the ARR samples each edge with probability $p_2$, we obtain:
\begin{align*}
    m_3^* &= \textstyle{p_2^3 \hat{m}_3} \\
    m_2^* &= \textstyle{3p_2^2(1-p_2) \hat{m}_3 + p_2^2 \hat{m}_2} \\
    m_1^* &= \textstyle{3p_2(1-p_2)^2 \hat{m}_3 + 2p_2(1-p_2) \hat{m}_2 + p_2 \hat{m}_1.}
\end{align*}
By these equations,
we obtain:
\begin{align}
    \hat{m}_3 &= \textstyle{\frac{m_3^*}{p_2^3}} \label{eq:hm_3} \\
    \hat{m}_2 &= \textstyle{\frac{m_2^*}{p_2^2} - 3(1-p_2)\hat{m}_3} \label{eq:hm_2} \\
    \hat{m}_1 &= \textstyle{\frac{m_1^*}{p_2} - 3(1-p_2)^2\hat{m}_3 - 2(1-p_2)\hat{m}_2} \label{eq:hm_1} \\
    \hat{m}_0 &= \textstyle{\frac{n(n-1)(n-2)}{6} - \hat{m}_3 - \hat{m}_2 - \hat{m}_1.} \label{eq:hm_0}
\end{align}
Therefore, after applying the ARR to the lower triangular part of $\bmA$, the server counts $m_3^*$, $m_2^*$, $m_1^*$, and $m_0^*$ in $G'$, and then calculates the unbiased estimates $\hat{m}_3$, $\hat{m}_2$, $\hat{m}_1$, and $\hat{m}_0$ by (\ref{eq:hm_3}), (\ref{eq:hm_2}), (\ref{eq:hm_1}), and (\ref{eq:hm_0}), respectively.
Finally, the server estimates $f_\triangle(G)$ from $\hat{m}_3$, $\hat{m}_2$, $\hat{m}_1$, and $\hat{m}_0$ in the same way as \textsf{RR (unbiased)}.
We denote this algorithm by \textsf{ARR (unbiased)}.
The time complexity of \textsf{ARR (unbiased)} is $O(\mu^3 n^3)$, where $\mu$ is the ARR parameter.

We compared \textsf{ARR (unbiased)} with our three algorithms with double clipping using \GPlus{} ($n=107614$) and \IMDB{} ($n=896308$).
For the sampling probability $p_2$, we set $p_2 = 10^{-3}$ or $10^{-6}$.
We averaged the relative error over $10$ runs.

\begin{figure}[t]
  \centering
  \includegraphics[width=0.99\linewidth]{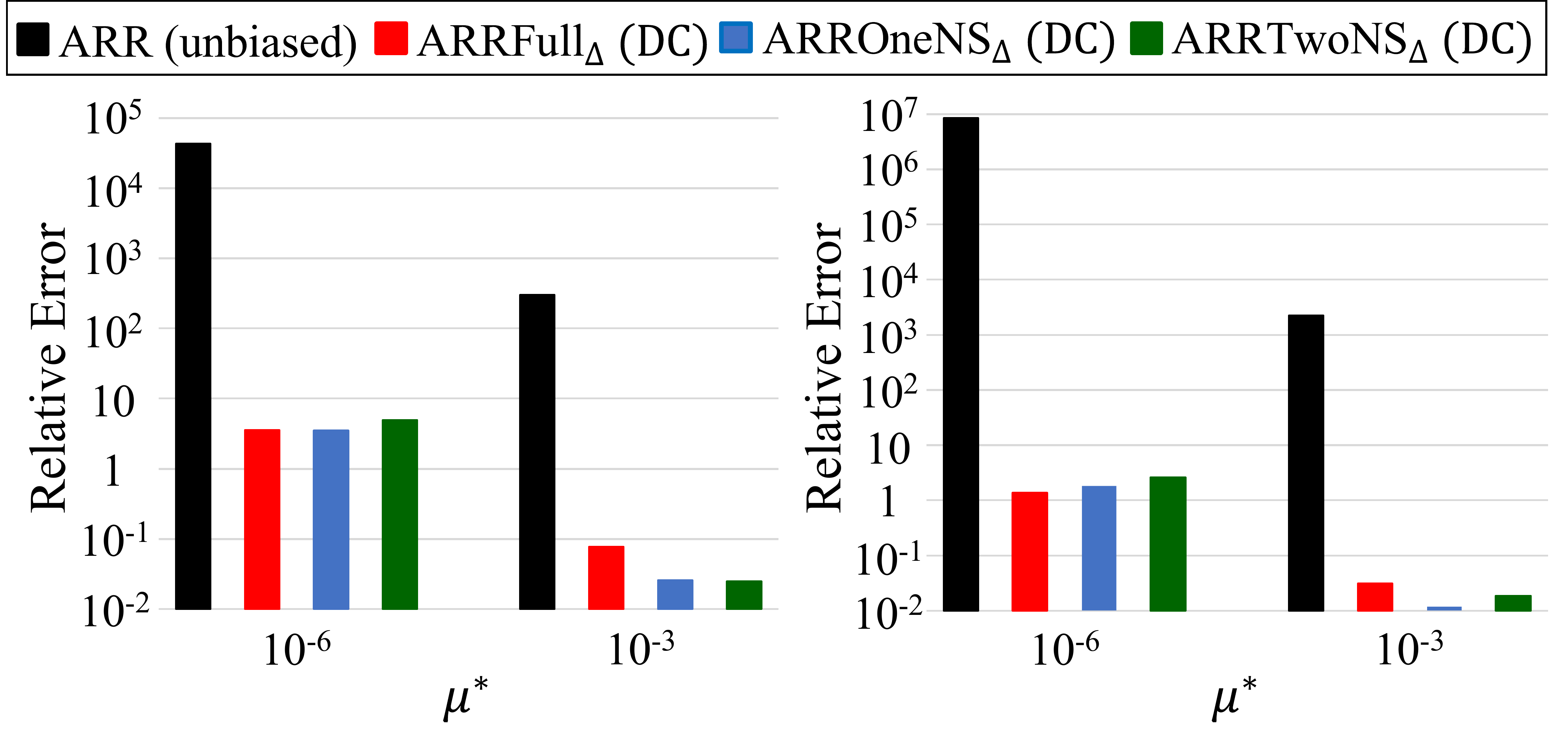}
  \vspace{-4mm}
  \caption{Relative error of the one-round algorithm \textsf{ARR (unbiased)} and our three two-rounds algorithms with double clipping for large datasets ($n=107614$ in \GPlus{}, $n=896308$ in \IMDB{}).}
  \label{fig:resB_large}
\end{figure}

Figure~\ref{fig:resB_large} shows the results, where we set
$\mu^* = 10^{-6}$ or $10^{-3}$.
In \textsf{ARR (unbiased)}, we used $\mu^*$ as 
the ARR parameter $\mu$. 
Thus, we can see \textit{how much the relative error is reduced by introducing an additional round with \AlgOne{}}.
Figure~\ref{fig:resB_large} shows that the relative error of \textsf{ARR (unbiased)} is prohibitively large; i.e., relative error $\gg 1$.
This is because three edges are noisy in any noisy triangle. 
The relative error is significantly reduced by  introducing an additional round
because only one edge is noisy in each noisy triangle at the second round.

In summary, one-round algorithms are far from acceptable in terms of the estimation error for large graphs, and two-round algorithms such as ours are necessary.

\section{Clustering Coefficient}
\label{sec:cluster}
Here we
evaluate the estimation error of the clustering coefficient using our algorithms.

We first estimated a triangle count by using our \AlgTwo{} with double clipping
($\epsilon_0 = \frac{\epsilon}{10}$ and
$\epsilon_1 = \epsilon_2 = \frac{9\epsilon}{20}$)
because it provides the best performance in Figures~\ref{fig:res2_w_Lap_abst}, \ref{fig:res2_w_Lap}, and \ref{fig:res3_n}.
Then we estimated a $2$-star count by
using the one-round $2$-star algorithm in~\cite{Imola_USENIX21} with the
edge clipping in Section~\ref{sec:double_clip}.

Specifically, we calculated a noisy degree $\td_i$ of each user $v_i$ 
by using the edge clipping with the privacy budget $\epsilon_0$.
Then we calculated the number $r_i \in \nnints$ of $2$-stars of which user $v_i$ is a center, and added $\Lap(\frac{\Delta}{\epsilon_1})$ to $r_i$, where $\Delta = \binom{\td_i}{2}$.
Let $\hr_i = r_i + \Lap(\frac{\Delta}{\epsilon_1})$ be the noisy $2$-star of $v_i$.
Finally, we calculated
the sum $\sum_{i=1}^n \hr_i$ as an estimate of the $2$-star count.
This $2$-star algorithm provides ($\epsilon_0 + \epsilon_1$)-edge privacy (see~\cite{Imola_USENIX21} for details).
For the privacy budgets $\epsilon_0$ and $\epsilon_1$, we set $\epsilon_0 = \frac{\epsilon}{10}$ and $\epsilon_1 = \frac{9\epsilon}{10}$.

Based on the triangle and $2$-star counts, we estimated the clustering coefficient as
$\frac{3 \times \hf_\triangle(G)}{\hf_{2\star}(G)}$,
where $\hf_\triangle(G)$ (resp.~$\hf_{2\star}(G)$) is the estimate of the triangle (resp.~$2$-star) count.

\begin{figure}[t]
  \centering
  \includegraphics[width=0.99\linewidth]{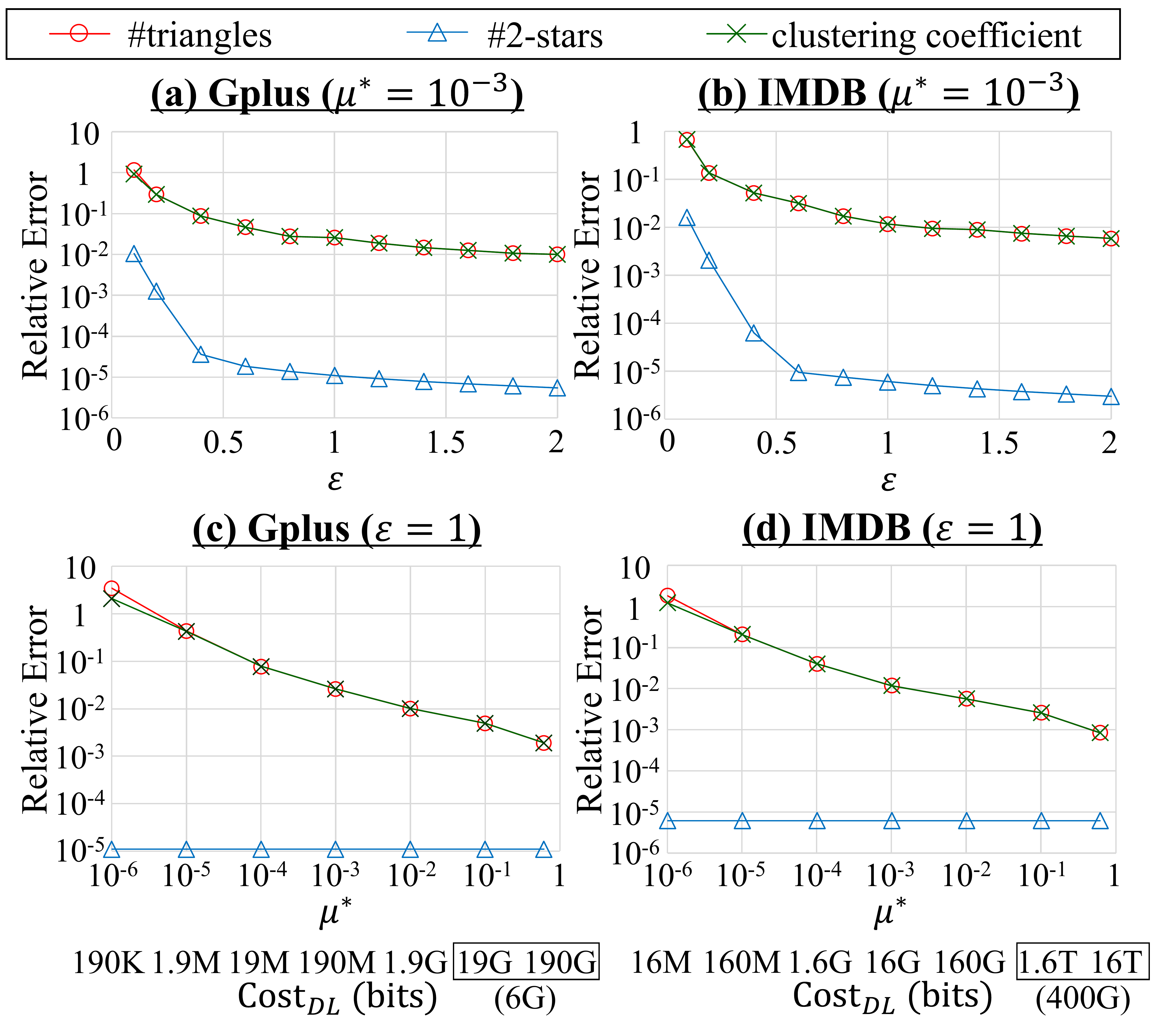}
  \vspace{-4mm}
  \caption{Relative errors of \#triangles, \#$2$-stars, and the clustering coefficient in \AlgTwo{} with double clipping.
  $\CostDL$ is calculated by (\ref{eq:CostDL_F})
  (when $\mu^* \geq 0.1$,
  $\CostDL$ can be $6$ Gbits and $400$ Gbits in \GPlus{} and \IMDB{}, respectively).
  }
  \label{fig:res5_cluster}
\end{figure}

Figure~\ref{fig:res5_cluster} shows the relative errors of the triangle count, $2$-star count, and clustering coefficient.
Note that the relative error of the 2-star count is not changed by changing
$\mu^*$
because the 2-star algorithm does not use the ARR.
Figure~\ref{fig:res5_cluster} shows that the relative error of the $2$-star count is much smaller than that of the triangle count.
This is because each user can count her 2-stars locally (whereas she cannot count her triangles), as described in Section~\ref{sec:intro}.
Consequently, the relative error of the clustering coefficient is almost the same as that of the triangle count, as the denominator $\hf_{2\star}(G)$ in the clustering coefficient is very accurate.

Note that the clustering coefficient requires the privacy budgets for
calculating both $\hf_\triangle(G)$ and $\hf_{2\star}(G)$
(in Figure~\ref{fig:res5_cluster}, $2\epsilon$ in total).
However, we can accurately
calculate $\hf_{2\star}(G)$
with a very small privacy budget, as shown in Figure~\ref{fig:res5_cluster}.
Thus, we can accurately estimate the clustering coefficient with almost the same privacy budget as
the triangle count
by assigning a very small privacy budget (e.g., $\epsilon=0.1$ or $0.2$) for
$\hf_{2\star}(G)$.

In summary, we can accurately estimate the clustering coefficient as well as the triangle count under edge LDP by using our \AlgTwo{} with double clipping.

\arxiv{
\section{Experiments Using the Barab\'{a}si-Albert Graph Datasets}
\label{sec:BAmodel}
In Section~\ref{sec:experiments}, we evaluated our algorithms using two real datasets.
Below we also evaluate our algorithms using a synthetic graph based on the BA (Barab\'{a}si-Albert) graph model~\cite{NetworkScience}, which has a power-law degree distribution.

In the BA graph model, a graph of $n$ nodes is generated by attaching new nodes one by one.
Each new node is connected to $m \in \nnints$ existing nodes, 
and each edge is connected to an existing node with probability proportional to its degree.
We used NetworkX \cite{Hagberg_SciPy08}, a Python package for complex networks, to generate synthetic graphs based on the BA graph model.

We generated a graph $G=(V,E)$ with the same number of nodes as
\GPlus{}; i.e., $n=107614$ nodes.
For the number $m$ of edges per node, we set $m=50$, $114$, or $500$.
Using these graphs, we evaluated our three algorithms with double clipping.
We set parameters in the same as Section~\ref{sec:experiments}; i.e.,
$\alpha = 150$, $\beta = 10^{-6}$, $\epsilon_0 = \frac{\epsilon}{10}$, and $\epsilon_1 = \epsilon_2 = \frac{9\epsilon}{20}$.
For each algorithm, we averaged the relative error over $10$ runs.

Figure~\ref{fig:resA_BAGraph} shows the results, where $\epsilon=1$ and
$\mu^* = 10^{-3}$.
We observe that \AlgTwo{} significantly outperforms \AlgOne{} and \AlgThree{} when $m=500$, and that \AlgTwo{} performs almost the same as \AlgOne{} when $m=50$ or $114$.

To examine the reason for this, we also decomposed the estimation error into two components (the first error by empirical estimation and the second error by the Laplacian noise) in the same way as Figure~\ref{fig:res3_emp_Lap}.
Figure~\ref{fig:resA_BAGraph_emp_Lap} shows the results.
We also show in Table~\ref{tab:resA_4cycles} the number $C_4$ of $4$-cycles in each BA graph ($m=50$, $114$, or $500$) and \GPlus{}.

From Figure~\ref{fig:resA_BAGraph_emp_Lap} and Table~\ref{tab:resA_4cycles}, we can explain Figure~\ref{fig:resA_BAGraph} as follows.
The BA graphs with $m=50$ and $114$ have a much smaller number $C_4$ of $4$-cycles than \GPlus{}, as shown in Table~\ref{tab:resA_4cycles}.
Consequently,
the Laplacian noise is relatively large and dominant for these two graphs, as shown in Figure~\ref{fig:resA_BAGraph_emp_Lap}.
In particular,
the Laplacian noise is the largest in \AlgThree{} because it cannot effectively reduce the global sensitivity by double clipping, as explained in Section~\ref{sec:double_clip}.
In contrast, the BA graph with $m=500$ has a larger number $C_4$ of $4$-cycles than \GPlus{}, and therefore the Laplacian noise is not dominant (except for \AlgThree{}).
This explains the results in Figure~\ref{fig:resA_BAGraph}.

These results
show that \AlgTwo{} outperforms \AlgOne{} especially when the number $C_4$ of $4$-cycles is large.
As we have shown in Section~\ref{sec:experiments} and Appendix~\ref{sec:BAmodel}, $C_4$ is large in a large graph (e.g., $n \approx 10^6$) or dense graph (e.g., \GPlus{}, BA graph with $m=500$).

\begin{figure}[t]
  \centering
  \includegraphics[width=0.99\linewidth]{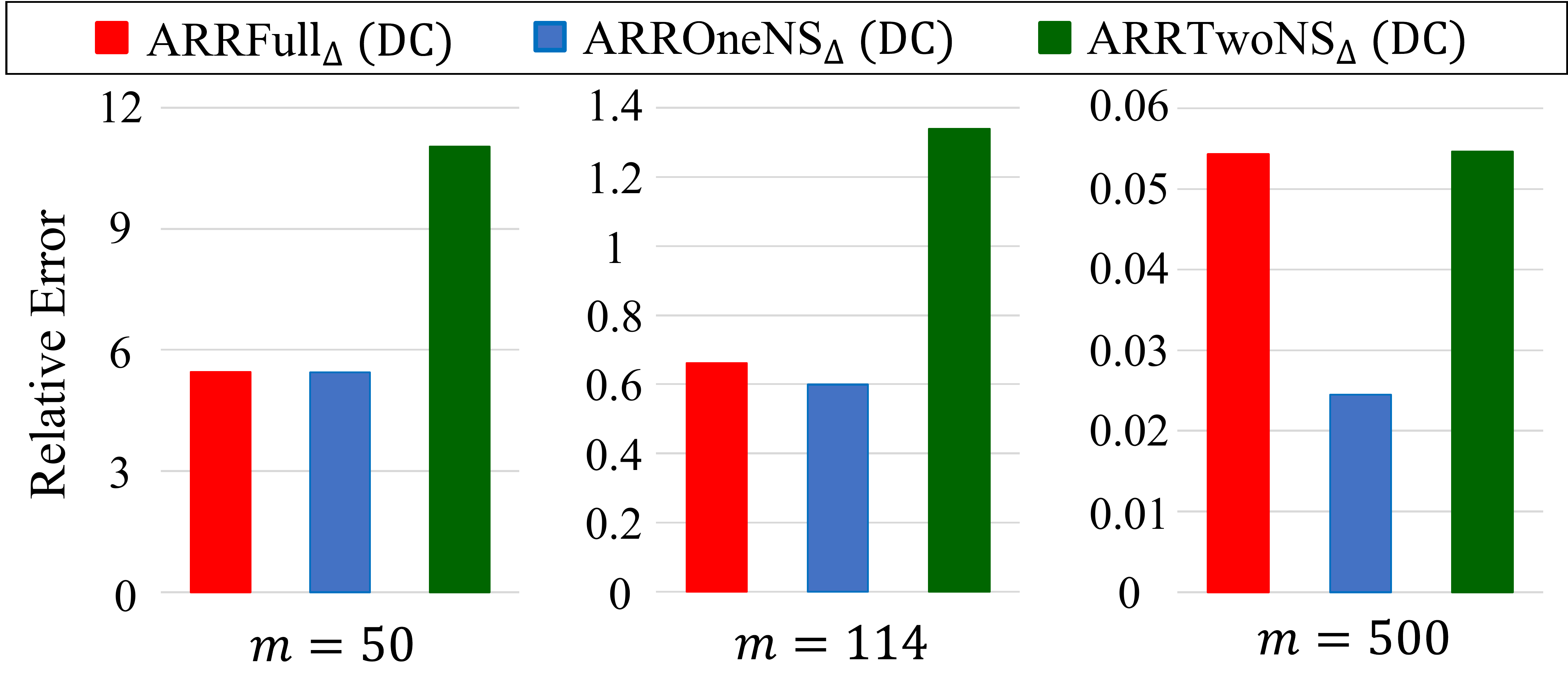}
  \vspace{-4mm}
  \caption{Relative error of our three algorithms with double clipping in the BA graphs ($n=107614$, $\epsilon=1$, $\mu^* = 10^{-3}$).}
  \label{fig:resA_BAGraph}
\vspace{5mm}
  \centering
  \includegraphics[width=0.95\linewidth]{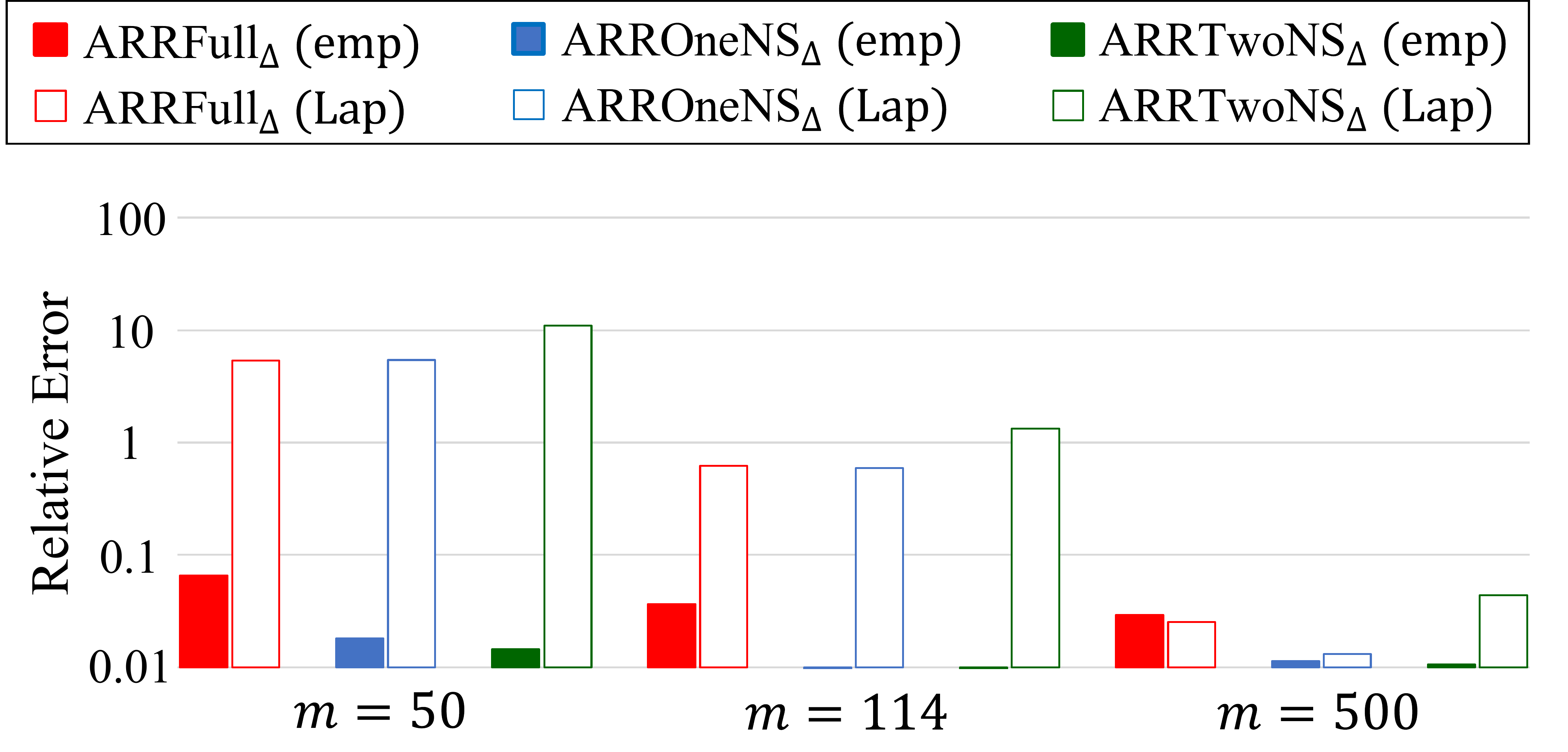}
  \vspace{-4mm}
  \caption{Relative error of empirical estimation and the Laplacian noise in our three algorithms with double clipping in the BA graphs ($n=107614$, $\epsilon=1$, $\mu^* = 10^{-3}$).}
  \label{fig:resA_BAGraph_emp_Lap}
\end{figure}

\section{Edge Clipping and Noisy Triangle Clipping}
\label{sec:EC_DC}
In Section~\ref{sec:experiments}, we showed that our double clipping significantly reduces the estimation error.
To investigate the effect of
edge clipping and noisy triangle clipping independently, we also performed the following ablation study.

We evaluated our three algorithms with only
edge clipping; i.e., each user calculates a noisy degree $\td_i$ (possibly with edge clipping) and then adds $\Lap(\frac{\td_i}{\epsilon_2}$) to her noisy triangle count.
Then we compared them with our algorithms with double clipping and without clipping.

\begin{table}[t]
\caption{\#$4$-cycles $C_4$ in each graph dataset.}
\centering
\hbox to\hsize{\hfil
\begin{tabular}{l|l|l|l|l}
\hline
		&	$m=50$   &  $m=114$   &  $m=500$   &  \GPlus{}\\
\hline
$C_4$   &   $8.8 \times 10^8$ &  $1.7 \times 10^{10}$ &   $3.1 \times 10^{12}$ &   $2.8 \times 10^{12}$\\
\hline
\end{tabular}
\hfil}
\label{tab:resA_4cycles}
\end{table}

\begin{figure}[t]
  \centering
  \includegraphics[width=0.99\linewidth]{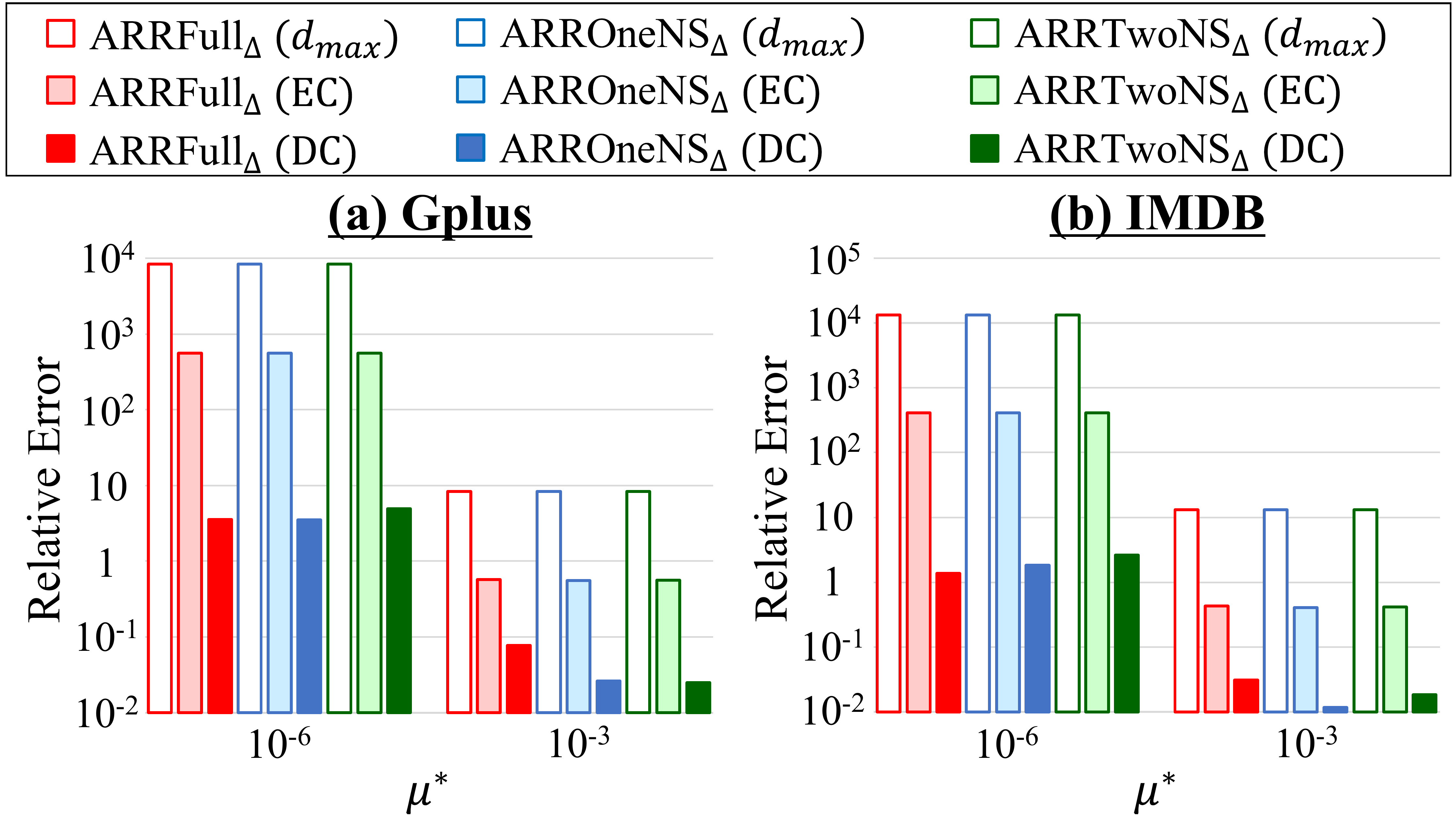}
  \vspace{-4mm}
  \caption{Relative error of our three algorithms without clipping (``$d_{max}$''), with only edge clipping (``EC''), and double clipping (``DC'') when $\epsilon=1$ and $\mu^* = 10^{-6}$ or $10^{-3}$ ($n=107614$ in \GPlus{}, $n=896308$ in \IMDB{}).}
  \label{fig:res2_w_Lap_EC}
\end{figure}

Figure~\ref{fig:res2_w_Lap_EC} shows the results, where $\epsilon=1$, $\mu^* = 10^{-6}$ or $10^{-3}$, and ``EC'' represents our algorithms with only
edge clipping.
We observe that ``EC''
outperforms ``$d_{max}$'' (w/o clipping) and is outperformed by ``DC'' (double clipping).
The difference between ``EC'' and ``DC'' is significant especially when $\mu^* = 10^{-6}$ (``DC'' is smaller than $\frac{1}{100}$ of ``EC'').
This is because our noisy triangle clipping reduces the global sensitivity by using a small value of
$\mu^*$. 
From Figure~\ref{fig:res2_w_Lap_EC}, we conclude that each component (i.e.,
edge clipping, noisy triangle clipping) is essential in our double clipping.

\section{Proof of Proposition~\ref{prop:seq_comp_edge_LDP}}
\label{sec:proof_seq_comp_edge_LDP}
Let $\calR_i(\bma_i) = (\calR_i^1(\bma_i), \calR_i^2(M_i)(\bma_i))$ be the
randomizer used by user $v_i$ in the composition. To establish that
$\calR_i(\bma_i)$ satisfies $\epsilon$-edge LDP for every $v_i \in V$, we will
prove that~\eqref{eq:edge_LDP} holds for $\calR_i(\bma_i)$. To do this, first write
\begin{align*}
  &\;\Pr[(\calR_i^1(\bma_i), \calR_i^2(M_i)(\bma_i)) = (r_i^1, r_i^2)] = \\
  &\qquad\Pr[\calR_i^1(\bma_i) = r_i^1]\Pr[\calR_i^2(M_i)(\bma_i) = r_i^2 | \calR_i^1(\bma_i) = r_i^1] \\
  &\qquad\Pr[\calR_i^1(\bma_i) = r_i^1]\Pr[\calR_i^2(M_i)(\bma_i) = r_i^2 | M_i = \lambda_i(r_i^1)],
\end{align*}
where the last equality follows because $M_i = \lambda_i(\calR_i^1(\bma_i))$ for a post-processing algorithm $\lambda_i$.
Notice that the same equalities are true when we replace $\bma_i$ with $\bma_i'$.
Because $\calR_i^1$ and $\calR_i^2(M_i)$ (for any $M_i$) satisfy $\epsilon_1, \epsilon_2$-edge LDP, respectively,
we have
\begin{align*}
  &\;\Pr[\calR_i^1(\bma_i) = r_i^1]\Pr[\calR_i^2(M_i)(\bma_i) = r_i^2 | M_i = \lambda_i(r_i^1)] \\
  &\qquad\leq e^{\epsilon_1}\Pr[\calR_i^1(\bma_i') = r_i^1]e^{\epsilon_2}\Pr[\calR_i^2(\bma_i') = r_i^2 | M_i = \lambda_i(r_i^1)] \\
  &\qquad= e^{\epsilon_1 + \epsilon_2} \Pr[(\calR_i^1(\bma_i'), \calR_i^2(M_i)(\bma_i')) = (r_i^1, r_i^2)].
\end{align*}
This establishes the result. \qed

\section{Proof of Statements in Section~\ref{sec:algorithms}}
\label{sec:proof_algorithms}
\subsection{Proof of Theorem~\ref{thm:privacy_algorithms}}
Let $\bma_i, \bma'_i \in \{0,1\}^n$ be two neighbor lists that differ in one bit.
Let $t'_i$, $s'_i$, and $w'_i$ be respectively the values of $t_i$ (line 11 of Algorithm~\ref{alg:unify}), $s_i$ (line 12), and $w_i$ (line 13) when the neighbor list of user $v_i$ is $\bma'_i$.
Let $\Delta w_i = |w'_i - w_i|$.
Then we have $t'_i - t_i \in [0,d_{max}]$ and $s'_i - s_i \in [0,d_{max}]$, and therefore $\Delta w_i = |(t'_i - t_i) - \mu^* \rho(s'_i - s_i)| \leq d_{max}$.

Since we add
$\Lap\left(\frac{d_{max}}{\epsilon_2}\right)$
to $w_i$, the second round provides $\epsilon_2$-edge LDP.
The first round uses $ARR_{\epsilon_1,\mu}$ and provides $\epsilon_1$-edge LDP.
Thus, by sequential composition (Proposition~\ref{prop:seq_comp_edge_LDP}),
Algorithm~\ref{alg:unify} provides ($\epsilon_1+\epsilon_2$)-edge LDP in total.
It also provides ($\epsilon_1+\epsilon_2$)-relationship DP because it uses only the lower-triangular part of $\bmA$ (Proposition~\ref{prop:edge_LDP_entire_edge_LDP}).
\qed

\subsection{Proof of Theorem~\ref{thm:l2loss_algorithms}}
\label{sub:prrof_l2loss_algorithms}
\noindent{\textbf{Unbiased Estimators.}}~~First, we will show that $\hf_\triangle(G)$ satisfies
$\E[\hf_\triangle(G)] = f_\triangle(G)$ for all $G \in \calG$, in \AlgOne{}, \AlgTwo{}, \AlgThree{}.
Regardless of algorithm, we have
\begin{align}
  &\;\E[\hf_\triangle(G)] \nonumber \\
  &= \frac{1}{\mu^*(1-\rho)}\sum_{i=1}^n \E[w_i] \nonumber \\
  &= \frac{1}{\mu^*(1-\rho)}\sum_{i=1}^n \E[t_i - \mu^* \rho s_i] \nonumber \\
  &= \frac{1}{\mu^*(1-\rho)}\sum_{i=1}^n \sum_{\substack{1 \leq j < k < i \leq n \\ a_{i,j} = a_{i,k} = 1}} \E[\textbf{1}_{(v_j, v_k) \in M_i} - \mu^* \rho],
  \label{eq:unbias_1}
\end{align}
where $\rho = e^{-\epsilon_1}$, and the quantites $\mu^*, t_i, s_i$ are defined in
Algorithm~\ref{alg:unify}. Given that $a_{i,j} = a_{i,k} = 1$, we have that
$\Pr[(v_i, v_j) \in E'] = \Pr[(v_i, v_k) \in E'] = \mu$ by definition of ARR.
Furthermore,
$\Pr[(v_j, v_k) \in E'] = \mu$ if $a_{j,k} = 1$, and $\Pr[(v_j, v_k) \in E'] = \mu\rho$ otherwise.
Examining~\eqref{eq:M_i_I},~\eqref{eq:M_i_II}, and~\eqref{eq:M_i_III}, we have
\[
  \Pr[(v_j, v_k) \in M_i] =
  \begin{cases}
    \mu^* & a_{j,k} = 1 \\
    \mu^*\rho & a_{j,k} = 0
  \end{cases}
\]
for all the three algorithms (note that $\mu^* = \mu$, $\mu^2$, and $\mu^3$ in \AlgOne{}, \AlgTwo{}, \AlgThree{}, respectively).
Thus, $\E[\textbf{1}_{(v_j, v_k) \in M_i}] = \mu^* (\rho + (1-\rho) a_{j,k})$
($= \mu^*$ if $a_{i,j}=1$ and $\mu^* \rho$ if $a_{i,j}=0$).
Plugging into~\eqref{eq:unbias_1}, we have
\begin{align*}
  &\;\E[\hf_\triangle(G)] \\
  &=
  \frac{1}{\mu^*(1-\rho)}\sum_{i=1}^n \sum_{\substack{1 \leq j < k < i \leq n \\ a_{i,j} = a_{i,k} =
  1}} \mu^*(\rho + (1-\rho)a_{j,k}) - \mu^* \rho \\
  &= \frac{1}{\mu^*(1-\rho)}\sum_{i=1}^n \sum_{\substack{1 \leq j < k < i \leq n \\ a_{i,j} = a_{i,k} =
  1}} \mu^*(1-\rho)a_{j,k} \\
  &= \sum_{i=1}^n \sum_{\substack{1 \leq j < k < i \leq n \\ a_{i,j} = a_{i,k} =
  1}} a_{j,k} \\
  &= f_\triangle(G).
\end{align*} 
Thus, $\hf_\triangle(G)$ is unbiased. \qed

\smallskip
\noindent{\textbf{$l_2$ Loss of Estimators.}}~~Using bias-variance decomposition, we have
for any graph $G$,
\begin{align*}
  l_2^2(f_\triangle(G), \hf_\triangle(G)) &= \E[(\hf_\triangle(G) - f_\triangle(G))^2] \\
  &= \E[(f_\triangle(G) - \E[\hf_\triangle(G)])^2] + \V[\hf_\triangle(G)] \\
  &= \V[\hf_\triangle(G)],
\end{align*}
where the last step follows because $\hf$ is unbiased.
Since $\hf_\triangle(G) = \frac{1}{\mu^*(1-\rho)} \sum_{i=1}^n \hw_i$, we have
\begin{align}
  &\;\V[\hf_\triangle(G)] \nonumber \\
  &= \frac{1}{(\mu^*)^2(1-\rho)^2}\V\left[\sum_{i=1}^n \hw_i\right] \nonumber \\
  &= \frac{1}{(\mu^*)^2(1-\rho)^2}\V\left[\sum_{i=1}^n w_i + \Lap\left(\frac{d_{max}}{\epsilon_2} \right)\right] \nonumber \\
  &= \frac{1}{(\mu^*)^2(1-\rho)^2}\left(\V\left[\sum_{i=1}^n w_i\right] +
  n\V\left[\Lap\left(\frac{d_{max}}{\epsilon_2}\right)\right]\right) \nonumber \\
  &= \frac{1}{(\mu^*)^2(1-\rho)^2}\left(\V\left[\sum_{i=1}^n w_i\right] +
  2n\frac{d_{max}^2}{\epsilon_2^2}\right), \label{eq:inter_var}
\end{align}
where the
fourth
line follows from independence of the added of Laplace noise.
Now, we will prove bounds on $\V[\sum_{i=1}^nw_i]$ for \AlgOne{}, \AlgTwo{}, and
\AlgThree{}. In the following, we let $S_k(G)$ be the number of $k$-stars in
$G$ and $C_4(G)$ be the number of $4$-cycles in $G$

\smallskip
\noindent{\textbf{Bounding the Variance in \AlgOne{}.}}~~In \AlgOne{}, $M_i$ is defined by \eqref{eq:M_i_I}. Thus we have
\begin{align*}
  \sum_{i=1}^n w_i &= \sum_{i=1}^n \sum_{\substack{1 \leq j < k < i \leq n \\ a_{i,j} = 1, a_{i,k} = 1}} \textbf{1}_{(v_j, v_k) \in M_i} \\
  &= \sum_{1 \leq j < k \leq n} \sum_{k < i \leq n} a_{i,j} a_{i,k} \textbf{1}_{(v_j, v_k) \in E'} \\
  &= \sum_{1 \leq j < k \leq n} \textbf{1}_{(v_j, v_k) \in E'}\sum_{k < i \leq n} a_{i,j} a_{i,k}
\end{align*}
For $j < k$, we introduce the constant $c_{jk} = \sum_{k < i \leq n} a_{i,j} a_{i,k}$. Notice that
for any choice of $j$ and $k$, $\textbf{1}_{(v_j, v_k) \in E'}$ for $1 \leq j \leq k$ are mutually independent,
because all edges in $E'$ are mutually independent. Furthermore, the indicator
$\textbf{1}_{(v_j, v_k) \in E'}$ is a Bernoulli random variable with parameter
either $\mu$ or $\mu\rho$, and in either case, $\V[\textbf{1}_{(v_j, v_k) \in
E'}] \leq \mu$. We have
\begin{align*}
  \V\left[\sum_{i=1}^n w_i\right]
  &= \V\left[\sum_{1 \leq j < k \leq n} \textbf{1}_{(v_j, v_k) \in E'}c_{jk}\right] \\
  &= \sum_{1 \leq j < k \leq n} \V[\textbf{1}_{(v_j, v_k) \in E'}c_{jk}] \\
  &= \sum_{1 \leq j < k \leq n} \mu c_{jk}^2.
\end{align*}
By Lemma~\ref{lem:c_ij_4cycle_2star} (which is shown at the end of Appendix~\ref{sub:prrof_l2loss_algorithms}), we have $\sum_{1 \leq j < k \leq n} c_{jk}^2 \leq 2 C_4(G) + S_2(G)$.
Plugging into~\eqref{eq:inter_var} (and substituting $\mu^* = \mu$), we obtain
\begin{align*}
  \V[\hf_\triangle(G)]
  &= \frac{1}{(1-\rho)^2}\left(\frac{1}{\mu}(2C_4(G) + S_2(G)) +
  2n\frac{d_{max}^2}{\mu^2\epsilon_2^2}\right). \\
\end{align*}
This establishes the result. \qed

\smallskip
\noindent{\textbf{Bounding the Variance in \AlgTwo{}.}}~~In \AlgTwo{}, for a fixed
$v_i \in V$, we have $(v_j, v_k) \in M_i$ if and only if
$j < k < i$, $(v_j, v_k) \in E'$, and $(v_i, v_k) \in E'$ from~\eqref{eq:M_i_II}. Thus,
\begin{align*}
\sum_{i=1}^n w_i &= \sum_{i=1}^n~\sum_{\substack{1 \leq j < k < i \leq n \\
a_{i,j} = 1, a_{i,k} = 1}} \textbf{1}_{(v_j, v_k) \in M_i} \\
&= \sum_{1 \leq j < k < i \leq n}
\textbf{1}_{(v_j, v_k) \in E'} a_{i,j}a_{i,k}\textbf{1}_{(v_i, v_k) \in E'}
\end{align*}

Define the random variable $F_{ijk} = \textbf{1}_{(v_j, v_k) \in
E'} \textbf{1}_{(v_i, v_k) \in E'}$. Substituting, we have
\begin{align*}
  \sum_{i=1}^n w_i &= \sum_{1 \leq j < k < i \leq n} a_{i,j} a_{i,k}F_{ijk} \\
  \V\left[\sum_{i=1}^n w_i\right]
  &= \sum_{\substack{1 \leq j < k < i \leq n \\ 1 \leq j' < k' <
  i' \leq n}} a_{i,j} a_{i,k} a_{i',j'} a_{i',k'} \cov(F_{ijk}, F_{i'j'k'}).
\end{align*}

The set $\{i,i',j,j',k,k'\}$ when $j < k < i$ and $j' < k' < i'$ can take between
three and six distinct values.
If $\{i,i', j,j', k,k'\}$ takes five or more distinct values, then $F_{ijk}$ and
$F_{i'j'k'}$ involve distinct edges and are independent random variables. Thus,
$\cov(F_{ijk}, F_{i'j'k'}) = 0$. Otherwise, the events are not independent, and we will
use the upper bound $\cov(F_{ijk}, F_{i'j'k'}) \leq \E[F_{ijk}F_{i'j'k'}]
=
\Pr[F_{ijk} = F_{i'j'k'} = 1]$, which holds because the $F_{ijk}$ have domain $\{0,1\}$.
Thus,
\begin{align*}
&\;\V\left[\sum_{i=1}^n w_i \right] \\
&\leq
  \sum_{\substack{1 \leq j < k < i \leq n \\ 1 \leq j' < k' <
  i' \leq n \\ |\{i,j,k,i',j',k'\}| = 3 \text{ or } 4}} a_{i,j} a_{i,k} a_{i',j'} a_{i',k'} \Pr[F_{ijk} = F_{i'j'k'} = 1].
\end{align*}

Define a choice of $(i,j,k,i',j',k') \in [n]^6$ to be a \emph{valid} choice if
$j < k < i$, $j' < k' < i'$ (ordering requirement),
$a_{i,k} = a_{i,j} = a_{i',j'} = a_{i', k'} = 1$ (edge requirement), and
$3 \leq \{i,i',j,j',k,k'\} \leq 4$ (size requirement). We can write the above sum as
\[
  \V\left[\sum_{i=1}^n w_i \right] \leq
  \sum_{i,i',j,j',k,k'\text{ valid}} \Pr[F_{ijk} = F_{i'j'k'} = 1].
\]

In the above sum, each valid choice implies there exists
a subgraph of $G$ that associates
each of $\{v_i, v_{i'}, v_j, v_{j'}, v_k, v_{k'}\}$ with one node in the subgraph and contains edges $(v_i, v_j), (v_i, v_k), (v_{i'}, v_{j'}), (v_{i'}, v_{k'})$.
Conversely,
each
subgraph of $G$ of three or four nodes can have a certain
number of valid choices mapped to it. For each
subgraph, we now go over the number of possible valid choices that can map to it:

\begin{enumerate}
    \item 4-cycle: By ordering, either $i$ or $i'$ is mapped to the node of the 4-cycle
    with maximal
    index. WLOG, suppose $i$ is mapped to this node. By edge requirements, $i'$ has an index equal
    to the opposite node in the $4$-cycle. By ordering, there is now one
    way to map $j,j',k,k'$. Thus, each 4-cycle can be associated with
    $\mathbf{2}$ valid choices.
    \item 3-path: Consider the middle node $v_\ell$ in the 3-path (path graph on 4 nodes) that has
    the second-largest index.
    By ordering, either $i=\ell$ or $i' = \ell$.
    WLOG, suppose $i = \ell$. Then, by the edge requirement $a_{i',j'} = a_{i', k'} = 1$, the middle node other than $v_\ell$ is
    $i'$. However, this means either $i = j'$ or $i = k'$, and we have $j' > i'$ or $k' > i'$.
    This violates the order requirement, and therefore there are $\textbf{0}$ valid choices.
    \item 3-star: By edge requirement, both $i,i'$ map to the central node in the 3-star.
    $j,j',k,k'$ can map to the other three nodes in any way that satisfies ordering.
    Suppose the three nodes are $v_{a}, v_{b}, v_c$ with $a < b < c$.
    Only one of
    the three nodes can be duplicated in this mapping.
    For example, if
    $a$ is duplicated, then both
    $j$ and $j'$
    map to $v_a$, and there are two remaining choices
    for how to map
    $k$ and $k'$.
    Thus,
    each $S_3$ can
    be associated with
    $\mathbf{6}$ valid choices.
    \item Triangle: Either $i$ or $i'$ maps to the maximal node in the triangle
    by ordering. WLOG, suppose $i$ does. By the edge requirement, $i'$ maps
    to a different node. However, this means $i = k'$ or $i=j'$, so $k' > i'$,
    contradicting ordering. Thus, there are $\textbf{0}$ valid choices.
    \item 2-star: By edge requirements, both $i$ and $i'$ map to the central node in the
    2-star. We then have just one mapping for the remaining indices. Thus, there
    is $\textbf{1}$ valid choice.
\end{enumerate}
Figure~\ref{fig:associated_subgraphs} shows an example of two 4-cycles, six 3-stars, and one 2-stars.
We can see that the other possible subgraphs on 3 or 4 nodes are immediately
ruled out because they have too many or too few edges, violating edge requirements.

\begin{figure}[t]
  \centering
  \includegraphics[width=0.99\linewidth]{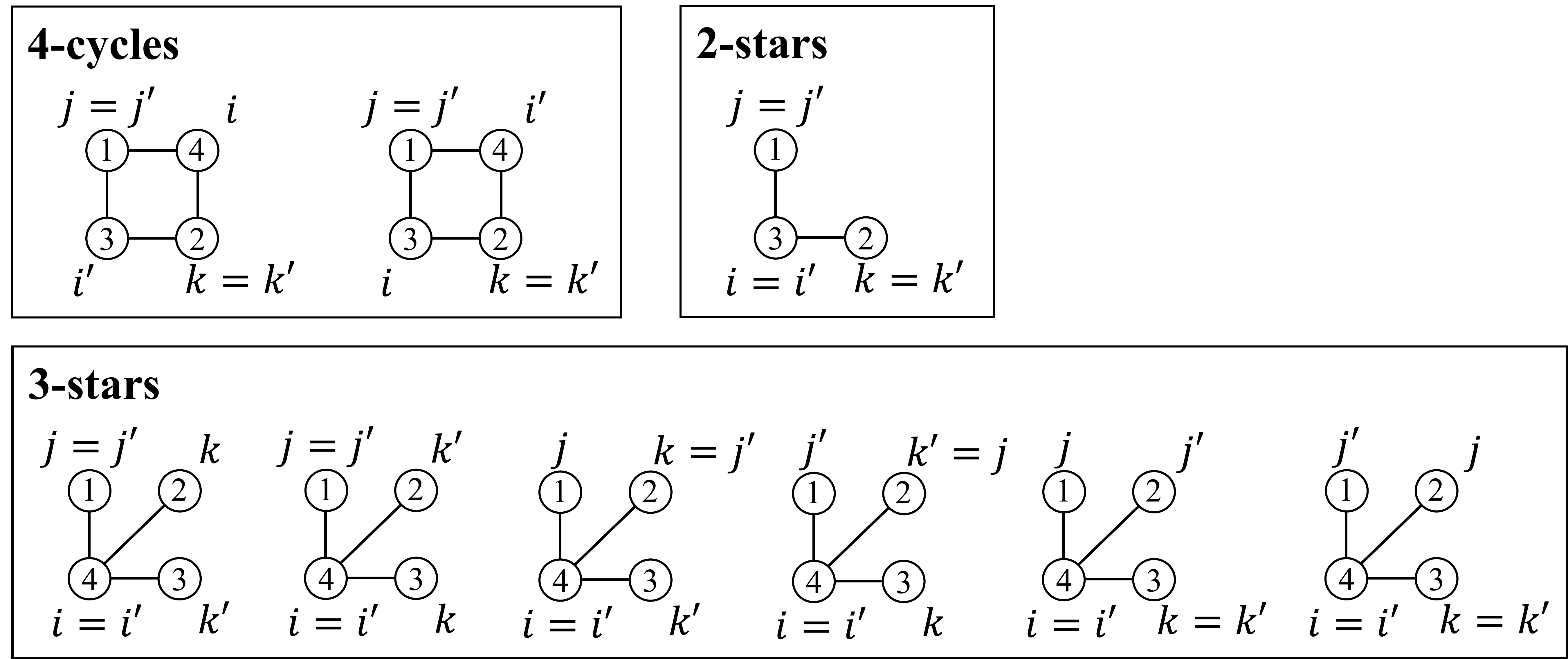}
  \vspace{-4mm}
  \caption{Examples of two 4-cycles, six 3-stars, and one 2-stars.}
  \label{fig:associated_subgraphs}
\end{figure}

In the following, let $P(i,j,k)$ be the event that $F_{ijk} = F_{i'j'k'} = 1$.
Observing Figure~\ref{fig:four-cycle}, we can see that in both possible ways in
which a valid choice maps to a $4$-cycle, then $P(i,j,k)$ holds when at least $3$
edges in $E'$ are present. Each edge in $E'$ is independent,
and is present with probability at most $\mu$. Thus, $\Pr[P(i,j,k)] \leq \mu^3$.
Next, if a valid choice maps to a $3$-star, then $P(i,j,k)$ implies at least $3$
edges in $E'$ are present. Thus, $\Pr[P(i,j,k) = 1] \leq \mu^3$.
Finally, if a valid choice maps to a $2$-star, then $P(i,j,k) = 1$
if and only if $2$ edges in $E'$ are present. Thus, $\Pr[P(i,j,k) = 1] \leq \mu^2$.

Putting this together,
\[
  \V\left[\sum_{i=1}^n w_i\right] \leq 2 C_4(G) \mu^3 + 6 S_3(G) \mu^3 + S_2(G) \mu^2.
\]
Plugging into~\eqref{eq:inter_var}, we get
\begin{align*}
  &\;\V[\hf_\triangle(G)] \\
  &\leq \frac{1}{(1-\rho)^2} \left(\frac{1}{\mu}(2C_4(G) + 6S_3(G)) + \frac{1}{\mu^2} S_2(G) + 2n\frac{d_{max}^2}{\mu^4\epsilon_2^2}\right).
\end{align*}
This establishes the result. \qed

\smallskip
\noindent{\textbf{Bounding the Variance in \AlgThree{}.}}~~In \AlgThree{}, for a fixed
$v_i \in V$, we have $(v_j, v_k) \in M_i$ if and only if
$j < k < i$, $(v_j, v_k) \in E'$, $(v_i, v_k) \in E'$, and $(v_i, v_k) \in E'$
from~\eqref{eq:M_i_III}. Thus,
\begin{align*}
\sum_{i=1}^n w_i &= \sum_{i=1}^n~\sum_{\substack{1 \leq j < k < i \leq n \\
a_{i,j} = 1, a_{i,k} = 1}} \textbf{1}_{(v_j, v_k) \in M_i} \\
&= \sum_{1 \leq j < k < i \leq n}
a_{i,j}a_{i,k}\textbf{1}_{(v_j, v_k) \in E'} \textbf{1}_{(v_i, v_k) \in E'}\textbf{1}_{(v_j, v_k) \in E'}.
\end{align*}

Define the random variable $F_{ijk} = \textbf{1}_{(v_j v_k) \in E'}
\textbf{1}_{(v_i, v_k) \in E'} \allowbreak \textbf{1}_{(v_j, v_k) \in E'}$. Following the same
steps as those in the proof of \AlgTwo{}, we have
\begin{align*}
  \V\left[\sum_{i=1}^n w_i\right]
  &= \sum_{\substack{1 \leq j < k < i \leq n \\ 1 \leq j' < k' <
  i' \leq n}} a_{i,j} a_{i,k} a_{i',j'} a_{i',k'} \cov(F_{ijk}, F_{i'j'k'}) \\
  &\leq
  \sum_{i,i',j,j',k,k'\text{ valid}} \Pr[F_{ijk} = F_{i'j'k'} = 1].
\end{align*}

As we showed in the proof for \AlgTwo{}, each $4$-cycle of $G$ has at most
2 valid choices mapped to it, each $3$-star of $G$ has at most 6 valid
choices mapped to it, and each $2$-star of $G$ has at most one valid choice
mapped to it.

In the following, let $P(i,j,k)$ be the event that $F_{ijk} = F_{i'j'k'} = 1$.
Observing Figure~\ref{fig:four-cycle}, we can see that for each possible
mapping of a valid choice to a $4$-cycle, five edges must be present in $G'$ in
order for $P(i,j,k) = 1$. Thus, $\Pr[P(i,j,k) = 1] \leq \mu^5$. For each possible
mapping of a valid choice to a $3$-star, five edges must be present in $G'$ in
order for $P(i,j,k) = 1$. Thus, $\Pr[P(i,j,k) = 1] \leq \mu^5$. For each possible
mapping of a valid choice to a $2$-star, three edges must be present in $G'$ in
order for $P(i,j,k) = 1$. Thus, $\Pr[P(i,j,k) = 1] \leq \mu^3$.

Plugging into~\eqref{eq:inter_var}, we get
\begin{align*}
  &\;\V[\hf_\triangle(G)] \\
  &\leq \frac{1}{(1-\rho)^2} \left(\frac{1}{\mu}(2C_4(G) + 6S_3(G)) + \frac{1}{\mu^3} S_2(G) + 2n\frac{d_{max}^2}{\mu^6\epsilon_2^2}\right).
\end{align*}
This establishes the result.
\qed

\begin{lemma}
  \label{lem:c_ij_4cycle_2star}
  Let $c_{ij} = \sum_{l < i \leq n} a_{l,i} a_{l,j}$. Then,
\[
    \sum_{i,j=1, i<j}^n c_{ij}^2 \leq 2 C_4(G) + S_2(G).
\]
\end{lemma}
\begin{proof}
  \begin{align*}
      \sum_{i,j=1, i<j}^n c_{ij}^2
      &= \sum_{i,j=1, i<j}^n c_{ij} + \sum_{i,j=1, i<j}^n c_{ij}(c_{ij}-1) \\
      &= S_2(G) + \sum_{i,j=1, i<j}^n c_{ij}(c_{ij}-1).
  \end{align*}
  Let $C_{i-*-j-*-i}(G)$ be the number of 4-cycles in $G$
  such that the first and third nodes are $v_i$ and $v_j$, respectively $(i<j)$, and the remaining two nodes have smaller indices than $i$.
  From middle nodes in 2-paths starting at $v_i$ and ending at $v_j$, we can choose two nodes as the second and fourth nodes in the 4-cycles.
  $c_{ij}$ is the number of nodes that have smaller IDs than $v_i$ and are connected to $v_i$ and $v_j$.
  Thus,
  $C_{i-*-j-*-i}(G) = \binom{c_{ij}}{2}$.
  Therefore, we have
  \begin{align*}
      \sum_{i,j=1, i<j}^n c_{ij}^2
       &= S_2(G) + \sum_{i,j=1, i<j}^n 2C_{i-*-j-*-i}(G) \\
      &\leq S_2(G) + 2 C_4(G).
  \end{align*}
The last inequality comes from the fact that
two nodes with the largest indices may not be opposite to each other in some 4-cycles in $G$.
\end{proof}

\section{Proof of Statements in Section~\ref{sec:double_clip}}
\label{sec:proof_double_clip}
\subsection{Proof of Theorem~\ref{thm:privacy_DC}}
\label{sub:proof_privacy_DC}
Let $\bma_i, \bma'_i \in \{0,1\}^n$ be two neighbor lists that differ in one bit.
Let ${t_i}'$, $s'_i$, and ${w_i}'$ be respectively the values of $t_i$ (in line 8 of Algorithm~\ref{alg:clip}), $s_i$ (in line 9), and $w_i$ (in line 10) when the neighbor list of user $v_i$ is $\bma'_i$.
Let $\Delta w_i = |{w_i}' - w_i|$.

We assume that $|\bma'_i| = |\bma_i| + 1$ without loss of generality.
Let $\bar{\bma}_i, \bar{\bma}'_i \in \{0,1\}^n$ be neighbor lists corresponding to $\bma_i$ and $\bma'_i$, respectively, after edge clipping.
Note that $|\bar{\bma}_i| = |\bar{\bma}'_i| \leq \td_i$.
There are three cases for $\bar{\bma}_i$ and $\bar{\bma}'_i$:
\begin{enumerate}
    \item $\bar{\bma}_i$ is identical to $\bar{\bma}'_i$ and $|\bar{\bma}_i| = |\bar{\bma}'_i| = \td_i$.
    \item $\bar{\bma}_i$ and $\bar{\bma}'_i$ differ in one bit and $|\bar{\bma}'_i| = |\bar{\bma}_i| + 1$.
    \item $\bar{\bma}_i$ and $\bar{\bma}'_i$ differ in two bits and $|\bar{\bma}_i| = |\bar{\bma}'_i| = \td_i$.
\end{enumerate}
Note that the third case can happen when $|\bma'_i| \geq \td_i$.
For example, assume that $n=8$, $\td_i=4$, $\bma_i=(1,1,0,1,0,1,1,1)$ and $\bma'_i=(1,1,1,1,0,1,1,1)$.
If we select four ``1''s in the order of 3, 1, 4, 6, 8, 2, 7, and 5-th bit in the neighbor list,
$\bar{\bma}_i$ and $\bar{\bma}'_i$ will be:
$\bar{\bma}_i=(1,0,0,1,0,1,0,1)$ and $\bar{\bma}'_i=(1,0,1,1,0,1,0,0)$,
which differ in two bits.

If $\bar{\bma}_i$ and $\bar{\bma}'_i$ differ in one bit ($|\bar{\bma}'_i| = |\bar{\bma}_i| + 1$), then ${t_i}' - t_i \in [0,\kappa_i]$ and $s'_i - s_i \in [0,\td_i]$, hence $\Delta w_i = |({t_i}' - t_i) - \mu^*\rho(s'_i - s_i)| \leq \kappa_i$.
If $\bar{\bma}_i$ and $\bar{\bma}'_i$ differ in two bits ($|\bar{\bma}_i| = |\bar{\bma}'_i| = d_{max}$), then ${t_i}' - t_i \in [-\kappa_i,\kappa_i]$ and $s_i = s'_i = \binom{\td_i}{2}$, hence $\Delta w_i \leq \kappa_i$.

Therefore, we always have $\Delta w_i \leq \kappa_i$
(if $\bar{\bma}_i$ is identical to $\bar{\bma}'_i$, $\Delta w_i =0$).
Since we add $\Lap(\frac{1}{\epsilon_0})$
to $d_i$ and
$\Lap(\frac{\kappa_i}{\epsilon_2})$
to $w_i^*$, the second round provides ($\epsilon_0+\epsilon_2$)-edge LDP.
The first round provides $\epsilon_1$-edge LDP and we use only the lower-triangular part of $\bmA$.
Thus, by sequential composition (Proposition~\ref{prop:seq_comp_edge_LDP}) 
and Proposition~\ref{prop:edge_LDP_entire_edge_LDP},
$\calR_i$ satisfies $(\epsilon_0 + \epsilon_1 + \epsilon_2)$-edge LDP, and $(\calR_1, \ldots, \calR_n)$ satisfies $(\epsilon_0 + \epsilon_1 + \epsilon_2)$-relationship DP.
\qed

\subsection{Proof of Theorem~\ref{thm:triangle_excess}}

Recall that
$t_{i,j} = |\{(v_i,v_j,v_k) : a_{i,k} = 1, (v_j,v_k) \in M_i, j<k<i \}|$.
Let
$t'_{i,j} = |\{(v_i,v_j,v_k) : a_{i,k} = 1, (v_j,v_k) \in M_i \}|$.
Then $t_{i,j} \leq t'_{i,j}$.
Thus we have
\begin{align*}
    \Pr(t_{i,j} > \kappa_i) \leq \Pr(t_{i,j} \geq \kappa_i) \leq \Pr(t'_{i,j} \geq \kappa_i).
\end{align*}

Below we first prove (\ref{eq:AlgI_clip_bound}) and (\ref{eq:AlgII_clip_bound}). Then we prove (\ref{eq:AlgIII_clip_bound}).

\smallskip
\noindent{\textbf{Proof of (\ref{eq:AlgI_clip_bound}) and (\ref{eq:AlgII_clip_bound}).}}~~For each edge $(v_i,v_j)$, we have $\sum_{k \ne i,j} \textbf{1}_{(v_i,v_k) \in E} \leq \td_i$.
In \AlgOne{}, each edge $(v_j,v_k)$ is included in $E'$ with probability at most $\mu$, and all the events are independent.
In \AlgTwo{}, each of the edges $(v_i,v_k)$ and $(v_j,v_k)$ is included in $E'$ with probability at most $\mu$, and all the events are independent.
Thus, $\Pr(t'_{i,j} \geq \kappa_i)$ is less than or equal to the probability that the number of successes in the binomial distribution $B(\td_i, \mu^*)$ ($\mu^* = \mu$ in \AlgOne{} and $\mu^2$ in \AlgTwo{}) is larger than or equal to $\kappa_i$.

Let $X_{n,p}$ be a random variable representing the number of successes in the binomial distribution $B(n,p)$, and $F(\kappa_i;n,p) = \Pr(X_{n,p} \leq \kappa_i)$; i.e., $F$ is a cumulative distribution function of $B(n,p)$.
Since $\kappa_i \geq \mu^* \td_i$, we have
\begin{align*}
    &\Pr(t'_{i,j} \geq \kappa_i) \\
    &\leq \Pr(X_{\td_i, \mu^*} \geq \kappa_i) \\
    &= F(\td_i - \kappa_i; \td_i, 1-\mu^*) \\
    &\leq \exp \left[-\td_i D \left(\frac{\td_i - \kappa_i}{\td_i} \parallel 1-\mu^* \right) \right]  \text{(by Chernoff bound)}\\
    &= \exp \left[-\td_i D \left(\frac{\kappa_i}{\td_i} \parallel \mu^* \right) \right],
\end{align*}
which proves (\ref{eq:AlgI_clip_bound}) and (\ref{eq:AlgII_clip_bound}) (as $\Pr(t_{i,j} > \kappa_i) \leq \Pr(t'_{i,j} \geq \kappa_i)$).

\smallskip
\noindent{\textbf{Proof of (\ref{eq:AlgIII_clip_bound}).}}~~Assume that $\kappa_i \geq \mu^2 \td_i$ in \AlgThree{}.
For each edge $(v_i,v_j)$, we have $\sum_{k \ne i,j} \textbf{1}_{(i,k) \in E} \leq \td_i$.
In addition, each of the edges $(v_i,v_k)$ and
$(v_j,v_k)$ are included in $E'$ with probability at most $\mu$, and all the events are independent.

If $(v_i,v_j)$ is included in $E'$ (which happens with probability at most $\mu$), $\Pr(t'_{i,j} \geq \kappa_i)$ is less than or equal to the probability that the number of successes in the binomial distribution $B(\td_i, \mu^2)$ is larger than or equal to $\kappa_i$.
Otherwise (i.e., if $(v_i,v_j)$ is not included in $E'$), then $t'_{i,j} = 0$.

Thus, if $\kappa_i \geq \mu^2 \td_i$, we have
\begin{align*}
    &\Pr(t'_{i,j} \geq \kappa_i) \\
    &\leq \mu \Pr(X_{\td_i, \mu^2} \geq \kappa_i) \\
    &= \mu F(\td_i - \kappa_i; \td_i, 1-\mu^2) \\
    &\leq \mu \exp \left[-\td_i D \left(\frac{\td_i - \kappa_i}{\td_i} \parallel 1-\mu^2 \right) \right]  ~ \text{(by Chernoff bound)}\\
    &= \mu \exp \left[-\td_i D \left(\frac{\kappa_i}{\td_i} \parallel \mu^2 \right) \right],
\end{align*}
and therefore (\ref{eq:AlgIII_clip_bound}) holds (as $\Pr(t_{i,j} > \kappa_i) \leq \Pr(t'_{i,j} \geq \kappa_i)$).

If $\mu^3 \td_i \leq \kappa_i < \mu^2 \td_i$, (\ref{eq:AlgIII_clip_bound}) can be written as: $\Pr(t_{i,j} > \kappa_i) \leq \mu \exp \left[-\td_i D \left(\mu^2 \parallel \mu^2 \right) \right] = \mu$.
This clearly holds because each edge $(v_i,v_j)$ is included in $E'$ with probability at most $\mu$.
\qed

\section{Addendum to Double Clipping}
\label{sec:addendum}
\subsection{Sensitivity}
\label{sub:sensitivity}
After this paper was published at USENIX Security 2022\footnote{https://www.usenix.org/system/files/sec22-imola.pdf}, we found that we underestimated the sensitivity of our double clipping technique in Section~\ref{sec:double_clip}. 
For example, consider a graph $G=(V,E)$, where $V=\{v_1,v_2,v_3,v_4 \}$ and $E=\{(v_1,v_4), (v_2,v_4), (v_3,v_4) \}$. 
Assume that \AlgOne{} with double clipping is used as a triangle counting algorithm and that a message $M_4$ for $v_4$ is $M_4 = \{(v_1,v_3),(v_2,v_3)\}$. 
Then, $G$ has two noisy triangles $(v_1,v_3,v_4)$ and $(v_2,v_3,v_4)$ whose noisy edges are $(v_1,v_3)$ and $(v_2,v_3)$, respectively. 
Assume that $\kappa$ is set to $\kappa_4 = 1$. 

Recall that for $j$ such that $a_{i,j} = 1$ and $j<i$, we calculate $t_{i,j} = |\{(v_i,v_j,v_k) : a_{i,k} = 1, (v_j,v_k) \in M_i, j<k<i \}|$ (see line 5 in Algorithm~\ref{alg:clip}). 
In the above example, we have that $t_{4,1} = t_{4,2} = 1$ and $t_{4,3} = 0$. 
Thus, no triangle clipping occurs for $v_4$, and consequently, $t_4 = 2$ (see line 8 in Algorithm~\ref{alg:clip}). 

Now, consider a graph $G'=(V,E')$ obtained by removing $(v_3,v_4)$ from $G$, i.e., $E'=\{(v_1,v_4), (v_2,v_4) \}$. 
$G'$ has no noisy triangles, as $(v_3,v_4)$ is removed from two noisy triangles $(v_1,v_3,v_4)$ and $(v_2,v_3,v_4)$ in $G$. 
Thus, $t_{4,1} = t_{4,2} = t_{4,3} = 0$, and consequently, $t_4 = 0$. 
This means that the sensitivity of $t_4$ is $2$. 
Therefore, $\kappa_4$ ($=1$) is not an upper bound on the sensitivity of $t_4$. 
The root cause of this issue is that adding a single edge of user $v_i$ could increment $t_{i,j}$ for more than $\kappa_i$ choices of $v_j$ ($2$ choices in the above example), even if each $t_{i,j}$ is clipped to be at most $\kappa_i$. 

\subsection{($\epsilon,\delta$)-edge LDP}
The sensitivity issue in Appendix~\ref{sub:sensitivity} implies that adding $\Lap(\frac{\kappa_i}{\epsilon_2})$ (line 11 in Algorithm~\ref{alg:clip}) does not provide $\epsilon_2$-edge LDP. 
Consequently, our algorithms with double clipping do not satisfy $(\epsilon_0 + \epsilon_1 + \epsilon_2)$-edge LDP; i.e., Theorem~\ref{thm:privacy_DC} is incorrect. 

However, our algorithms with double clipping provide $(\epsilon,\delta)$-edge LDP and $(\epsilon,\delta)$-relationship DP, where $\epsilon = \epsilon_0 + \epsilon_1 + \epsilon_2$ and $\delta$ is immediately derived from Theorem~\ref{thm:triangle_excess} that upper bounds the triangle excess probability. 
Fortunately, our double clipping is very effective even if $\delta$ is extremely small (e.g., $\delta \ll n^{-1}$ or $n^{-2}$), as shown in Appendix~\ref{sub:clip_exp}. 

Formally, we define $(\epsilon,\delta)$-edge LDP and $(\epsilon,\delta)$-relationship DP as follows: 

\begin{definition} [$(\epsilon,\delta)$-edge LDP] \label{def:edge_LDP_delta} 
Let $\epsilon \in \nnreals$ and $\delta \in [0,1]$. 
For $i \in [n]$, let $\calR_i$ be a local randomizer of user $v_i$ that takes $\bma_i$ as input. We say $\calR_i$ provides
\emph{$(\epsilon,\delta)$-edge LDP} 
if for any two neighbor lists  $\bma_i, \bma'_i \in \{0,1\}^n$ that differ in one bit and any $S \subseteq \mathrm{Range}(\calR_i)$, 
\begin{align*}
\Pr[\calR_i(\bma_i) \in S] \leq e^\epsilon \Pr[\calR_i(\bma'_i) \in S] + \delta.
\end{align*}
\end{definition}

\begin{definition} [$(\epsilon,\delta)$-relationship DP] 
\label{def:entire_edge_LDP_delta} 
  Let $\epsilon \in \nnreals$ and $\delta \in [0,1]$. For 
  $i \in [n]$, 
  let $\calR_i$ be a 
  local randomizer of user $v_i$ that 
  takes $\bma_i$ as input. We say 
  $(\calR_1, \ldots, \calR_n)$ provides 
\emph{$(\epsilon,\delta)$-relationship DP} 
if for any two neighboring graphs $G, G' \in \calG$ that differ in one edge and 
  any $S \subseteq \mathrm{Range}(\calR_1) \times \ldots \times \mathrm{Range}(\calR_n)$, 
\begin{align*}
  &\Pr[(\calR_1(\bma_1), \ldots, \calR_n(\bma_n)) \in S] \nonumber\\
  &\leq e^\epsilon \Pr[(\calR_1(\bma'_1), \ldots, \calR_n(\bma'_n)) \in S] + \delta,
\end{align*}
  where $\bma_i$ (resp. $\bma_i'$) $\in \{0,1\}^n$ is the $i$-th row of the
  adjacency matrix of graph $G$ (resp. $G'$).
\end{definition}

\begin{proposition} [$(\epsilon,\delta)$-edge LDP and $(\epsilon,\delta)$-relationship DP] 
\label{prop:edge_LDP_entire_edge_LDP_delta} 
  If 
  each 
  of local randomizers $\calR_1, \ldots, \calR_n$ 
  provides 
  $(\epsilon,\delta)$-edge LDP, then $(\calR_1, \ldots, \calR_n)$ provides 
  $(2\epsilon,2\delta)$-relationship DP. 
  Additionally, if each $\calR_i$ uses only bits $a_{i,1}, \ldots, a_{i,i-1}$ for users with smaller IDs (i.e., only the lower triangular part of $\bmA$), then $(\calR_1, \ldots, \calR_n)$ provides 
  $(\epsilon,\delta)$-relationship DP. 
\end{proposition}
\begin{proof}
Immediately derived from group privacy \cite{DP}. 
\end{proof}

Then, our algorithms with double clipping provide the following privacy guarantees:
\begin{theorem}\label{thm:privacy_DC_delta}
  For $i \in [n]$, 
  let $\calR_i^1, \calR_i^2(M_i)$ be the randomizers used by user $v_i$ in
  rounds $1$ and $2$ of our algorithms with double clipping (Algorithms~\ref{alg:unify} and \ref{alg:clip}). 
  Assume that the clipping threshold $\kappa_i$ is set so that the upper-bound in Theorem~\ref{thm:triangle_excess} is smaller than or equal to $\beta \in [0,1]$; i.e., 
  \begin{align}
    \hspace{-1mm} \textstyle{\exp \left[-\td_i D \left(\frac{\kappa_i}{\td_i} \parallel \mu \right) \right]} &\leq \beta \label{eq:AlgI_delta} \\
    \hspace{-1mm} \textstyle{\exp \left[-\td_i D \left(\frac{\kappa_i}{\td_i} \parallel \mu^2 \right) \right]} &\leq \beta\label{eq:AlgII_delta}\\
    \hspace{-1mm} \textstyle{\mu \exp \left[-\td_i D \left(\frac{\max\{\kappa_i,\mu^2 \td_i\}}{\td_i} \parallel \mu^2 \right) \right]} &\leq \beta
    \label{eq:AlgIII_delta}
  \end{align}
  in \AlgOne{}, \AlgTwo{}, and \AlgThree{}, respectively. 
  Let $\calR_i(\bma_i) = (\calR_i^1(\bma_i), \calR_i^2(M_i)(\bma_i))$ 
  be the composition of the two randomizers. 
  Then,
  $\calR_i$ satisfies $(\epsilon, \delta)$-edge LDP, 
  and $(\calR_1,
  \ldots, \calR_n)$ satisfies $(\epsilon, \delta)$-relationship DP, where $\epsilon = \epsilon_0 + \epsilon_1 + \epsilon_2$ and $\delta = n \beta$. 
\end{theorem}

\begin{proof}
We first show that 
$\calR_i$ satisfies $(\epsilon, \delta)$-edge LDP, where $\epsilon = \epsilon_0 + \epsilon_1 + \epsilon_2$ and $\delta = n\beta$. 
For $S \subseteq \mathrm{Range}(\calR_i)$, let $S_0 \subseteq S$ be the set such that 
\begin{align}
\Pr[\calR_i(\bma_i) = s] \leq e^\epsilon \Pr[\calR_i(\bma'_i) = s]
\label{eq:calR_i_less_than_eps}
\end{align}
holds for any $s \in S_0$ and 
\begin{align*}
\Pr[\calR_i(\bma_i) = s] > e^\epsilon \Pr[\calR_i(\bma'_i) = s]
\end{align*}
holds for any $s \in S \setminus S_0$. 
Then, we have
\begin{align*}
\Pr[\calR_i(\bma_i) \in S] 
&= \Pr[\calR_i(\bma_i) \in S_0] + \Pr[\calR_i(\bma_i) \in S \setminus S_0] \\
&\leq e^\epsilon  \Pr[\calR_i(\bma'_i) \in S_0] + \delta,
\end{align*}
where
\begin{align*}
\delta = \Pr[\calR_i(\bma_i) \in S \setminus S_0].
\end{align*}
By (\ref{eq:AlgI_delta}), (\ref{eq:AlgII_delta}), (\ref{eq:AlgIII_delta}), and Theorem~\ref{thm:triangle_excess}, we always have 
\begin{align*}
\Pr(t_{i,j} > \kappa_i) \leq \beta.
\end{align*}
When $t_{i,1}, \cdots, t_{i,i-1} \leq \kappa_i$, 
$\kappa_i$ upper bounds the sensitivity of $t_i$, 
hence (\ref{eq:calR_i_less_than_eps}) holds (see Appendix~\ref{sub:proof_privacy_DC} for the proof in this case). 
Thus, $\delta$ can be written as follows:
\begin{align*}
\delta 
\leq 1 - \Pr(t_{i,1}, \cdots, t_{i,i-1} \leq \kappa_i) 
\leq \sum_{j=1}^{i-1} \Pr(t_{i,j} > \kappa_i) 
\leq n\beta. 
\end{align*}
Therefore, $\calR_i$ satisfies $(\epsilon, \delta)$-edge LDP, where $\epsilon = \epsilon_0 + \epsilon_1 + \epsilon_2$ and $\delta = n\beta$. 

Next, we show that $(\calR_1,
\ldots, \calR_n)$ satisfies $(\epsilon, \delta)$-relationship DP. 
This is immediately derived from Proposition~\ref{prop:edge_LDP_entire_edge_LDP_delta} and the fact that each $\calR_i$ uses only the lower triangular part of $\bmA$.
\end{proof}

\smallskip
\noindent{\textbf{Remark.}}~~In line 8 of Algorithm~\ref{alg:clip}, we calculate the total number $t_i$ of noisy triangles for user $v_i$ 
by clipping the number of noisy triangles; i.e., 
$t_i = \sum_{a_{i,j}=1,j<i} \min\{t_{i,j}, \kappa_i\}$. 
Theorem~\ref{thm:privacy_DC_delta} holds even when we do not clip the number of noisy triangles and calculate $t_i$ by simply summing up $t_{i,j}$; i.e., $t_i = \sum_{a_{i,j}=1,j<i} t_{i,j}$. 
In this case, we perform calculate $\td_i$ by edge clipping 
and then set $\kappa_i$ so that the triangle excess probability is smaller than or equal to $\beta$; 
i.e., (\ref{eq:AlgI_delta}), (\ref{eq:AlgII_delta}), and (\ref{eq:AlgIII_delta}) hold. 
Then, the failure probability for user $v_i$ is upper bounded by $\delta = n\beta$, and consequently, our algorithms provide $(\epsilon, \delta)$-edge LDP and $(\epsilon, \delta)$-relationship DP. 
When $\delta = n \beta$ is extremely small (e.g., $\delta \ll n^{-1}$ or $n^{-2}$), triangle removal rarely occurs. 
Thus, whether or not we clip the number of noisy triangles in line 8 of Algorithm~\ref{alg:clip} makes little difference to the utility. 
Therefore, in Appendix~\ref{sub:clip_exp}, we evaluate our algorithms without modifying line 8 of Algorithm~\ref{alg:clip}.

\subsection{Experimental Results}
\label{sub:clip_exp}

We evaluated the relationship between the relative error and $\delta$ $(=n \beta)$ in our algorithms with double clipping. 
We set the triangle excess probability $\beta$ to various values from $10^{-6}$ to $10^{-36}$. 
The other settings are the same as Section~\ref{sec:experiments}.  

Figure~\ref{fig:resC_delta} shows the results when $\epsilon=1$, $\mu^*=10^{-3}$, and $\beta = 10^{-6}$, $10^{-12}$, $10^{-18}$, $10^{-24}$, $10^{-30}$, or $10^{-36}$. 
We observe that our double clipping dramatically reduces the relative error even when $\delta$ is extremely small. 
It is well known that $\delta$ should be much smaller than $n^{-1}$ in tabular data \cite{DP}. 
In graph data, an ideal value of $\delta$ might be much smaller than $n^{-2}$, as a graph has at most $\binom{n}{2}$ edges. 
This is achieved by, for example, setting $\beta = 10^{-24}$ (in this case, $\delta = 1.08 \times 10^{-19} \ll n^{-2}$ in \GPlus{} and $\delta = 8.96 \times 10^{-19} \ll n^{-2}$ in \IMDB{}). 
Figure~\ref{fig:resC_delta} shows that our double clipping is very effective even when we set $\beta = 10^{-24}$ so that $\delta \ll n^{-2}$. 

\begin{figure}[t]
  \centering
  \includegraphics[width=0.96\linewidth]{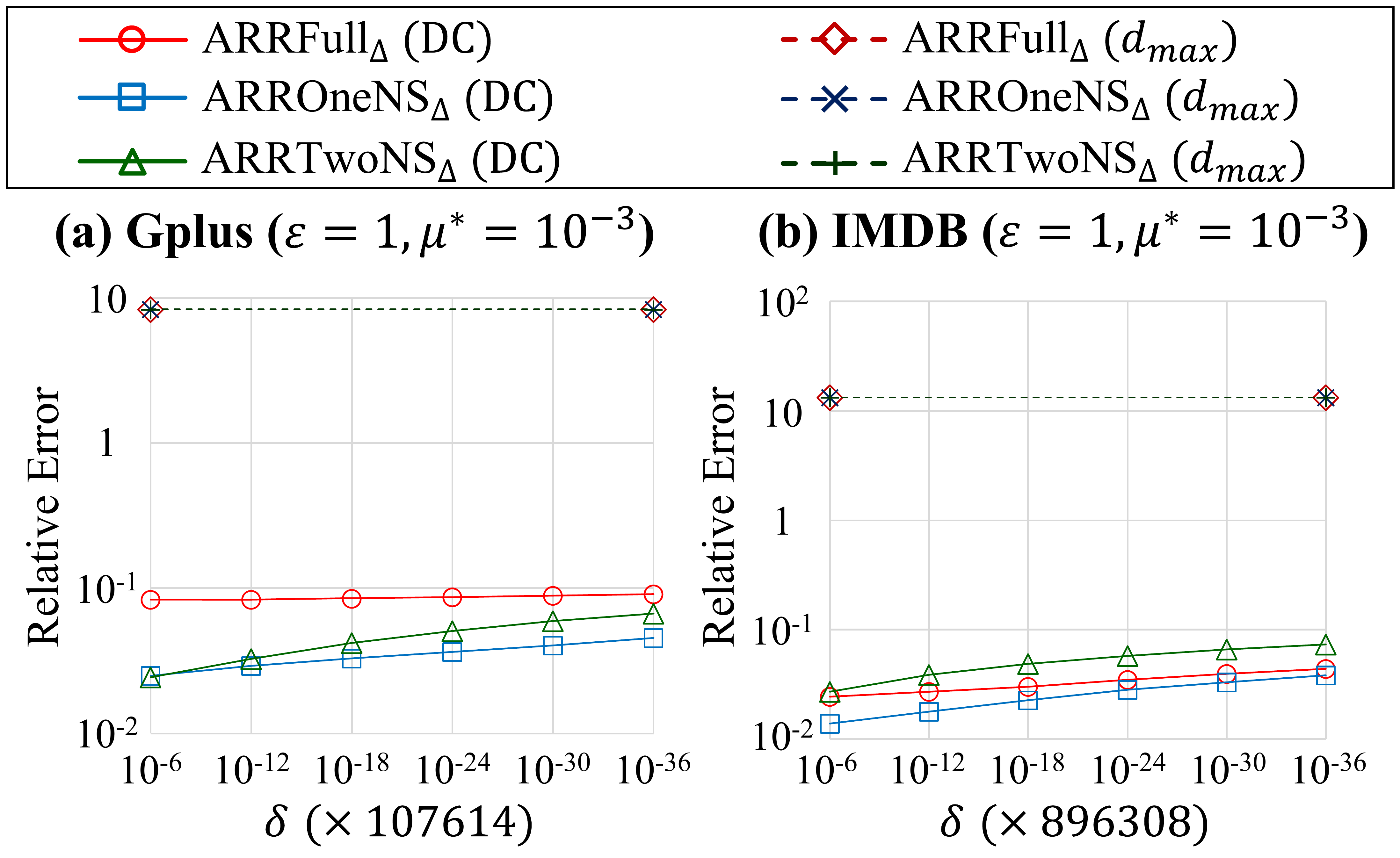}
  \vspace{-4mm}
  \caption{Relationship between the relative error and $\delta$ $(=n \beta)$ in our three algorithms with double clipping (``DC'') ($\epsilon=1$, $\mu^*=10^{-3}$, $\beta = 10^{-6}, \cdots, 10^{-30}$, or $10^{-36}$, $n=107614$ in \GPlus{}, $n=896308$ in \IMDB{}). 
  Note that $\delta = 0$ in our three algorithms without double clipping (``$d_{max}$'').
  } 
  \label{fig:resC_delta}
\end{figure}

\begin{figure}[t]
  \centering
  \includegraphics[width=0.99\linewidth]{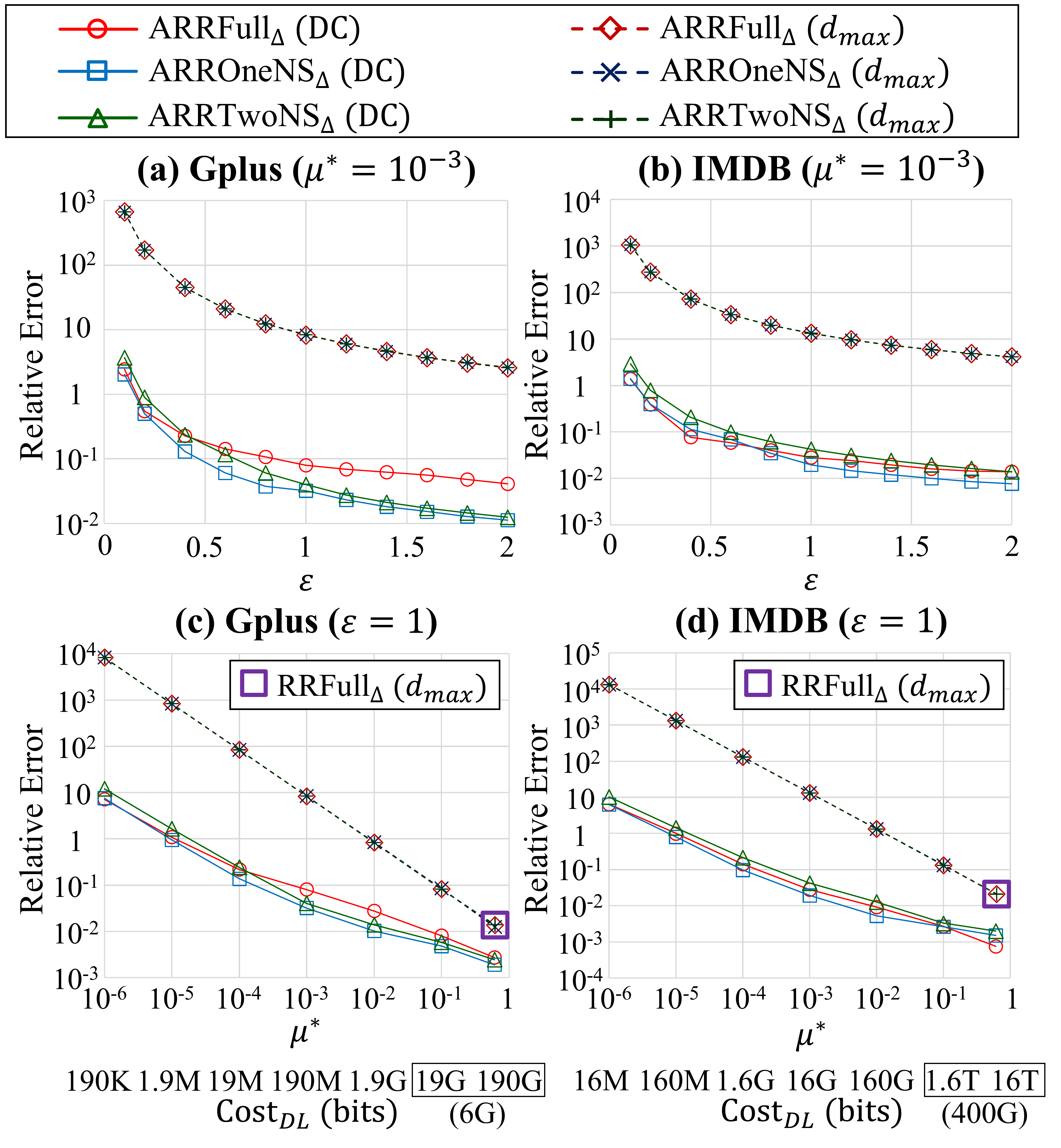}
  \vspace{-4mm}
  \caption{Relative error of our three algorithms with (``DC'') or without (``$d_{max}$'') double clipping when $\beta = 10^{-14}$ ($n=107614$ and 
  $\delta = 1.08 \times 10^{-9} \ll n^{-1}$ in \GPlus{}; 
  $n=896308$ and $\delta = 8.96 \times 10^{-9} \ll n^{-1}$ in \IMDB{}). 
  \AlgSec{} is the 
  algorithm in~\cite{Imola_USENIX21}. 
  $\CostDL$ is 
  an upper-bound in 
  (\ref{eq:CostDL_F}). 
  } 
  \label{fig:resC_middle_delta}
\end{figure}

\begin{figure}[t]
  \centering
  \includegraphics[width=0.99\linewidth]{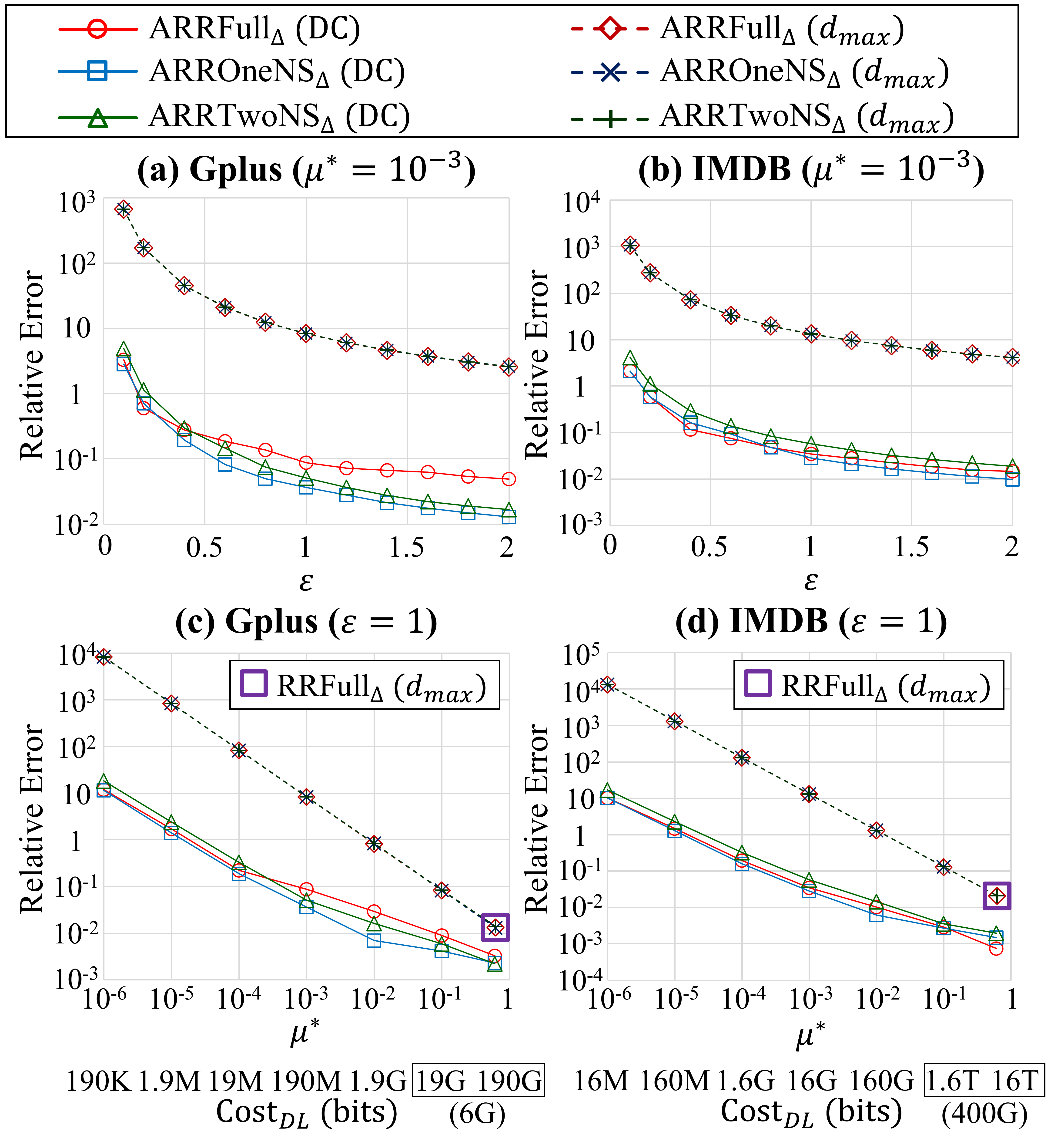}
  \vspace{-4mm}
  \caption{Relative error of our three algorithms with (``DC'') or without (``$d_{max}$'') double clipping when $\beta = 10^{-24}$ ($n=107614$ and 
  $\delta = 1.08 \times 10^{-19} \ll n^{-2}$ in \GPlus{}; 
  $n=896308$ and $\delta = 8.96 \times 10^{-19} \ll n^{-2}$ in \IMDB{}). 
  \AlgSec{} is the 
  algorithm in~\cite{Imola_USENIX21}. 
  $\CostDL$ is 
  an upper-bound in 
  (\ref{eq:CostDL_F}). 
  } 
  \label{fig:resC_small_delta}
\end{figure}

Figures~\ref{fig:resC_middle_delta} and \ref{fig:resC_small_delta} show the relative error when we set 
$\beta = 10^{-14}$ (i.e., $\delta \ll n^{-1}$) and 
$\beta = 10^{-24}$ (i.e., $\delta \ll n^{-2}$), respectively. 
These figures show very similar tendencies to Figure~\ref{fig:res2_w_Lap}. 
For example, 
our \AlgTwo{} (DC) with $\beta = 10^{-14}$ (resp.~$10^{-24}$) reduces the communication cost in \GPlus{} from $6$ Gbits (300 seconds when $20$ Mbps) to $19$ Mbits (0.95 seconds) or less while keeping relative error $= 0.14$ (resp.~$0.21$). 
For \IMDB{}, our \AlgTwo{} (DC) with $\beta = 10^{-14}$ (resp.~$10^{-24}$) reduces the communication cost from $400$ Gbits (6 hours) to $1.6$ Gbits (80 seconds) or less while keeping relative error $= 0.097$ (resp.~$0.16$). 

In summary, our algorithms with double clipping provide $(\epsilon,\delta)$-edge LDP and $(\epsilon,\delta)$-relationship DP, and dramatically reduce the communication cost of the algorithm in \cite{Imola_USENIX21} even when $\delta \ll n^{-1}$ or $n^{-2}$.
\newpage
\hspace{5mm}
}

\end{document}